\documentclass[a4paper,USenglish,cleveref, autoref, thm-restate]{lipics-v2021}
\usepackage[T1]{fontenc}
\usepackage[utf8]{inputenc}
\usepackage{numprint}

\nolinenumbers

\usepackage{amsmath}
\usepackage{stackrel}
\usepackage{amsthm}
\usepackage{mathrsfs}
\usepackage{amssymb}
\usepackage{amsfonts}
\usepackage{wasysym}
\usepackage{tikz}
\usepackage[ruled, vlined, linesnumbered]{algorithm2e}
\usepackage{stmaryrd}
\usepackage{enumitem}
\usetikzlibrary{automata}
\usetikzlibrary{shapes}
\usepackage{hyperref}

\newcommand{\<}{\langle}
\renewcommand{\>}{\rangle}

\newcommand{\bsigma}{\bar{\sigma}}
\newcommand{\btau}{\bar{\tau}}

\renewcommand{\|}{\upharpoonright}
\newcommand{\hpi}{\hat{\pi}}
\newcommand{\hmu}{\hat{\mu}}

\newcommand{\N}{\mathbb{N}}
\newcommand{\Z}{\mathbb{Z}}
\newcommand{\R}{\mathbb{R}}
\newcommand{\Q}{\mathbb{Q}}
\newcommand{\bR}{\overline{\mathbb{R}}}

\renewcommand{\t}{\!^{\mathrm{t}}}

\newcommand{\playcircle}{{\scriptsize{\Circle}}}

\renewcommand{\l}{\ell}
\renewcommand{\epsilon}{\varepsilon}
\renewcommand{\phi}{\varphi}

\newcommand{\opt}{\mathrm{opt}}

\newcommand{\card}{\mathrm{card}}
\newcommand{\SC}{\mathrm{SC}}
\newcommand{\Plays}{\mathrm{Plays}}
\newcommand{\Hist}{\mathrm{Hist}}

\newcommand{\Conv}{\mathrm{Conv}}

\newcommand{\SConn}{\mathrm{SConn}}

\newcommand{\nego}{\mathrm{nego}}
\newcommand{\lCons}{\lambda\mathrm{Cons}}
\newcommand{\lRat}{\lambda\mathrm{Rat}}

\newcommand{\Prop}{\mathrm{Prop}}
\newcommand{\Acc}{\mathrm{Acc}}
\newcommand{\Dev}{\mathrm{Dev}}

\newcommand{\ML}{\mathrm{ML}}
\newcommand{\minstar}{\min\!^\star}
\newcommand{\Rat}{\mathrm{Rat}}
\newcommand{\Abs}{\mathrm{Abs}}
\newcommand{\Conc}{\mathrm{Conc}}

\newcommand{\val}{\mathrm{val}}

\renewcommand{\P}{\mathbb{P}}
\newcommand{\C}{\mathbb{C}}

\renewcommand{\H}{\mathcal{H}}
\newcommand{\A}{\mathcal{A}}

\newcommand{\bx}{\bar{x}}
\newcommand{\by}{\bar{y}}

\newcommand{\balpha}{\bar{\alpha}}
\newcommand{\bbeta}{\bar{\beta}}
\newcommand{\bgamma}{\bar{\gamma}}

\newcommand{\blambda}{\lambda \mspace{-10mu} \bar{\phantom{v}}}
\newcommand{\bdelta}{\delta \mspace{-9mu} \bar{\phantom{v}}}

\newcommand{\dH}{\dot{H}}

\newcommand{\dc}{\dot{c}}
\newcommand{\drho}{\dot{\rho}}
\newcommand{\deta}{\dot{\eta}}

\theoremstyle{plain}
\newtheorem{thm}{Theorem}
\newtheorem{pptn}{Proposition}
\newtheorem{lm}{Lemma}
\newtheorem{cor}{Corollary}

\theoremstyle{definition}
\newtheorem{defi}{Definition}

\theoremstyle{remark}
\newtheorem*{rek}{Remark}
\newtheorem*{reks}{Remarks}

\newtheorem{ex}{Example}

\title{Subgame-perfect Equilibria in Mean-payoff Games}

\author{Léonard Brice}
{Université Gustave Eiffel, France}{lnrd.brice@gmail.com}{}{}

\author{Jean-François Raskin}
{Université Libre de Bruxelles, Belgium}{jraskin@ulb.ac.be}{}{}

\author{Marie van den Bogaard}
{Université Gustave Eiffel, France}{marie.van-den-bogaard@univ-eiffel.fr}{}{}

\authorrunning{L. Brice, J.-F. Raskin, and M. van den Bogaard}

\Copyright{Léonard Brice, Jean-François Raskin, Marie van den Bogaard}



\begin{document}

\maketitle

\begin{abstract}
In this paper, we provide an effective characterization of all the subgame-perfect equilibria in infinite duration games played on finite graphs with mean-payoff objectives.
To this end, we introduce the notion of requirement, and the notion of negotiation function.
We establish that the plays that are supported by SPEs are exactly those that are consistent with the least fixed point of the negotiation function.
Finally, we show that the negotiation function is piecewise linear, and can be analyzed using the linear algebraic tool box.
As a corollary, we prove the decidability of the SPE constrained existence problem, whose status was left open in the literature.
\end{abstract}

\section{Introduction}
The notion of Nash equilibrium (NE) is one of the most important and most studied solution concepts in game theory.
A profile of strategies is an NE when no rational player has an incentive to change their strategy unilaterally, i.e. while the other players keep their strategies. Thus an NE models a stable situation. 
Unfortunately, it is well known that, in sequential games, NEs suffer from the problem of {\em non-credible threats}, see e.g.~\cite{Osborne}. In those games, some NE only exists when some players do {\em not} play rationally in subgames and so use non-credible threats to force the NE.
This is why, in sequential games, the stronger notion of {\em subgame-perfect equilibrium} is used instead: a profile of strategies is a subgame-perfect equilibrium (SPE) if it is an NE in all the subgames of the sequential game. Thus SPE imposes rationality even after a deviation has occured.

In this paper, we study sequential games that are infinite-duration games played on graphs with mean-payoff objectives, and focus on SPEs. While NEs are guaranteed to exist in infinite duration games played on graphs with mean-payoff objectives, it is known that it is not the case for SPEs, see e.g.~\cite{solan2003deterministic,DBLP:conf/csl/BrihayeBMR15}. We provide in this paper a constructive characterization of the entire set of SPEs, which allows us to decide, among others, the SPE (constrained) existence problem. This problem was left open in previous contributions on the subject. More precisely, our contributions are described in the next paragraphs.


\noindent \textbf{Contributions.~} First, we introduce two important new notions that allow us to capture NEs, and more importantly SPEs, in infinite duration games played on graphs with mean-payoff objectives\footnote{A large part of our results apply to the larger class of games with prefix-independent objectives. For the sake of readability of this introduction, we focus here on mean-payoff games but the technical results in the paper are usually covering broader classes of games.}: the notion of {\em requirement} and the notion of {\em negotiation function}. 

A requirement $\lambda$ is a function that assigns to each vertex $v \in V$ of a game graph a value in $\mathbb{R} \cup \{ -\infty, +\infty\}$. The value $\lambda(v)$ represents a requirement on any play $\rho= \rho_0 \rho_1 \dots \rho_n \dots$ that traverses this vertex: if we want the player who controls the vertex $v$ to follow $\rho$ and to give up deviating from $\rho$, then the play must offer a payoff to this player that is at least $\lambda(v)$. An infinite play $\rho$ is $\lambda$-consistent if, for each player~$i$, the payoff of $\rho$ for player~$i$ is larger than or equal to the largest value of $\lambda$ on vertices occurring along $\rho$ and controlled by player~$i$.

We first use those notions to rephrase a classical result about NEs: if $\lambda$ maps a vertex $v$ to the largest value that the player that controls $v$ can secure against a fully adversarial coalition of the other players, i.e. if $\lambda(v)$ is the zero-sum worst-case value, then the set of plays that are $\lambda$-consistent is exactly the set of plays that are supported by an NE (Theorem~\ref{thm_ne}). 

As SPEs are forcing players to play rationally in all subgames, we cannot rely on the zero-sum worst-case value to characterize them. Indeed, when considering the worst-case value, we allow adversaries to play fully adversarially after a deviation and so potentially in an irrational way w.r.t. their own objective.  In fact, in an SPE, a player is refrained to deviate when opposed by a coalition of {\em rational adversaries}.  To characterize this relaxation of the notion of worst-case value, we rely on our notion of \emph{negotiation function}.

The negotiation function $\nego$ operates from the set of requirements into itself. To understand the purpose of the negotiation function, let us consider its application on the requirement $\lambda$ that maps every vertex $v$ on the worst-case value as above. Now, we can naturally formulate the following question: given $v$ and $\lambda$, can the player who controls $v$ improve the value that they can ensure against all the other players, if only plays that are consistent with $\lambda$ are proposed by the other players? In other words, can this player enforce a better value when playing against the other players if those players are not willing to give away their own worst-case value? Clearly, securing this worst-case value can be seen as a minimal goal for any {\em rational} adversary. So $\nego(\lambda)(v)$ returns this value; and this reasoning can be iterated. One of the contributions of this paper is to show that the least fixed point $\lambda^*$ of the negotiation function is exactly characterizing the set of plays supported by SPEs (Theorem~\ref{thm_spe}).

To turn this fixed point characterization of SPEs into algorithms, we additionally draw links between the negotiation function and two classes of zero-sum games, that are called {\em abstract} and {\em concrete} negotiation games (see Theorem~\ref{thm_concrete_game}).  We show that the latter can be solved effectively and allow, given $\lambda$, to compute $\nego(\lambda)$ (Lemma~\ref{lm_resolution_concrete_game}). While solving concrete negotiation games allows us to compute $\nego(\lambda)$ for any requirement $\lambda$, and even if the function $\nego(\cdot)$ is monotone and Scott-continuous, a direct application of the Kleene-Tarski fixed point theorem is not sufficient to obtain an effective algorithm to compute $\lambda^*$. Indeed, we give examples that require a transfinite number of iterations to converge to the least fixed point. To provide an algorithm to compute $\lambda^*$, we show that the function $\nego(\cdot)$ is piecewise linear and we provide an effective representation of this function (Theorem~\ref{thm_formula_nego}). This effective representation can then be used to extract all its fixed points and in particular its least fixed point using linear algebraic techniques, hence the decidability of the SPE (constrained) existence problem (Theorem~\ref{thm_decidable}). Finally, all our results are also shown to extend to $\epsilon$-SPEs, those are quantitative relaxations of SPEs.

\noindent \textbf{Related works.~} Non-zero sum infinite duration games have attracted a large attention in recent years, with applications targeting reactive synthesis problems. We refer the interested reader to the following survey papers~\cite{DBLP:conf/lata/BrenguierCHPRRS16,DBLP:conf/dlt/Bruyere17} and their references for the relevant literature. We detail below contributions more closely related to the work presented here.

In~\cite{DBLP:conf/lfcs/BrihayePS13}, Brihaye et al. offer a characterization of NEs in quantitative games for cost-prefix-linear reward functions based on the worst-case value. The mean-payoff is cost-prefix-linear. In their paper, the authors do not consider the stronger notion of SPE, which is the central solution concept studied in our paper.
In \cite{DBLP:conf/csl/BruyereMR14}, Bruy{\`{e}}re et al. study secure equilibria that are a refinement of NEs. Secure equilibria are not subgame-perfect and are, as classical NEs, subject to non-credible threats in sequential games.

In~\cite{DBLP:conf/fsttcs/Ummels06}, Ummels proves that there always exists an SPE in games with $\omega$-regular objectives and defines algorithms based on tree automata to decide constrained SPE problems. Strategy logics, see e.g.~\cite{DBLP:journals/iandc/ChatterjeeHP10}, can be used to encode the concept of SPE in the case of $\omega$-regular objectives with application to the rational synthesis problem~\cite{DBLP:journals/amai/KupfermanPV16} for instance.  
In~\cite{DBLP:journals/mor/FleschKMSSV10}, Flesch et al. show that the existence of $\epsilon$-SPEs is guaranteed when the reward function is {\em lower-semicontinuous}. The mean-payoff reward function is neither $\omega$-regular, nor lower-semicontinuous, and so the techniques defined in the papers cited above cannot be used in our setting. Furthermore, as already recalled above, see e.g.~\cite{vieille:hal-00464953,DBLP:conf/csl/BrihayeBMR15}, contrary to the $\omega$-regular case, SPEs in games with mean-payoff objectives may fail to exist.

In~\cite{DBLP:conf/csl/BrihayeBMR15}, Brihaye et al. introduce and study the notion of weak subgame-perfect equilibria, which is a weakening of the classical notion of SPE. This weakening is equivalent to the original SPE concept on reward functions that are {\em continuous}. This is the case for example for the quantitative reachability reward function, on which Brihaye et al. solve the problem of the constrained existence of SPEs in \cite{DBLP:conf/concur/BrihayeBGRB19}. On the contrary, the mean-payoff cost function is not continuous and the techniques used in~\cite{DBLP:conf/csl/BrihayeBMR15}, and generalized in~\cite{DBLP:conf/fossacs/Bruyere0PR17}, cannot be used to characterize SPEs for the mean-payoff reward function.


In~\cite{thesis_noemie}, Meunier develops a method based on Prover-Challenger games to solve the problem of the existence of SPEs on games with a finite number of possible outcomes. This method is not applicable to the mean-payoff reward function, as the number of outcomes in this case is uncountably infinite.

In~\cite{DBLP:journals/mor/FleschP17}, Flesch and Predtetchinski present another characterization of SPEs on games with finitely many possible outcomes, based on a game structure that we will present here under the name of \emph{abstract negotiation game}.
Our contributions differ from this paper in two fundamental aspects. First, it lifts the restriction to finitely many possible outcomes. This is crucial as mean-payoff games violate this restriction. Instead, we identify a class of games, that we call {\em with steady negotiation}, that encompasses mean-payoff games and for which some of the conceptual tools introduced in that paper can be generalized. Second, the procedure developed by Flesch and Predtetchinski is {\em not} an algorithm in CS acceptation: it needs to solve infinitely many games that are not represented effectively, and furthermore it needs a transfinite number of iterations. On the contrary, our procedure is effective and leads to a complete algorithm in the classical sense: with guarantee of termination in finite time and applied on effective representations of games.


\noindent \textbf{Structure of the paper.~}
In Sect.~2, we introduce the necessary background.
Sect.~3 defines the notion of requirement and the negotiation function.
Sect.~4 shows that the set of plays that are supported by an SPE are those that are $\lambda^*$-consistent, where $\lambda^*$ is the least fixed point of the negotiation function.
Sect.~5 draws a link between the negotiation function and negotiation games.
Sect.~6 establishes that the  negotiation function is effectively piecewise linear.
Finally, Sect.~7 applies those results to prove the decidability of the SPE constrained existence problem on mean-payoff games, and adds some complexity considerations.
All the detailed proofs of our results can be found in a well-identified appendix and a large number of examples are provided in the main part of the paper to illustrate the main ideas behind our new concepts and constructions.


	\section{Background} \label{sec_background}

In all what follows, we will use the word \emph{game} for the infinite duration turn-based quantitative games on finite graphs with complete information.

\begin{defi}[Game]\label{defi_game}
	A \emph{game} is a tuple
	$G =\left(\Pi, V, (V_i)_{i \in \Pi}, E, \mu\right)$, where:
	
	\begin{itemize}
		\item $\Pi$ is a finite set of \emph{players};
		
		\item $(V, E)$ is a finite directed graph, whose vertices are sometimes called \emph{states} and whose edges are sometimes called \emph{transitions}, and in which every state has at least one outgoing transition.
		For the simplicity of writing, a transition $(v, w) \in E$ will often be written $vw$.
		
		\item $(V_i)_{i \in \Pi}$ is a partition of $V$, in which $V_i$ is the set of states \emph{controlled} by player $i$;
		
		\item $\mu: V^\omega \to \R^\Pi$ is an \emph{outcome function}, that maps each infinite word $\rho$ to the tuple $\mu(\rho) = (\mu_i(\rho))_{i \in \Pi}$ of the players' \emph{payoffs}.
	\end{itemize}
\end{defi}

\begin{defi}[Initialized game]
	An \emph{initialized game} is a tuple $(G, v_0)$, often written $G_{\|v_0}$, where $G$ is a game and $v_0 \in V$ is a state called \emph{initial state}. Moreover, the game $G_{\|v_0}$ is \emph{well-initialized} if any state of $G$ is accessible from $v_0$ in the graph $(V, E)$.
\end{defi}

\begin{defi}[Play, history]
	A \emph{play} (resp. history) in the game $G$ is an infinite (resp. finite) path in the graph $(V, E)$.
	It is also a play (resp. history) in the initialized game $G_{\|v_0}$, where $v_0$ is its first vertex.
	The set of plays (resp. histories) in the game $G$ (resp. the initialized game $G_{\|v_0}$) is denoted by $\Plays G$ (resp. $\Plays G_{\|v_0}, \Hist G, \Hist G_{\|v_0}$).
	We write $\Hist_i G$ (resp. $\Hist_i G_{\|v_0}$) for the set of histories in $G$ (resp. $G_{\|v_0}$) of the form $hv$, where $v$ is a vertex controlled by player $i$.
\end{defi}

\begin{rek}
    In the literature, the word \emph{outcome} can be used to name plays, and the word \emph{payoff} to name what we call here outcome. Here, the word \emph{payoff} will be used to refer to outcomes, seen from the point of view of a given player -- or in other words, an \emph{outcome} will be seen as the collection of all players' payoffs.
\end{rek}


\begin{defi}[Strategy, strategy profile]
	A \emph{strategy} for player $i$ in the initialized game $G_{\|v_0}$ is a function $\sigma_i: \Hist_i G_{\|v_0} \to V$, such that $v\sigma_i(hv)$ is an edge of $(V, E)$ for every $hv$.
	A history $h$ is \emph{compatible} with a strategy $\sigma_i$ if and only if $h_{k+1} = \sigma_i(h_0 \dots h_k)$ for all $k$ such that $h_k \in V_i$. A play $\rho$ is compatible with $\sigma_i$ if all its prefixes are.

	A \emph{strategy profile} for $P \subseteq \Pi$ is a tuple $\bsigma_P = (\sigma_i)_{i \in P}$, where for each $i$, $\sigma_i$ is a strategy for player $i$ in $G_{\|v_0}$.
	A \emph{complete} strategy profile, usually written $\bsigma$, is a strategy profile for $\Pi$.
    A play or a history is \emph{compatible} with $\bsigma_P$ if it is compatible with every $\sigma_i$ for $i \in P$.
    
    When $i$ is a player and when the context is clear, we will often write $-i$ for the set $\Pi \setminus \{i\}$.
    We will often refer to $\Pi \setminus \{i\}$ as the \emph{environment} against player $i$.
    When $\btau_P$ and $\btau'_Q$ are two strategy profiles with $P \cap Q = \emptyset$, $(\btau_P, \btau'_Q)$ denotes the strategy profile $\bsigma_{P \cup Q}$ such that $\sigma_i = \tau_i$ for $i \in P$, and $\sigma_i = \tau'_i$ for $i \in Q$.
\end{defi}


Before moving on to SPEs, let us recall the notion of Nash equilibrium.

\begin{defi}[Nash equilibrium]
	Let $G_{\|v_0}$ be an initialized game. The strategy profile $\bsigma$ is a \emph{Nash equilibrium} --- or \emph{NE} for short --- in $G_{\|v_0}$ if and only if for each player $i$ and for every strategy $\sigma'_i$, called \emph{deviation of $\sigma_i$}, we have the inequality $\mu_i\left(\< \sigma'_i, \bsigma_{-i} \>_{v_0}\right) \leq \mu_i\left(\< \bsigma \>_{v_0}\right)$.
\end{defi}

To define SPEs, we need the notion of subgame.

\begin{defi}[Subgame, substrategy]
	Let $hv$ be a history in the game $G$. The \emph{subgame} of $G$ after $hv$ is the initialized game $\left(\Pi, V, (V_i)_i, E, \mu_{\|hv}\right)_{\|v}$, where $\mu_{\|hv}$ maps each play to its payoff in $G$, assuming that the history $hv$ has already been played: formally, for every $\rho \in \Plays G_{\|hv}$, we have $\mu_{\|hv}(\rho) = \mu(h\rho)$.
    
    If $\sigma_i$ is a strategy in $G_{\|v_0}$, its \emph{substrategy} after $hv$ is the strategy $\sigma_{i\|hv}$ in $G_{\|hv}$, defined by $\sigma_{i\|hv}(h') = \sigma_i(hh')$ for every $h' \in \Hist_i G_{\|hv}$.
\end{defi}

\begin{rek}
The initialized game $G_{\|v_0}$ is also the subgame of $G$ after the one-state history $v_0$.
\end{rek}

\begin{defi}[Subgame-perfect equilibrium]
	Let $G_{\|v_0}$ be an initialized game. The strategy profile $\bsigma$ is a \emph{subgame-perfect equilibrium} --- or \emph{SPE} for short --- in $G_{\|v_0}$ if and only if for every history $h$ in $G_{\|v_0}$, the strategy profile $\bsigma_{\|h}$ is a Nash equilibrium in the subgame $G_{\|h}$.
\end{defi}

The notion of subgame-perfect equilibrium can be seen as a refinement of Nash equilibrium: it is a stronger equilibrium which excludes players resorting to non-credible threats.


\begin{ex}
In the game represented in Figure~\ref{fig_ne_spe}, where the square state is controlled by player $\Box$ and the round states by player $\Circle$, if both players get the payoff $1$ by reaching the state $d$ and the payoff $0$ in the other cases, there are actually two NEs: one, in blue, where $\Box$ goes to the state $b$ and then player $\Circle$ goes to $d$, and both win, and one, in red, where player $\Box$ goes to the state $c$ because player $\Circle$ was planning to go to $e$. However, only the blue one is an SPE, as moving from $b$ to $e$ is irrational for player $\Circle$ in the subgame $G_{\|ab}$.
\end{ex}

An $\epsilon$-SPE is a strategy profile which is \emph{almost} an SPE: if a player deviates after some history, they will not be able to improve their payoff by more than a quantity $\epsilon \geq 0$.

\begin{defi}[$\epsilon$-SPE]
	Let $G_{\|v_0}$ be an initialized game, and $\epsilon \geq 0$. A strategy profile $\bsigma$ from $v_0$ is an $\epsilon$-SPE if and only if for every history $hv$, for every player $i$ and every strategy $\sigma'_i$, we have $\mu_i(\< \bsigma_{-i\|hv}, \sigma'_{i\|hv} \>_v) \leq \mu_i(\< \bsigma_{\|hv} \>_v) + \epsilon$.
\end{defi}

Note that a $0$-SPE is an SPE, and conversely.

Hereafter, we focus on \emph{prefix-independent} games, and in particular \emph{mean-payoff} games.

	\begin{defi}[Mean-payoff game]
		A \emph{mean-payoff game} is a game $G = \left(\Pi, V, (V_i)_i, E, \mu \right)$, where $\mu$ is defined from a function $\pi: E \to \Q^\Pi$, called \emph{weight function}, by, for each player $i$:
		$$\mu_i: \rho \mapsto \underset{n \to \infty}{\liminf} \frac{1}{n} \underset{k=0}{\overset{n-1}{\sum}} \pi_i\left(\rho_k\rho_{k+1}\right).$$
	\end{defi}
	
In a mean-payoff game, the weight given by the function $\pi$ represents the immediate reward that each action gives to each player. The final payoff of each player is their average payoff along the play, classically defined as the limit inferior over $n$ (since the limit may not be defined) of the average payoff after $n$ steps.

\begin{defi}[Prefix-independent game]
	A game $G$ is \emph{prefix-independent} if, for every history $h$ and for every play $\rho$, we have $\mu(h\rho) = \mu(\rho)$.
	We also say, in that case, that the outcome function $\mu$ is prefix-independent.
\end{defi}

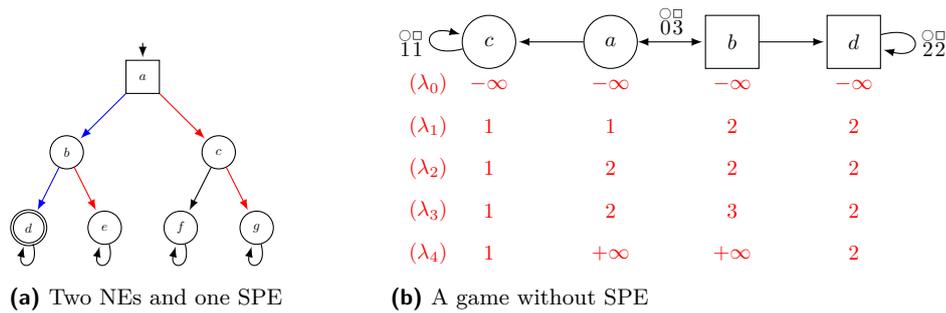
\begin{figure}
    \centering
    \begin{subfigure}[b]{0.35\textwidth}
    	\begin{tikzpicture}[->,>=latex,shorten >=1pt, initial text={}, scale=0.5, every node/.style={scale=0.5}]
    	\node[initial above, state, rectangle] (a) at (0, 0) {$a$};
    	\node[state] (b) at (-2, -2) {$b$};
    	\node[state] (c) at (2, -2) {$c$};
    	\node[state, double] (d) at (-3, -4) {$d$};
    	\node[state] (e) at (-1, -4) {$e$};
    	\node[state] (f) at (1, -4) {$f$};
    	\node[state] (g) at (3, -4) {$g$};
    	\path[->, blue] (a) edge (b);
    	\path[->, red] (a) edge (c);
    	\path[->, red] (b) edge (e);
    	\path[->, blue] (b) edge (d);
    	\path[->] (c) edge (f);
    	\path[->, red] (c) edge (g);
    	\path (d) edge [loop below] (d);
    	\path (e) edge [loop below] (e);
    	\path (f) edge [loop below] (f);
    	\path (g) edge [loop below] (g);
    	\end{tikzpicture}
    	\caption{Two NEs and one SPE}
    	\label{fig_ne_spe}
    \end{subfigure}
    \begin{subfigure}[b]{0.5\textwidth}
		\begin{tikzpicture}[->,>=latex,shorten >=1pt, initial text={}, scale=0.8, every node/.style={scale=0.8}]
		\node[state] (a) at (0, 0) {$a$};
		\node[state] (c) at (-2, 0) {$c$};
		\node[state, rectangle] (b) at (2, 0) {$b$};
		\node[state, rectangle] (d) at (4, 0) {$d$};
		\path[->] (a) edge (c);
		\path[<->] (a) edge node[above] {$\stackrel{\playcircle}{0} \stackrel{\Box}{3}$} (b);
		\path[->] (b) edge (d);
		\path (d) edge [loop right] node {$\stackrel{\playcircle}{2} \stackrel{\Box}{2}$} (d);
		\path (c) edge [loop left] node {$\stackrel{\playcircle}{1} \stackrel{\Box}{1}$} (c);
		
		\node[red] (l0) at (-3, -0.7) {$(\lambda_0)$};
		\node[red] (l0a) at (0, -0.7) {$-\infty$};
		\node[red] (l0b) at (2, -0.7) {$-\infty$};
		\node[red] (l0c) at (-2, -0.7) {$-\infty$};
		\node[red] (l0d) at (4, -0.7) {$-\infty$};
		
		\node[red] (l1) at (-3, -1.4) {$(\lambda_1)$};
		\node[red] (l1a) at (0, -1.4) {$1$};
		\node[red] (l1b) at (2, -1.4) {$2$};
		\node[red] (l1c) at (-2, -1.4) {$1$};
		\node[red] (l1d) at (4, -1.4) {$2$};
		
		\node[red] (l2) at (-3, -2.1) {$(\lambda_2)$};
		\node[red] (l2a) at (0, -2.1) {$2$};
		\node[red] (l2b) at (2, -2.1) {$2$};
		\node[red] (l2c) at (-2, -2.1) {$1$};
		\node[red] (l2d) at (4, -2.1) {$2$};
		
		\node[red] (l3) at (-3, -2.8) {$(\lambda_3)$};
		\node[red] (l3a) at (0, -2.8) {$2$};
		\node[red] (l3b) at (2, -2.8) {$3$};
		\node[red] (l3c) at (-2, -2.8) {$1$};
		\node[red] (l3d) at (4, -2.8) {$2$};
		
		\node[red] (l4) at (-3, -3.5) {$(\lambda_4)$};
		\node[red] (l4a) at (0, -3.5) {$+\infty$};
		\node[red] (l4b) at (2, -3.5) {$+\infty$};
		\node[red] (l4c) at (-2, -3.5) {$1$};
		\node[red] (l4d) at (4, -3.5) {$2$};
		\end{tikzpicture}
		\caption{A game without SPE} 
		\label{fig_sans_spe}
    \end{subfigure}
    \caption{Two examples of games}
\end{figure}

Mean-payoff games are prefix-independent.
We now recall a classical result about two-player zero-sum games.
	
	\begin{defi}[Zero-sum game]
		A game $G$, with $\Pi = \{1, 2\}$, is \emph{zero-sum} if $\mu_2 = -\mu_1$.
	\end{defi}
	
	\begin{defi}[Borel game]
		A game $G$ is \emph{Borel} if the function $\mu$, from the set $V^\omega$ equipped with the product topology to the Euclidian space $\R^\Pi$, is Borel, i.e. if, for every Borel set $B \subseteq \R^\Pi$, the set $\mu^{-1}(B)$ is Borel.
	\end{defi}
	
\begin{defi}[Determinacy]
    Let $G_{\|v_0}$ be an initialized zero-sum Borel game, with $\Pi = \{1, 2\}$.
    The game $G_{\|v_0}$ is \emph{determined} if we have the following equality:
    $$\sup_{\sigma_1} ~ \inf_{\sigma_2} ~ \mu_1(\< \bsigma \>_{v_0}) = \inf_{\sigma_2} ~ \sup_{\sigma_1} ~ \mu_1(\< \bsigma \>_{v_0}).$$
    That quantity is called \emph{value} of $G_{\|v_0}$, denoted by $\val_1(G_{\|v_0})$; \emph{solving} the game $G$ means computing its value.
\end{defi}
	
	\begin{pptn}[Determinacy of two-player zero-sum Borel games~\cite{BorelDeterminacy}]
		Zero-sum Borel games are determined.
	\end{pptn}

The following examples illustrate the SPE existence problem in mean-payoff games.

\begin{ex}\label{ex_sans_spe}
	Let $G$ be the mean-payoff game of Figure~\ref{fig_sans_spe}, where each edge is labelled by its weights $\pi_\playcircle$ and $\pi_\Box$. No weight is given for the edges $ac$ and $bd$ since they can be used only once, and therefore do not influence the final payoff.
	For now, the reader should not pay attention to the red labels below the states.
	As shown in~\cite{DBLP:journals/corr/Bruyere0PR16}, this game does not have any SPE, neither from the state $a$ nor from the state $b$.
	
	Indeed, the only NE plays from the state $b$ are the plays where player $\Box$ eventually leaves the cycle $ab$ and goes to $d$: if he stays in the cycle $ab$, then player $\Circle$ would be better off leaving it, and if she does, player $\Box$ would be better off leaving it before.
    From the state $a$, if player $\Circle$ knows that player $\Box$ will leave, she has no incentive to do it before: there is no NE where $\Circle$ leaves the cycle and $\Box$ plans to do it if ever she does not. Therefore, there is no SPE where $\Circle$ leaves the cycle.
    But then, after a history that terminates in $b$, player $\Box$ has actually no incentive to leave if player $\Circle$ never plans to do it afterwards: contradiction.
\end{ex}

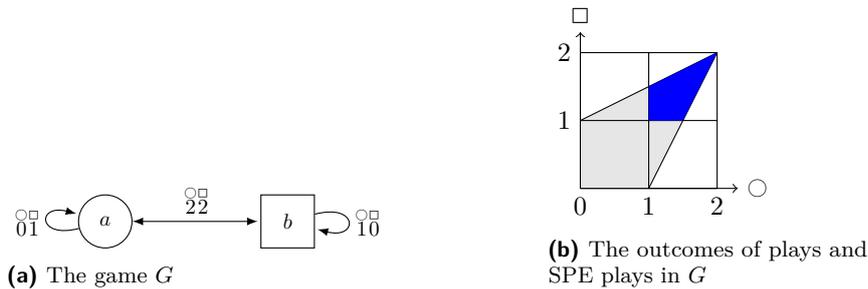
\begin{figure} 
	\centering
	\begin{subfigure}[b]{0.5\textwidth}
		\begin{tikzpicture}[->,>=latex,shorten >=1pt, initial text={}, scale=0.8, every node/.style={scale=0.8}]
		\node[state] (a) at (0, 0) {$a$};
		\node[state, rectangle] (b) at (3, 0) {$b$};
		\path[<->] (a) edge node[above, rectangle] {$\stackrel{\playcircle}{2}\stackrel{\Box}{2}$} (b);
		\path (a) edge [loop left] node {$\stackrel{\playcircle}{0}\stackrel{\Box}{1}$} (a);
		\path (b) edge [loop right] node {$\stackrel{\playcircle}{1}\stackrel{\Box}{0}$} (b);
		\end{tikzpicture}
		\caption{The game $G$}
		\label{fig_inf_spe}
	\end{subfigure}
	\begin{subfigure}[b]{0.3\textwidth}
		\begin{tikzpicture}[scale=0.9]
		\draw [->] (0,0) -- (2.3,0);
		\draw (2.3,0) node[right] {${\Circle}$};
		\draw [->] (0,0) -- (0,2.3);
		\draw (0,2.3) node[above] {${\Box}$};
		\fill [gray!20] (0, 1) -- (0, 0) -- (1, 0) -- (2, 2);
		\draw (0, 0) grid (2, 2);
		\foreach \x in {1,2} \draw(0,\x)node[left]{\x};
		\foreach \x in {0,1,2} \draw(\x,0)node[below]{\x};
		\draw [black] (0, 1) -- (2, 2);
		\draw [black] (1, 0) -- (2, 2);
		
		\fill [blue] (1, 1) -- (1, 1.5) -- (2, 2) -- (1.5, 1);
		\end{tikzpicture}
	    \caption{The outcomes of plays and SPE plays in $G$} \label{fig_inf_spe_outcomes}
	\end{subfigure}
	
	\caption{A game with an infinity of SPEs}
\end{figure}

\begin{ex} \label{ex_inf_spe}
    Let us now study the game of Figure~\ref{fig_inf_spe}.
    Using techniques from \cite{DBLP:conf/concur/ChatterjeeDEHR10}, we can represent the outcomes of possible plays in that game as in Figure~\ref{fig_inf_spe_outcomes} (gray and blue areas).

	Following exclusively one of the three simple cycles $a$, $ab$ and $b$ of the game graph during a play yields the outcomes $01, 10$ and $22$, respectively.
	By combining those cycles with well chosen frequencies, one can obtain any outcome in the convex hull of those three points.
	Now, it is also possible to obtain the point $00$ by using the properties of the limit inferior: it is for instance the outcome of the play $a^2 b^4 a^{16} b^{256} \dots a^{2^{2^n}} b^{2^{2^{n+1}}} \dots$.
	In fact, one can construct a play that yields any outcome in the convex hull of the four points $00, 10, 01$, and $22$.

	We claim that the outcomes of SPEs plays correspond to the entire blue area in Figure~\ref{fig_inf_spe_outcomes}: there exists an SPE $\bsigma$ in $G_{\|a}$ with $\< \bsigma \>_a = \rho$ if and only if $\mu_{\Box}(\rho), \mu_\playcircle(\rho) \geq 1$.
	That statement will be a direct consequence of the results we show in the remaining sections, but let us give a first intuition: a play with such an outcome necessarily uses infinitely often both states. It is an NE play because none of the players can get a better payoff by looping forever on their state, and they can both force each other to follow that play, by threatening them to loop for ever on their state whenever they can. But such a strategy profile is clearly not an SPE.
	
	It can be transformed into an SPE as follows: when a player deviates, say player $\Box$, then player $\Circle$ can punish him by looping on $a$, not forever, but a great number of times, until player $\Box$'s mean-payoff gets very close to $1$. 
	Afterwards, both players follow again the play that was initially planned. 
	Since that threat is temporary, it does not affect player $\Circle$'s payoff on the long term, but it really punishes player $\Box$ if that one tries to deviate infinitely often.
	
\end{ex}
	

	\section{Requirements and negotiation} \label{sec_negotiation}

We will now see that SPEs are strategy profiles that respect some \emph{requirements} about the payoffs, depending on the states it traverses. In this part, we develop the notions of \emph{requirement} and \emph{negotiation}.

		\subsection{Requirement}

In the method we will develop further, we will need to analyze the players' behaviour when they have some \emph{requirement} to satisfy.
Intuitively, one can see requirements as \emph{rationality constraints} for the players, that is, a threshold payoff value under which a player will not accept to follow a play.
In all what follows, $\bR$ denotes the set $\R \cup \{\pm \infty\}$.

\begin{defi}[Requirement]
	A \emph{requirement} on the game $G$ is a function $\lambda: V \to \bR$.
	
	For a given state $v$, the quantity $\lambda(v)$ represents the minimal payoff that the player controlling $v$ will require in a play beginning in $v$.
\end{defi}

\begin{defi}[$\lambda$-consistency]
	Let $\lambda$ be a requirement on a game $G$. A play $\rho$ in $G$ is \emph{$\lambda$-consistent} if and only if, for all $i \in \Pi$ and $n \in \N$ with $\rho_n \in V_i$, we have $\mu_i(\rho_n \rho_{n+1} \dots)~\geq~\lambda(\rho_n)$.
	The set of the $\lambda$-consistent plays from a state $v$ is denoted by $\lCons(v)$.
\end{defi}

\begin{defi}[$\lambda$-rationality]
	Let $\lambda$ be a requirement on a mean-payoff game $G$. Let $i \in \Pi$. A strategy profile $\bsigma_{-i}$ is \emph{$\lambda$-rational} if and only if there exists a strategy $\sigma_i$ such that, for every history $hv$ compatible with $\bsigma_{-i}$, the play $\< \bsigma_{\|hv} \>_v$ is $\lambda$-consistent.
	We then say that the strategy profile $\bsigma_{-i}$ is $\lambda$-rational \emph{assuming} $\sigma_i$.
	The set of $\lambda$-rational strategy profiles in $G_{\|v}$ is denoted by $\lRat(v)$.	
\end{defi}
	
Note that $\lambda$-rationality is a property of a strategy profile for all the players but one, player $i$. Intuitively, their rationality is justified by the fact that they collectively assume that player $i$ will, eventually, play according to the strategy $\sigma_i$: if player $i$ does so, then everyone gets their payoff satisfied.
Finally, let us define a particular requirement: the \emph{vacuous requirement}, that requires nothing, and with which every play is consistent.

\begin{defi}[Vacuous requirement]
	In any game, the \emph{vacuous requirement}, denoted by $\lambda_0$, is the requirement constantly equal to $-\infty$.
\end{defi}

		\subsection{Negotiation} \label{ss_def_nego}

We will show that SPEs in prefix-independent games are characterized by the fixed points of a function on requirements. That function can be seen as a \emph{negotiation}: when a player has a requirement to satisfy, another player can hope a better payoff than what they can secure in general, and therefore update their own requirement.

\begin{defi}[Negotiation function]
	Let $G$ be a game. The \emph{negotiation function} is the function that transforms any requirement $\lambda$ on $G$ into a requirement $\nego(\lambda)$ on $G$, such that for each $i \in \Pi$ and $v \in V_i$, with the convention $\inf\emptyset = +\infty$, we have:
	$$\nego(\lambda)(v) = \underset{\bsigma_{-i} \in \lRat(v)}{\inf} \underset{\sigma_i}{\sup}~ \mu_i(\< \bsigma\>_v).$$
\end{defi}

\begin{reks}
There exists a $\lambda$-rational strategy profile from $v$ against the player controlling $v$ if and only if $\nego(\lambda)(v) \neq +\infty$.
The negotiation function is monotone: if $\lambda \leq \lambda'$ (for the pointwise order, i.e. if for each $v$, $\lambda(v) \leq \lambda'(v)$), then $\nego(\lambda) \leq \nego(\lambda')$.
The negotiation function is also non-decreasing: for every $\lambda$, we have $\lambda \leq \nego(\lambda)$.
\end{reks}

In the general case, the quantity $\nego(\lambda)(v)$ represents the worst case value that the player controlling $v$ can ensure, assuming that the other players play $\lambda$-rationally.

\begin{ex}
    Let us consider the game of Example~\ref{ex_sans_spe}: in Figure~\ref{fig_sans_spe}, on the two first lines below the states, we present the requirements $\lambda_0$ and $\lambda_1 = \nego(\lambda_0)$, which is easy to compute since any strategy profile is $\lambda_0$-rational: for each $v$, $\lambda_1(v)$ is the classical \emph{worst-case value} or \emph{antagonistic value} of $v$, i.e. the best value the player controlling $v$ can enforce against a fully hostile environment. Let us now compute the requirement $\lambda_2 = \nego(\lambda_1)$.
	
	From $c$, there exists exactly one $\lambda_1$-rational strategy profile $\bsigma_{-\playcircle} = \sigma_\Box$, which is the empty strategy since player $\Box$ has never to choose anything. Against that strategy, the best and the only payoff player $\Circle$ can get is $1$, hence $\lambda_2(c) = 1$.
	For the same reasons, $\lambda_2(d) = 2$.
	
	From $b$, player $\Circle$ can force $\Box$ to get the payoff $2$ or less, with the strategy profile $\sigma_\playcircle: h \mapsto c$. Such a strategy is $\lambda_1$-rational, assuming the strategy $\sigma_\Box: h \mapsto d$. Therefore, $\lambda_2(b) = 2$.
		
	Finally, from $a$, player $\Box$ can force $\Circle$ to get the payoff $2$ or less, with the strategy profile $\sigma_\Box: h \mapsto d$. Such a strategy is $\lambda_1$-rational, assuming the strategy $\sigma_\playcircle: h \mapsto c$. But, he cannot force her to get less than the payoff $2$, because she can force the access to the state $b$, and the only $\lambda_1$-consistent plays from $b$ are the plays with the form $(ba)^k b d^\omega$. Therefore, $\lambda_2(a) = 2$.
\end{ex}

		\subsection{Steady negotiation}

In what follows, we will often need a game to be \emph{with steady negotiation}, i.e. such that there always exists a worst $\lambda$-rational behaviour for the environment against a given player.

\begin{defi}[Game with steady negotiation]
	A game $G$ is \emph{with steady negotiation} if and only if for every player $i$, for every vertex $v$, and for every requirement $\lambda$, the set $\left\{\left. \sup_{\sigma_i} ~\mu_i(\< \bsigma_{-i}, \sigma_i \>_v) ~\right|~ \bsigma_{-i} \in \lRat(v)\right\}$
	is either empty, or has a minimum.
\end{defi}

\begin{rek}
    In particular, when a game is with steady negotiation, the infimum in the definition of negotiation is always reached.
\end{rek}

It will be proved in Section \ref{sec_negotiation_games} that mean-payoff games are with steady negotiation.

		\subsection{Link with Nash equilibria}

Requirements and the negotiation function are able to capture Nash equilibria. Indeed, if $\lambda_0$ is the vacuous requirement, then $\nego(\lambda_0)$ characterizes the plays that are supported by a Nash equilibrium (abbreviated by NE plays), in the following formal sense:

\begin{thm}[App. \ref{pf_ne}] \label{thm_ne}
	Let $G$ be a game with steady negotiation. Then, a play $\rho$ in $G$ is an NE play if and only if $\rho$ is $\nego(\lambda_0)$-consistent.
\end{thm}


    
    

\begin{ex}
    Let us consider again the game of Example~\ref{ex_sans_spe}, with the requirement $\lambda_1$ given in Figure~\ref{fig_sans_spe}. The only $\lambda_1$-consistent plays in this game, starting from the state $a$, are $ac^\omega$, and $(ab)^k d^\omega$ with $k \geq 1$. One can check that those plays are exactly the NE plays in that game.
\end{ex}

In the following section, we will prove that as well as $\nego(\lambda_0)$ characterizes the NEs, the requirement that is the least fixed point of the negotiation function characterizes the SPEs.

	\section{Link between negotiation and SPEs} \label{sec_link_nego_spe}

The notion of negotiation will enable us to find the SPEs, but also more generally the $\epsilon$-SPEs, in a game. For that purpose, we need the notion of $\epsilon$-fixed points of a function.

\begin{defi}[$\epsilon$-fixed point]
	Let $\epsilon \geq 0$, let $D$ be a finite set and let $f: \bR^D \to \bR^D$ be a mapping. A tuple $\bx \in \R^D$ is a \emph{$\epsilon$-fixed point} of $f$ if for each $d \in D$, for $\by = f(\bx)$, we have $y_d \in [x_d - \epsilon, x_d + \epsilon]$.
\end{defi}

\begin{rek}
    A $0$-fixed point is a fixed point, and conversely.
\end{rek}

The set of requirements, equipped with the componentwise order, is a complete lattice. Since the negotiation function is monotone, Tarski's fixed point theorem states that the negotiation function has a least fixed point. That result can be generalized to $\epsilon$-fixed points:

\begin{lm}[App. \ref{pf_least_delta_fixed point}] \label{lm_least_delta_fixed point}
	Let $\epsilon \geq 0$. On each game, the function $\nego$ has a least $\epsilon$-fixed point.
\end{lm}

Intuitively, the $\epsilon$-fixed points of the negotiation function are the requirements $\lambda$ such that, from every vertex $v$, the player $i$ controlling $v$ cannot enforce a payoff greater than $\lambda(v) + \epsilon$ against a $\lambda$-rational behaviour.
Therefore, the $\lambda$-consistent plays are such that if one player tries to deviate, it is possible for the other players to prevent them improving their payoff by more than $\epsilon$, while still playing rationally.
Formally:

\begin{thm}[App.~\ref{pf_spe}] \label{thm_spe}
	Let $G_{\|v_0}$ be an initialized prefix-independent game, and let $\epsilon \geq 0$.
	Let $\lambda^*$ be the least $\epsilon$-fixed point of the negotiation function.
	Let $\xi$ be a play starting in $v_0$.
	If there exists an $\epsilon$-SPE $\bsigma$ such that $\< \bsigma \>_{v_0} = \xi$, then $\xi$ is $\lambda^*$-consistent.
	The converse is true if the game $G$ is with steady negotiation.
\end{thm}

	\section{Negotiation games} \label{sec_negotiation_games}

We have now proved that SPEs are characterized by the requirements that are fixed points of the negotiation function; but we need to know how to compute, in practice, the quantity $\nego(\lambda)$ for a given requirement $\lambda$. In other words, we need a algorithm that computes, given a state $v_0$ controlled by a player $i$ in the game $G$, and given a requirement $\lambda$, which value player $i$ can ensure in $G_{\|v_0}$ if the other players play $\lambda$-rationally.

		\subsection{Abstract negotiation game}

We first define an \emph{abstract negotiation game}, that is conceptually simple but not directly usable for an algorithmic purpose, because it is defined on an uncoutably infinite state space.

A similar definition was given in \cite{DBLP:journals/mor/FleschP17}, as a tool in a general method to compute SPE plays in games whose payoff functions have finite range, which is not the case of mean-payoff games.
Here, linking that game with our concepts of requirements, negotiation function and steady negotiation enables us to present an effective algorithm in the case of mean-payoff games, by constructing a finite version of the abstract negotiation game, the \emph{concrete negotiation game}, and afterwards by analyzing the negotiation function with linear algebra tools.

The abstract negotiation game from a state $v_0$, with regards to a player $i$ and a requirement $\lambda$, is denoted by $\Abs_{\lambda i}(G)_{\|[v_0]}$ and opposes two players, \emph{Prover} and \emph{Challenger}, as follows:
	
	\begin{itemize}
		\item Prover proposes a $\lambda$-consistent play $\rho$ from $v_0$ (or loses, if she has no play to propose).
		
		\item Then, either Challenger accepts the play and the game terminates; or, he chooses an edge $\rho_k \rho_{k+1}$, with $\rho_k \in V_i$, from which he can make player $i$ deviate, using another edge $\rho_k v$ with $v \neq \rho_{k+1}$: then, the game starts again from $v$ instead of $v_0$.
		
		\item In the resulting play (either eventually accepted by Challenger, or constructed by an infinity of deviations), Prover wants player $i$'s payoff to be low, and Challenger wants it to be high.
	\end{itemize}

	That game gives us the basis of a method to compute $\nego(\lambda)$ from $\lambda$: the maximal outcome that Challenger --- or $\C$ for short --- can ensure in $\Abs_{\lambda i}(G)_{\|[v_0]}$, with $v_0 \in V_i$, is also the maximal payoff that player $i$ can ensure in $G_{\|v_0}$, against a $\lambda$-rational environment; hence the equality $\val_\C\left(\Abs_{\lambda i}(G)_{\|[v_0]}\right) = \nego(\lambda)(v_0).$
	A proof of that statement, with a complete formalization of the abstract negotiation game, is presented in Appendix~\ref{app_abstract}.

\begin{ex} \label{ex_abstract_game}
Let us consider again the game of Example~\ref{ex_sans_spe}: the requirement $\lambda_2 = \nego(\lambda_1)$, computed in Section~\ref{ss_def_nego}, is also presented on the third line below the states in Figure~\ref{fig_sans_spe}.
Let us use the abstract negotiation game to compute the requirement $\lambda_3 = \nego(\lambda_2)$.

From $a$, Prover can propose the play $abd^\omega$, and the only deviation Challenger can do is going to $c$; he has of course no incentive to do it. Therefore, $\lambda_3(a) = 2$.
From $b$, whatever Prover proposes at first, Challenger can deviate and go to $a$. Then, from $a$, Prover cannot propose the play $ac^\omega$, which is not $\lambda_2$-consistent: she has to propose a play beginning by $ab$, and to let Challenger deviate once more. He can then deviate infinitely often that way, and generate the play $(ba)^\omega$: therefore, $\lambda_3(b) = 3$.
The other states keep the same values.
Note that there exists no $\lambda_3$-consistent play from $a$ or $b$, hence $\nego(\lambda_3)(a) = \nego(\lambda_3)(b) = +\infty$.
This proves that there is no SPE in that game.
\end{ex}

The interested reader will find other such examples in Appendix~\ref{app_ex}.

		\subsection{Concrete negotiation game}

In the abstract negotiation game, Prover has to propose complete plays, on which we can make the hypothesis that they are $\lambda$-consistent. In practice, there will often be an infinity of such plays, and therefore it cannot be used directly for an algorithmic purpose. Instead, those plays can be given edge by edge, in a finite state game. Its definition is more technical, but it can be shown that it is equivalent to the abstract one.
In order to make the definition as clear as possible, we give it only when the original game is a mean-payoff game.
However, one could easily adapt this definition to other classes of prefix-independent games.

\begin{defi}[Concrete negotiation game]
	Let $G_{\|v_0}$ be an initialized mean-payoff game, and let $\lambda$ be a requirement on $G$, with either $\lambda(V) \subseteq \R$, or $\lambda = \lambda_0$.
	
	The \emph{concrete negotiation game} of $G_{\|v_0}$ for player $i$ is the two-player zero-sum game $\Conc_{\lambda i}(G)_{\|s_0}~=~\left( \{\P, \C\}, S, (S_{\P}, S_{\C}), \Delta, \nu\right)_{\|s_0}$, defined as follows:

	\begin{itemize}
		\item The set of states controlled by Prover is $S_{\P} = V \times 2^V$, where the state $s = (v, M)$ contains the information of the current state $v$ on which Prover has to define the strategy profile, and the \emph{memory} $M$ of the states that have been traversed so far since the last deviation, and that define the requirements Prover has to satisfy.
		The initial state is $s_0 = (v_0, \{v_0\})$.
		
		\item The set of states controlled by Challenger is $S_{\C} = E \times 2^V$, where in the state $s = (uv, M)$, the edge $uv$ is the edge proposed by Prover.
		
		\item The set $\Delta$ contains three types of transitions: \emph{proposals}, \emph{acceptations} and \emph{deviations}.
		
		\begin{itemize}
			\item The proposals are transitions in which Prover proposes an edge of the game $G$:
			$$\Prop = \left\{ (v, M) (vw, M) ~\left|~
			vw \in E, M \in 2^V \right.\right\};$$
			
			\item the acceptations are transitions in which Challenger accepts to follow the edge proposed by Prover (it is in particular his only possibility when that edge begins on a state that is not controlled by player $i$) --- note that the memory is updated:
			$$\Acc = \left\{ (vw, M)\left(w, M \cup \{w\}\right) ~\left|~
			j \in \Pi, w \in V_j
			\right.\right\};$$
			
			\item the deviations are transitions in which Challenger refuses to follow the edge proposed by Prover, as he can if that edge begins in a state controlled by player $i$ --- the memory is erased, and only the new state the deviating edge leads to is memorized:
			$$\Dev = \left\{ (uv, M) (w, \{w\})	~\left|~ 
			u \in V_i, w \neq v, uw \in E
			\right.\right\}.$$
		\end{itemize}

		\item On those transitions, we define a multidimensional weight function $\hpi: \Delta \to \R^{\Pi \cup \{\star\}}$, with one dimension per player (\emph{non-main} dimensions) plus one special dimension (\emph{main} dimension) denoted by the symbol $\star$.
		For each non-main dimension $j \in \Pi$, we define:
		\begin{itemize}
		    \item on proposals: $\hpi_j \left((v, M) (vw, M)\right) = 0$;
		    
		    \item on acceptations and deviations: $\hpi_j\left((uv, M) (w, N)\right) = 2 \left(\pi_j(uw) - \underset{v_j \in M \cap V_j}{\max} \lambda(v_j)\right)$;
		\end{itemize}
		
    	and on the main dimension:
    	\begin{itemize}
    	    \item on proposals: $\hpi_\star\left((v, M), (vw, M)\right) = 0$;
    	    
    	    \item on acceptations and deviations: $\hpi_\star\left((uv, M), (w, N)\right) = 2 \pi_i(uw)$.
    	\end{itemize}

	    For each dimension $d$, we write $\hmu_d$ the corresponding mean-payoff function:
	    $$\hmu_d(\rho) = \liminf_{n \in \N} \frac{1}{n} \sum_{k=0}^{n-1} \hpi_d(\rho_k \rho_{k+1}).$$
	    
	    Thus, the mean-payoff along the main dimension corresponds to player $i$'s payoff, while the mean-payoff along a non-main dimension $j$ corresponds to player $j$'s payoff... minus the maximal requirement player $j$ has to satisfy.

		\item Then, the outcome function $\nu_\C = -\nu_\P$ measures player $i$'s payoff, with a winning condition if the constructed strategy profile is not $\lambda$-rational, that is to say if after finitely many player $i$'s deviations, it can generate a play which is not $\lambda$-consistent:
		
		\begin{itemize}
		    \item $\nu_\C(\eta) = +\infty$ if after some index $n \in \N$, the play $\eta_n \eta_{n+1} \dots$ contains no deviation, and if  $\hmu_j(\eta) < 0$ for some $j \in \Pi$;
		    
		    \item $\nu_\C(\eta) = \hmu_\star(\eta)$ otherwise.
		\end{itemize}
	\end{itemize}
\end{defi}

Like in the abstract negotiation game, the goal of Challenger is to find a $\lambda$-rational strategy profile that forces the worst possible payoff for player $i$, and the goal of Prover is to find a possibly deviating strategy for player $i$ that gives them the highest possible payoff.

A play or a history in the concrete negotiation game has a projection in the game on which that negotiation game has been constructed, defined as follows:

\begin{defi}[Projection of a history, of a play]
	Let $G$ be a prefix-independent game. Let $\lambda$ be a requirement and $i$ a player, and let $\Conc_{\lambda i}(G)$ be the corresponding concrete negotiation game. Let $H = (h_0, M_0) (h_0h'_0, M_0) \dots (h_nh'_n, M_n)$ be a history in $\Conc_{\lambda i}(G)$: the \emph{projection} of the history $H$ is the history $\dH = h_0 \dots h_n$ in the game $G$.
	That definition is naturally extended to plays.
\end{defi}

\begin{rek}
    For a play $\eta$ without deviations, we have $\hmu_j(\eta) \geq 0$ for each $j \in \Pi$ if and only if $\deta$ is $\lambda$-consistent.
\end{rek}

The concrete negotiation game is equivalent to the abstract one: the only differences are that the plays proposed by Prover are proposed edge by edge, and that their $\lambda$-consistency is not written in the rules of the game but in its outcome function.

\begin{thm}[App. \ref{pf_concrete_game}] \label{thm_concrete_game}
	Let $G_{\|v_0}$ be an initialized mean-payoff game. Let $\lambda$ be a requirement and $i$ a player. Then, we have:
	$$\val_\C\left(\Conc_{\lambda i}(G)_{\|s_0}\right) = \inf_{\bsigma_{-i} \in \lRat(v_0)} ~ \sup_{\sigma_i} ~ \mu_i(\< \bsigma \>_{v_0}).$$
\end{thm}

An example of concrete negotiation game is given in Appendix \ref{ex_concrete_game}.

		\subsection{Solving the concrete negotiation game}

We now know that $\nego(\lambda)(v)$, for a given requirement $\lambda$, a given player $i$ and a given state $v \in V_i$, is the value of the concrete negotiation game $\Conc_{\lambda i}(G)_{\|(v, \{v\})}$. Let us now show how, in the mean-payoff case, that value can be computed.

\begin{defi}[Memoryless strategy]
	A strategy $\sigma_i$ in a game $G$ is \emph{memoryless} if for all vertices $v \in V_i$ and for all histories $h$ and $h'$, we have $\sigma_i(hv) = \sigma_i(h'v)$.
\end{defi}

For any game $G$ and any memoryless strategy $\sigma_i$, $G[\sigma_i]$ denotes the graph \emph{induced} by $\sigma_i$, that is the graph $(V, E')$, with $E' = \left\{vw \in E ~|~ v \not\in V_i \mathrm{~or~} w = \sigma_i(v)\right\}.$
For any finite set $D$ and any set $X \subseteq \R^D$, $\Conv X$ denotes the convex hull of $X$.

We can now prove that in the concrete negotiation game constructed from a mean-payoff game, Challenger has an optimal strategy that is memoryless.

\begin{lm}[App. \ref{pf_memoryless}] \label{lm_memoryless}
	Let $G_{\|v_0}$ be an initialized mean-payoff game, let $i$ be a player, let $\lambda$ be a requirement and let $\Conc_{\lambda i}(G)_{\|s_0}$ be the corresponding concrete negotiation game. There exists a memoryless strategy $\tau_\C$ that is optimal for Challenger, i.e. such that:
	$$\underset{\tau_\P}{\inf} ~\nu_\C(\< \btau \>_{s_0}) = \val_\C\left( \Conc_{\lambda i}(G)_{\|s_0} \right).$$
\end{lm}

    

For every game $G_{\|v_0}$ and each player $i$, $\ML_i\left(G_{\|v_0}\right)$, or $\ML\left(G_{\|v_0}\right)$ when the context is clear, denotes the set of memoryless strategies for player $i$ in $G_{\|v_0}$.
When $(V, E)$ is a graph, $\SC(V, E)$ denotes the set of its simple cycles, and $\SConn(V, E)$ the set of its strongly connected components.
For any closed set $C \subseteq \R^{\Pi \cup \{\star\}}$, the quantity $\minstar C = \min \left\{ x_\star ~|~ \bx \in C, \forall j \in \Pi, x_j \geq 0 \right\}$ is the \emph{$\star$-minimum} of $C$: it will capture, in the concrete  negotiation game, the least payoff that can be imposed on player $i$ while keeping every player's payoff above their requirements, among a set of possible outcomes.

With Lemma \ref{lm_memoryless}, we can now solve the concrete negotiation game.

\begin{lm}[App. \ref{pf_resolution_concrete_game}] \label{lm_resolution_concrete_game}
	Let $G_{\|v_0}$ be an initialized mean-payoff game, and let $\Conc_{\lambda i}(G)_{\|s_0}$ be its concrete negotiation game for some $\lambda$ and some $i$. Then, the value of the game $\Conc_{\lambda i}(G)_{\|s_0}$ is given by the formula:
	$$\max_{\tau_\C \in \ML_\C\left(\Conc_{\lambda i}(G)\right)} ~
	\min_{\scriptsize{\begin{matrix}
			K \in \SConn\left(\Conc_{\lambda i}(G)[\tau_\C]\right) \\
			\mathrm{accessible~from~} s_0
			\end{matrix}}} \opt(K),$$
	where $\opt(K)$ is the minimal value $\nu_\C(\rho)$ for $\rho$ among the infinite paths in $K$.
	
	If $K$ contains a deviation, then Prover can choose among its simple cycles the one that minimizes player $i$'s payoff:
	$$\opt(K) = \underset{c \in \SC(K)}{\min} ~ \hmu_\star(c^\omega).$$
		
	If $K$ does not contain a deviation, then Prover must choose a combination of its simple cycles that minimizes the main dimension while keeping the other dimensions above $0$:
	$$\opt(K) = \minstar \underset{c \in \SC(K)}{\Conv} ~ \hmu(c^\omega).$$
\end{lm}

\begin{cor}
	For each player $i$ and every state $v \in V_i$, the value $\nego(\lambda)(v)$ can be computed with the formula given in Lemma \ref{lm_resolution_concrete_game} applied to the game $\Conc_{\lambda i}(G)_{\|(v, \{v\})}$
\end{cor}

Another corollary of that result is that there always exists a best play that Prover can choose, i.e. Prover has an optimal strategy; by Theorem \ref{thm_concrete_game}, this is equivalent to saying that:

\begin{cor}
	Mean-payoff games are games with steady negotiation.
\end{cor}

	\section{Analysis of the negotiation function in mean-payoff games} \label{sec_analysis_and_algo}

When one wants to compute the least fixed point of a function, the usual method is to iterate it on the minimal element of the considered set, to go until that fixed point.
That approach is sufficient in many simple examples.
In Appendix~\ref{app_nego_seq}, we present its technical details, and an example on which it does not enable to find the least fixed point in a finite number of iterations; which is why another approach is necessary.

In this section, we will show that, in the case of mean-payoff games, the negotiation function is a piecewise linear function from the vector space of requirements into itself, which can therefore be computed and analyzed using classical linear algebra techniques.
Then, it becomes possible to search for the fixed points or the $\epsilon$-fixed points of such a function, and to decide the existence or not of SPEs or $\epsilon$-SPEs in the game studied.

\begin{thm}[App.~\ref{pf_formula_nego}] \label{thm_formula_nego}
	Let $G$ be a mean-payoff game. Let us assimilate any requirement $\lambda$ on $G$ with finite values to the tuple $\blambda = (\lambda(v))_{v \in V}$, element of the vector space $\R^V$. Then, for each player $i$ and every vertex $v_0 \in V_i$, the quantity $\nego(\lambda)(v_0)$ is a piecewise linear function of $\blambda$, and an effective expression of that function can be computed in 2-\textsc{ExpTime}.
\end{thm}

	
	

\begin{ex}
    Let us consider the game of Example~\ref{ex_inf_spe}.
    If a requirement $\lambda$ is represented by the tuple $(\lambda(a), \lambda(b))$, the function $\nego: \R^2 \to \R^2$ can be represented by Figure~\ref{fig_linear}, where in any one of the regions delimited by the dashed lines, we wrote a formula for the couple $(\nego(\lambda)(a), \nego(\lambda)(b))$.
    The orange area indicates the fixed points of the function, and the yellow area the other $\frac{1}{2}$-fixed points.

\end{ex}

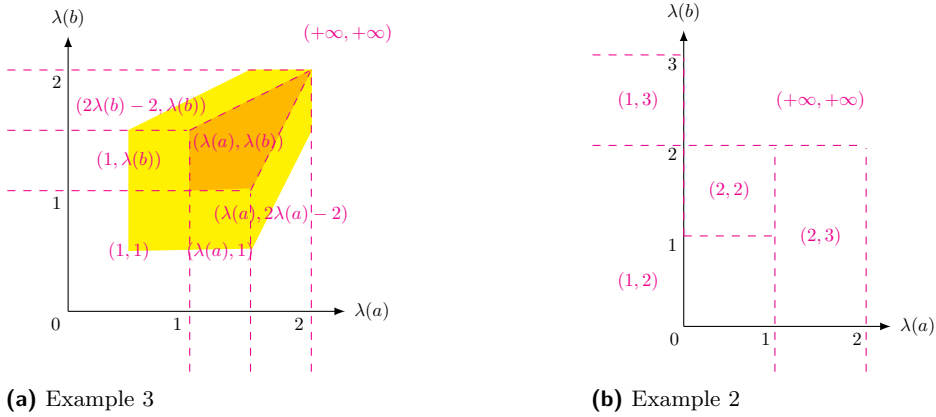
\begin{figure}
	\begin{center}
	\begin{subfigure}[b]{0.5\textwidth}
        	\begin{tikzpicture}[scale=1.6,>=latex,shorten >=1pt, every node/.style={scale=0.7}]
        	\filldraw[yellow] (1/2, 1/2) -- (1/2, 3/2) -- (3/2, 2) -- (2, 2) -- (2, 3/2) -- (3/2, 1/2);
        	\filldraw[orange, opacity=0.5] (1, 1) -- (1, 1.5) -- (2, 2) -- (1.5, 1);
        	
        	\draw[->] (0,0) -- (2.3,0);
        	\draw (2.3,0) node[right] {$\lambda(a)$};
        	\draw[->] (0,0) -- (0,2.3);
        	\draw (0,2.3) node[above] {$\lambda(b)$};
        	\foreach \x in {1,2} \draw(0,\x-0.1)node[left]{\x};
        	\foreach \x in {0,1,2} \draw(\x-0.1,0)node[below]{\x};
        	
        	\path[magenta, dashed] (1, -0.5) edge (1, 1.5);
        	\path[magenta, dashed] (2, -0.5) edge (2, 2);
        	\path[magenta, dashed] (-0.5, 1) edge (1.5, 1);
        	\path[magenta, dashed] (-0.5, 2) edge (2, 2);
        	\path[magenta, dashed] (1, 1.5) edge (2, 2);
        	\path[magenta, dashed] (1.5, 1) edge (2, 2);
        	\path[magenta, dashed] (-0.5, 1.5) edge (1, 1.5);
        	\path[magenta, dashed] (1.5, -0.5) edge (1.5, 1);
        	
        	\draw[magenta] (0.5, 0.5) node {$(1, 1)$};
        	\draw[magenta] (0.5, 1.25) node {$(1, \lambda(b))$};
        	\draw[magenta] (1.25, 0.5) node {$(\lambda(a), 1)$};
        	\draw[magenta] (0.6, 1.7) node {$(2\lambda(b)-2, \lambda(b))$};
        	\draw[magenta] (1.75, 0.8) node {$(\lambda(a), 2\lambda(a)-2)$};
        	\draw[magenta] (1.4, 1.4) node {$(\lambda(a), \lambda(b))$};
        	\draw[magenta] (2.3, 2.3) node {$(+\infty,  +\infty)$};
        	\end{tikzpicture}
        \caption{Example~\ref{ex_inf_spe}} \label{fig_linear}
        \end{subfigure}
        \hfill
        \begin{subfigure}[b]{0.45\textwidth}
		\begin{tikzpicture}[scale=1.2,>=latex,shorten >=1pt, every node/.style={scale=0.7}]
		
		\draw[->] (0,0) -- (2.3,0);
		\draw (2.3,0) node[right] {$\lambda(a)$};
		\draw[->] (0,0) -- (0,3.3);
		\draw (0,3.3) node[above] {$\lambda(b)$};
		\foreach \x in {1,2,3} \draw(0,\x-0.1)node[left]{\x};
		\foreach \x in {0,1,2} \draw(\x-0.1,0)node[below]{\x};
		
		\path[magenta, dashed] (0, 1) edge (0, 3);
		\path[magenta, dashed] (1, -0.5) edge (1, 2);
		\path[magenta, dashed] (2, -0.5) edge (2, 2);
		\path[magenta, dashed] (0, 1) edge (1, 1);
		\path[magenta, dashed] (-1, 2) edge (2, 2);
		\path[magenta, dashed] (-1, 3) edge (0, 3);
		
		\draw[magenta] (-0.5, 0.5) node {$(1, 2)$};
		\draw[magenta] (0.5, 1.5) node {$(2, 2)$};
		\draw[magenta] (1.5, 1) node {$(2, 3)$};
		\draw[magenta] (-0.5, 2.5) node {$(1, 3)$};
		\draw[magenta] (1.5, 2.5) node {$(+\infty,  +\infty)$};
		\end{tikzpicture}
		\caption{Example~\ref{ex_sans_spe}} \label{fig_linear2}
	    \end{subfigure}
	\end{center}
	\caption{The negotiation function on the games of Examples~\ref{ex_inf_spe} and~\ref{ex_sans_spe}}
    \label{fig_linear_total}
\end{figure}

\begin{ex}
	Now, let us consider the game of Example~\ref{ex_sans_spe}.
	If we fix $\lambda(c) = 1$ and $\lambda(d) = 2$, and represent the requirements $\lambda$ by the tuples $(\lambda(a), \lambda(b))$, as in the previous example. Then, the negotiation function can be represented as in Figure~\ref{fig_linear2}.
	One can check that there is no fixed point here, and even no $\frac{1}{2}$-fixed point --- except $(+\infty, +\infty)$.
\end{ex}

    \section{Conclusion: algorithm and complexity}

Thanks to all the previous results, we are now able to compute the least fixed point, or the least $\epsilon$-fixed point, of the negotiation function, on every mean-payoff game, and to use it as a characterization of all the SPEs or all the $\epsilon$-SPEs.
A direct application is an algorithm that solves the \emph{$\epsilon$-SPE constrained existence problem}, i.e. that decides, given an initialized mean-payoff game $G_{\|v_0}$, two thresholds $\bx, \by \in \Q^\Pi$, and a rational number $\epsilon \geq 0$, whether there exists an SPE $\bsigma$ such that $\bx \leq \mu(\< \bsigma \>_{v_0}) \leq \by$.

We leave for future work the optimal complexity of that problem.
However, we can easily prove that it cannot be solved in polynomial time, unless $\mathbf{P = NP}$.

\begin{thm}[App. \ref{app_np_hard}] \label{thm_np_hard}
    The $\epsilon$-SPE constrained existence problem is NP-hard.
\end{thm}

Given $G_{\|v_0}$, by Theorem \ref{thm_formula_nego}, computing a general expression of the negotiation function as a piecewise linear function can be done in time double exponential in the size of $G$.
Then, for each linear piece of $\nego$, computing its set of $\epsilon$-fixed points is a polynomial problem. Since the number of pieces is at most double exponential in the size of $G$, computing its entire set of fixed points, and thus its least $\epsilon$-fixed point $\lambda$, can be done in double exponential time.
    
Then, from the requirement $\lambda$ and the thresholds $\bx$ and $\by$, we can construct a multi-mean-payoff automaton $\A_{\lambda}$ of exponential size, that accepts an infinite word $\rho \in V^\omega$, if and only if $\rho$ is a $\lambda$-consistent play of $G_{\|v_0}$, and $\bx \leq \mu(\rho) \leq \by$ --- see Appendix \ref{app_automaton} for the construction of $\A_\lambda$.
    
Finally, by Theorem \ref{thm_spe}, there exists an SPE $\bsigma$ in $G_{\|v_0}$ with $\bx \leq \mu(\< \bsigma \>_{v_0}) \leq \by$ if and only if the language of the automaton $\A_{\lambda}$ is nonempty, which can be known in a time polynomial in the size of $\A_{\lambda}$ (see for example \cite{DBLP:conf/fossacs/AlurDMW09}), i.e. in a time exponential in the size of $G$.
We can therefore conclude on the following result:

\begin{thm} \label{thm_decidable}
    The $\epsilon$-SPE constrained existence problem 
    is decidable and 2-\textsc{ExpTime}-easy.
\end{thm}

\bibliography{ARPE}

\newpage


\appendix

The following appendices are providing the detailed proofs of all our results. They are not necessary to understand our results and are meant to provide full formalization and rigorous proofs. They also provide further intuitions through additional examples for the interested reader. To improve readability, we have chosen to recall the statements that appeared in the main body of the paper before giving their detailed proofs in order to ease the work of the reader.

	\section{Proof of Theorem \ref{thm_ne}} \label{pf_ne}

\noindent \textbf{Theorem~\ref{thm_ne}.} \emph{Let $G$ be a game with steady negotiation. Then, a play $\rho$ in $G$ is an NE play if and only if $\rho$ is $\nego(\lambda_0)$-consistent.}

\begin{proof} 
\begin{itemize}
	\item Let $\bsigma$ be a Nash equilibrium in $G_{\|v_0}$, for some state $v_0$, and let $\rho = \< \bsigma \>_{v_0}$ : let us prove that the play $\rho$ is $\nego(\lambda_0)$-consistent.
	
	Let $k \in \N$, let $i \in \Pi$ be such that $\rho_k \in V_i$, and let us prove that $\mu_i\left(\rho_k \rho_{k+1} \dots\right) \geq \nego(\lambda_0)(\rho_k)$.
	
	For any deviation $\sigma'_i$ of $\sigma_{i\|\rho_0 \dots \rho_k}$, by definition of NEs, $\mu_i\left(\< \bsigma_{-i\|\rho_0 \dots \rho_k}, \sigma'_i \>_{\rho_k}\right) \leq \mu_i(\rho)$. Therefore:
	$$\mu_i(\rho) \geq \sup_{\sigma'_i} ~ \mu_i\left(\< \bsigma_{-i\|\rho_0 \dots \rho_k}, \sigma'_i \>_{\rho_k}\right)$$
	hence:
	$$\mu_i(\rho) \geq \inf_{\btau_{-i}} ~\sup_{\tau_i} ~ \mu_i\left(\< \btau_{-i\|\rho_0 \dots \rho_k}, \tau_i \>_{\rho_k}\right)$$
	i.e.:
	$$\mu_i(\rho) \geq \nego(\lambda_0)(\rho_k).$$

	\item Let $\rho$ be a $\nego(\lambda_0)$-consistent play from a state $v_0$. Let us define a strategy profile $\bsigma$ such that $\< \bsigma \>_{v_0} = \rho$, by:
	
	\begin{itemize}
		\item $\< \bsigma \>_{v_0} = \rho$;
		
		\item for all histories of the form $\rho_0 \dots \rho_k v$ with $v \neq \rho_{k+1}$, let $i$ be the player controlling $\rho_k$.
		
		Since the game $G$ is with steady negotiation, the infimum:
		$$\inf_{\btau_{-i} \in \lambda_0\Rat(\rho_k)}~ \sup_{\tau_i}~ \mu_i(\< \btau \>_{\rho_k})$$
		is a minimum. Let $\btau^k_{-i}$ be $\lambda_0$-rational strategy profile from $\rho_k$ realizing that minimum, and let $\tau^k_i$ be some strategy from $\rho_k$ such that $\tau^k_i(\rho_k) = v$. Then, we define:
		$$\< \bsigma_{\|\rho_0 \dots \rho_k v} \>_v = \< \btau^k_{\rho_k v} \>_v;$$
		
		\item for every other history $h$, $\bsigma(h)$ is defined arbitrarily.
	\end{itemize}

	Let us prove that $\bsigma$ is an NE: let $\sigma'_i$ be a deviation of $\sigma_i$, let $\rho' = \< \bsigma_{-i}, \sigma'_i \>_{v_0}$ and let $\rho_0 \dots \rho_k$ be the longest common prefix of $\rho$ and $\rho'$. Let $v = \rho'_{k+1}$.
	
	Then, we have:
	$$\mu_i(\rho') \leq \sup_{\tau_i^k} ~\mu_i\left(\< \btau^k \>_{\rho_k}\right) = \nego(\lambda_0)(\rho_k),$$
	and since $\rho$ is $\lambda_0$-consistent, $\nego(\lambda_0)(\rho_k) \leq \mu_i(\rho)$, hence $\mu_i(\rho') \leq \mu_i(\rho)$.
\end{itemize}
\end{proof}

		\section{Proof of Lemma \ref{lm_least_delta_fixed point}} \label{pf_least_delta_fixed point}

\noindent \textbf{Lemma~\ref{lm_least_delta_fixed point}.} \emph{Let $G$ be a game, and let $\epsilon \geq 0$. The negotiation function has a least $\epsilon$-fixed point.}

\begin{proof}
The following proof is a generalization of a classical proof of Tarski's fixed point theorem.

Let $\Lambda$ be the set of the $\epsilon$-fixed points of the negotiation function. The set $\Lambda$ is not empty, since it contains at least the requirement $v \mapsto +\infty$. Let $\lambda^*$ be the requirement defined by:
$$\lambda^*: v \mapsto \inf_{\lambda \in \Lambda} \lambda(v).$$

For every $\epsilon$-fixed point $\lambda$ of the negotiation function, we have then for each $v$, $\lambda^*(v) \leq \lambda(v)$, and $\nego(\lambda^*)(v) \leq \nego(\lambda)(v)$ since $\nego$ is monotone; and therefore, $\nego(\lambda^*)(v) \leq \lambda(v) + \epsilon$.

As a consequence, we have:
$$\nego(\lambda^*)(v) \leq \inf_{\lambda \in \Lambda} \lambda(v) + \epsilon = \lambda^*(v) + \epsilon.$$
The requirement $\lambda^*$ is an $\epsilon$-fixed point of the negotiation function, and is therefore the least $\epsilon$-fixed point of the negotiation function.
\end{proof}

    \section{Proof of Theorem~\ref{thm_spe}} \label{pf_spe}

\noindent \textbf{Theorem~\ref{thm_spe}.}   \emph{Let $G_{\|v_0}$ be an initialized prefix-independent game, and let $\epsilon \geq 0$.
Let $\lambda^*$ be the least $\epsilon$-fixed point of the negotiation function.
Let $\xi$ be a play starting in $v_0$.
If there exists an $\epsilon$-SPE $\bsigma$ such that $\< \bsigma \>_{v_0} = \xi$, then $\xi$ is $\lambda^*$-consistent.
The converse is true if the game $G$ is with steady negotiation.}

\begin{proof}
    First, let us recall that $\lambda^*$ exists by Lemma \ref{lm_least_delta_fixed point}.
    
    Then, our proof can be decomposed in two lemmas:
    
    \begin{lm} \label{lm_spe_lambda_cons}
    	Let $G_{\|v_0}$ be a well-initialized prefix-independent game, and let $\epsilon \geq 0$.
    	Let $\bsigma$ be an $\epsilon$-SPE in $G_{\|v_0}$.
    	Then, there exists an $\epsilon$-fixed point $\lambda$ of the negotiation function such that for every history $hv$ starting in $v_0$, the play $\< \bsigma_{\|hv} \>_v$ is $\lambda$-consistent.
    \end{lm}
    
    \begin{proof}
        Let us define the requirement $\lambda$ by, for each $i \in \Pi$ and $v \in V_i$:
        $$\lambda(v) = \inf_{hv \in \Hist G_{\|v_0}} \mu_i(\< \bsigma_{\|hv} \>_v).$$
        Note that the set $\left\{ \mu_i(\< \bsigma_{\|hv} \>_v) ~|~ hv \in \Hist G_{\|v_0} \right\}$ is never empty, since the game $G_{\|v_0}$ is well-initialized.
        
        Then, for every history $hv$ starting in $v_0$, the play $\< \bsigma_{\|hv} \>_v$ is $\lambda$-consistent. Let us prove that $\lambda$ is an $\epsilon$-fixed point of $\nego$: let $i \in \Pi$, let $v \in V_i$, and let us assume towards contradiction (since the negotiation function is non-decreasing) that $\nego(\lambda)(v) > \lambda(v) + \epsilon$, that is to say:
        $$\inf_{\btau_{-i} \in \lRat(v)} ~ \sup_{\tau_i} ~ \mu_i(\< \btau \>_v) > \inf_{hv \in \Hist G_{\|v_0}} \mu_i(\< \bsigma_{\|hv} \>_v) + \epsilon.$$
        
        Then, since all the plays generated by the strategy profile $\bsigma$ are $\lambda$-consistent, and therefore since any strategy profile of the form $\bsigma_{-i\|hv}$ is $\lambda$-rational, we have:
        $$\inf_{hv} ~ \sup_{\tau_i} ~ \mu_i(\< \bsigma_{-i\|hv}, \tau_i \>_v) > \inf_{hv} ~ \mu_i(\< \bsigma_{\|hv} \>_v) + \epsilon.$$
        
        Therefore, there exists a history $hv$ such that:
        $$\sup_{\tau_i} ~ \mu_i(\< \bsigma_{-i\|hv}, \tau_i \>_v) > \mu_i(\< \bsigma_{\|hv} \>_v) + \epsilon,$$
        which is impossible if the strategy profile $\bsigma$ is an $\epsilon$-SPE. Therefore, there is no such $v$, and the requirement $\lambda$ is an $\epsilon$-fixed point of the negotiation function.
    \end{proof}

    \begin{lm} \label{lm_existence_spe}
    	Let $G_{\|v_0}$ be a well-initialized prefix-independent game with steady negotiation, and $\epsilon \geq 0$.
    	Let $\lambda$ be an $\epsilon$-fixed point of the function $\nego$.
    	Then, for every $\lambda$-consistent play $\xi$ starting in $v_0$, there exists an $\epsilon$-SPE $\bsigma$ such that $\< \bsigma \>_{v_0} = \xi$.
    \end{lm}
    
    \begin{proof}
 \begin{itemize}
	\item \emph{Particular case: if there exists $v$ such that $\lambda(v) = +\infty$.}
	
	In that case, for each $u$ such that $uv \in E$, if the player controlling $u$ chooses to go to $v$, no $\lambda$-consistent play can be proposed to them from there, hence there is no $\lambda$-rational strategy profile against that player from $u$, and $\nego(\lambda)(u) = +\infty$. Since $\epsilon$ is finite and since $\lambda$ is an $\epsilon$-fixed point of the negotiation function, it follows that $\lambda(u) = +\infty$. Since $G_{\|v_0}$ is well-initialized, we can repeat this argument and show that $\lambda(v_0) = +\infty$; in that case, there is no $\lambda$-consistent play $\xi$ from $u$, and then the proof is done.
	
	Therefore, for the rest of the proof, we assume that for all $v$, we have $\lambda(v) \neq +\infty$. As a consequence, since $\lambda$ is an $\epsilon$-fixed point of the function $\nego$, for all $v$, we have $\nego(\lambda)(v) \neq +\infty$; and so finally, for each such $v$, there exists a $\lambda$-consistent play starting from $v$.

	\item \emph{Preliminary result: a game with steady negotiation is also with subgame-steady negotiation.}
	
	Recall that since $G$ is a game with steady negotiation, for every requirement $\lambda$, for every player $i$ and for every state $v$, there exists a $\lambda$-rational strategy profile $\btau^v$ such that:
	$$\sup_{\tau^v_i} ~ \mu_i(\< \btau^v \>_v) = \inf_{\btau_{-i} \in \lRat(v)} ~\sup_{\tau_i} ~\mu_i(\< \btau \>_v)$$
	i.e. there exists a worst $\lambda$-rational strategy profile against player $i$ from the state $v$, with regards to player $i$'s payoff.
	
	Our goal in this part of the proof is to show that $G$ is then also with \emph{subgame-steady negotiation}, that is to say, for every requirement $\lambda$, for every player $i$ and for every state $v$, there exists a $\lambda$-rational strategy profile $\btau^{v*}_{-i}$ such that for every history $hw$ starting from $v$ compatible with $\btau^{v*}_{-i}$, we have:
	$$\sup_{\tau^{v*}_i} ~ \mu_i(\< \btau^{v*}_{\|hw} \>_w) = \inf_{\btau_{-i} \in \lRat(w)} ~ \sup_{\tau_i} ~ \mu_i(\< \btau \>_w),$$
	i.e. there exists a $\lambda$-rational strategy profile against player $i$ from the state $v$, that is the worst with regards to player $i$'s payoff in any subgame, in other words a \emph{subgame-worst} strategy profile.
	
	Let us construct inductively the strategy profile $\btau^{v*}_{-i}$ and the strategy $\tau^{v*}_i$ assuming which it is $\lambda$-rational. We define them only on histories that are compatible with $\btau^{v*}_{-i}$, since they can be defined arbitrarily on any other histories. We proceed by assembling the strategy profiles of the form $\btau^w$, and the histories after which we follow a new $\btau^w$ will be called the \emph{resets} of $\btau^{v*}_{-i}$.
	
	\begin{itemize}
		\item First, $\< \btau^{v*} \>_v = \< \btau^v \>_v$: the one-state history $v$ is then the first reset of $\btau^{v*}_{-i}$;
		
		\item then, for every history $hw$ from $v$ such that $h$ is compatible with $\btau^{v*}_{-i}$ and ends in $V_i$, and such that $w \neq \tau^{v*}_i(h)$: let us write $hw = h'uh''$ so that $h'u$ is the longest reset of $\btau^{v*}_{-i}$ among the prefixes of $h$, and therefore so that the strategy profile $\btau^{v*}_{\|h'u}$ has been defined as equal to $\btau^u$ over the prefixes of $h''$ until $w$. Then, we have:
		
		$$\sup_{\tau_i} \mu_i(\< \btau^w_{-i}, \tau_i \>_w) \leq \sup_{\tau_i} \mu_i(\< \btau^u_{-i\|uh''}, \tau_i \>_w)$$
		by prefix-independence of $G$ and since by its definition, the strategy profile $\btau^w_{-i}$ minimizes the quantity $\sup_{\tau_i} ~ \mu_i(\< \btau^w_{-i}, \tau_i \>_w)$. Let us separate two cases.
		
		\begin{itemize}
			\item Suppose first that:
			$$\sup_{\tau_i} \mu_i(\< \btau^w_{-i}, \tau_i \>_w) = \sup_{\tau_i} \mu_i(\< \btau^u_{-i\|uh''}, \tau_i \>_w).$$
			Then, $\< \btau^{v*}_{\|hw} \> = \< \btau^u_{\|uh''} \>_w$: the coalition of players against player $i$ keeps following their strategy profile so that player $i$ will have no more than the payoff they can ensure.
			
			\item Suppose now that:
			$$\sup_{\tau_i} \mu_i(\< \btau^w_{-i}, \tau_i \>_w) < \sup_{\tau_i} \mu_i(\< \btau^u_{-i\|uh''}, \tau_i \>_w).$$
			Then, $\< \btau^{v^*}_{\|hw} \> = \< \btau^w \>_w$: player $i$ has done something that lowers the payoff they can ensure, and therefore the other players have to update their strategy profile in order to enforce that new minimum.
			
			The history $hw$ is a reset of $\btau^{v*}_{-i}$.
		\end{itemize}
	
		All the plays constructed are $\lambda$-consistent, hence $\btau^{v*}_{-i}$ is indeed $\lambda$-rational assuming $\tau^{v*}_i$.
		
		Let us now prove that $\tau^{v*}_i$ is the subgame-worst $\lambda$-rational strategy profile against player $i$. Let $hw$ be a history starting in $v$ compatible with $\btau^{v*}_{-i}$, let $\tau'_i$ be a strategy from the state $w$, let $\eta = \< \btau^{v*}_{-i\|hw}, \tau'_i \>_w$ and let us prove that:
		$$\mu_i(\eta) \leq \inf_{\btau_{-i} \in \lRat(w)} ~\sup_{\tau_i} ~\mu_i(\< \btau \>_w).$$
		
		Let us consider the sequence $(\alpha_n)_{n \in \N}$, defined by:
		$$\alpha_n = \inf_{\btau_{-i} \in \lRat(\eta_n)} ~\sup_{\tau_i} ~\mu_i(\< \btau \>_{\eta_n}).$$
		
		That sequence is non-increasing. Indeed, for all $n$:
		
		\begin{itemize}
			\item If $\eta_n \in V_i$, then no action of player $i$ can improve the payoff player $i$ themself can secure against a $\lambda$-rational environment.
			
			\item If $\eta_n \not\in V_i$, then:
			$\eta_{n+1} = \btau^{v*}_{-i}(h \eta_0 \dots \eta_n) = \btau^{\eta_k}_{-i}(\eta_k \dots \eta_n)$
			for some $k$ such that, by construction of $\btau^{v*}_{-i}$, $\alpha_k = \dots = \alpha_n$. Since the strategy profile $\btau^{\eta_k}_{-i}$ is defined to realize the payoff $\alpha_k = \alpha_n$, we have $\alpha_{n+1} = \alpha_n$.
		\end{itemize}
		
		Moreover, that sequence can only take a finite number of values (at most $\card V$). Therefore, it is stationary: there exists $n_0 \in \N$ such that $(\alpha_n)_{n \geq n_0}$ is constant, and there are no resets of $\btau^{v*}_{-i}$ among the prefixes of $\eta$ of length greater than $n_0$.
		
		Therefore, if we choose $n_0$ minimal (i.e., $n_0$ is the index of the last reset in $\eta$), then the play $\eta_{n_0} \eta_{n_0+1} \dots$ is compatible with the strategy profile $\btau^{\eta_{n_0}}_{-i}$. Then, we have:
		$$\mu_i(\eta) \leq \alpha_{n_0} \leq \alpha_0,$$
		and:
		$$\alpha_0 = \inf_{\btau_{-i} \in \lRat(w)} ~\sup_{\tau_i} ~\mu_i(\< \btau \>_w),$$
		which proves that $\btau^{v*}$ is the subgame-worst $\lambda$-strategy profile against player $i$ from the state $w$, and therefore that the game $G$ is a game with subgame-steady negotiation.
	\end{itemize}

	\item \emph{Construction of $\bsigma$.}

	Let $\H_0 = \Hist G_{\|v_0}$. Let us construct inductively $\bsigma$ by defining all the plays $\< \bsigma_{\|hv} \>_v$, for $hv \in \H_0$, keeping the hypothesis that at any step $n$, the set $\H_n$ contains exactly the histories $hv$ such that the play $\< \bsigma_{\|hv} \>_v$ has been defined, and that such a play is always $\lambda$-consistent: it will define a $\lambda$-rational strategy profile, and we will then prove it is an $\epsilon$-SPE.

	\begin{itemize}
		\item First, $\< \bsigma \>_{v_0} = \xi$, which satisfies the induction hypothesis. We remove then all the finite prefixes of $\xi$ form $\H_0$ to obtain $\H_1$. Note that the only history of length $1$ has been removed.

		\item At the $n$-th step, with $n > 0$: let us choose $hv \in \H_n$ of minimal length, and therefore minimal for the prefix order: the strategy profile $\bsigma$ has been defined on all the strict prefixes of $hv$, but not on $hv$ itself, and $v \neq \bsigma(h)$. Let then $i$ be the player controlling the last state of $h$ (which exists since all the histories of $\H_n$ have length at least $2$). Let $\btau^{v*}_{-i}$ be a subgame-worst $\lambda$-rational strategy profile against player $i$ from $v$, whose existence has been proved in the previous point, and let $\tau^{v*}_i$ be a strategy assuming which it is $\lambda$-rational.
		
		Then, we define $\< \bsigma_{\|hv} \>_v = \< \btau^{v*} \>_v$, and inductively, for every history $h'w$ starting from $v$ and compatible with $\bsigma_{-i\|hv}$ as it has been defined so far, we define $\< \bsigma_{\|hh'w} \>_v = \< \btau^{v*}_{\|h'w} \>_w$. The strategy profile $\bsigma_{\|hv}$ is then equal to $\btau^{v*}$ on any history compatible with $\btau^{v*}_{-i}$.
		
		We remove all such histories from $\H_n$ to obtain $\H_{n+1}$. All the plays we built are $\lambda$-consistent, which was our induction hypothesis.
	\end{itemize}

	Since each step removes from $\H_n$ a history of minimal length, and since there are finitely many histories of any given length, we have $\bigcap_n \H_n = \emptyset$, and this process completely defines $\bsigma$.

	\item \emph{Such $\bsigma$ is an $\epsilon$-SPE.}
	
	Let $h^{(0)}w \in \Hist G_{\|v_0}$, let $i \in \Pi$, let $\sigma'_i$ be a deviation of $\sigma_i$. Let $\rho = h^{(0)} \< \bsigma_{\|h^{(0)}w} \>_w$ and let $\rho' = h^{(0)} \< \sigma'_{i\|h^{(0)}w}, \bsigma_{-i\|h^{(0)}w} \>_w$. We prove that $\mu_i(\rho') \leq \mu_i(\rho) + \epsilon$.
	
	If $\rho'$ is compatible with $\sigma_i$, then $\rho' = \rho$ and the proof is immediate. If it is not, we let $huv$ denote the shortest prefix of $\rho'$ such that $u \in V_i$ and $v \neq \sigma_i(hu)$. The transition $uv$ can be considered as the first deviation of player $i$, but note that $hu$ can be both longer or shorter than $h^{(0)}$: player $i$ may have already deviated in $h^{(0)}$.
	
	Be that as it may, the history $hu$ is a common prefix of the play $\rho$ and $\rho'$, and if $\btau^{v*}_{-i}$ denotes a subgame-worst strategy profile against player $i$ from the state $v$, $\lambda$-rational assuming a strategy $\tau^{v*}_i$, then $\bsigma_{\|huv}$ has been defined as equal to $\btau^{v*}$ on any history compatible with $\bsigma_{-i\|huv}$.
	
	\begin{itemize}
		\item If $huv$ is a prefix of $\rho$: let $huh'w'$ be the longest common prefix of $\rho$ and $\rho'$. Necessarily, $w' \in V_i$. Then, by definition of $\btau^{v*}_{-i}$, we have:
		$$\mu_i(\rho') \leq \inf_{\btau_{-i} \in \lRat(w')}~ \sup_{\tau_i} ~\mu_i(\< \btau \>_{w'}) = \nego(\lambda)(w'),$$
		and since $\lambda$ is an $\epsilon$-fixed point of $\nego$:
		$$\mu_i(\rho') \leq \lambda(w') + \epsilon.$$
		
		On the other hand, the play $\< \bsigma_{\|h'w'} \>_{w'}$, which is a suffix of $\rho$, is $\lambda$-consistent, hence $\mu_i(\rho) \geq \lambda(w')$.
		
		Therefore, $\mu_i(\rho') \leq \mu_i(\rho) + \epsilon$.

		\item If $huv$ is not a prefix of $\rho$: then, $\rho = h \< \bsigma_{\|hu} \>_u$. Since $u \in V_i$, we have:
		$$\nego(\lambda)(u) = \sup_{uv' \in E} ~ \inf_{\btau_{-i} \in \lRat(v')} ~ \sup_{\tau_i} ~\mu_i(\< \btau \>_{v'}).$$
		
		In particular, we have:
		$$\inf_{\btau_{-i} \in \lRat(v)} ~ \sup_{\tau_i} ~\mu_i(\< \btau \>_v) \leq \nego(\lambda)(u) \leq \lambda(u) + \epsilon.$$
		
		Then, for the same reason as above, we know that:
		$$\mu_i(\rho') \leq \inf_{\btau_{-i} \in \lRat(v)} ~ \sup_{\tau_i} ~\mu_i(\< \btau \>_v).$$
		
		Finally, since the suffix $\< \bsigma_{\|hu} \>_u$ of $\rho$ is $\lambda$-consistent, we have $\mu_i(\rho) \geq \lambda(u) \geq \nego(\lambda)(u) - \epsilon \geq \mu_i(\rho')$.
	\end{itemize}
	
	The strategy profile $\bsigma$ is an $\epsilon$-SPE.
\end{itemize}
\end{proof}	
    
    If $\bsigma$ is an $\epsilon$-SPE, then by Lemma \ref{lm_spe_lambda_cons}, there exists an $\epsilon$-fixed point $\lambda$ of the negotiation function such that all the plays generated by $\bsigma$ after some history are $\lambda$-consistent; in particular, the play $\xi$ is $\lambda$-consistent, and therefore $\lambda^*$-consistent since $\lambda^* \leq \lambda$.
    
    Conversely, if the game $G$ is with steady negotiation, and if the play $\xi$ is $\lambda^*$-consistent, then by Lemma \ref{lm_existence_spe}, there exists an $\epsilon$-SPE $\bsigma$ such that $\< \bsigma \>_{v_0} = \xi$.
\end{proof}

		\section{Abstract negotiation game} \label{app_abstract}

\begin{defi}[Abstract negotiation game]
	Let $G_{\|v_0}$ be an initialized game, let $i \in \Pi$, and let $\lambda$ be a requirement on $G$. The \emph{abstract negotiation game} of $G_{\|v_0}$ for player $i$ with requirement $\lambda$ is the two-player zero-sum initialized game:
	
	$$\Abs_{\lambda i}(G)_{\|[v_0]} = \left( \{\P, \C\}, S, (S_{\P}, S_{\C}), \Delta, \nu\right)_{\|[v_0]},$$
	
	where:
	
	\begin{itemize}
		\item $\P$ denotes the player \emph{Prover} and $\C$ the player \emph{Challenger};
		
		\item the states of $S_\C$ are written $[\rho]$, where $\rho$ is a $\lambda$-consistent play in $G$;
		
		\item the states of $S_\P$ are written $[hv]$, where $hv$ is a history in $G$, with $h \in \Hist_i(G)$, or $[v]$ with $v \in V$, plus two additional states $\top$ and $\bot$;
		
		\item the set $\Delta$ contains the transitions of the form:
		
		\begin{itemize}
			\item $[hv][v\rho]$, where $[hv] \in S_\P$ and $[v\rho] \in S_{\C}$ (Prover proposes a play);
			
			\item $[\rho] [\rho_0...\rho_n v]$, where $[\rho] \in S_{\C}, n \in \N,	\rho_n \in V_i$, and $v \neq \rho_{n+1}$ (Challenger makes player $i$ deviate);
			
			\item $[\rho] \top$, where $[\rho] \in S_{\C}$ (Challenger accepts the proposed play);
			
			\item $\top \top$ (the game is over);
			
			\item $[hv] \bot$ (Prover has no more play to propose);
			
			\item $\bot \bot$ (the game is over).
		\end{itemize}

		\item $\nu$ is the outcome function defined by, for all $\rho^{(0)}, \rho^{(1)}, \dots, h^{(1)} v_1, h^{(2)} v_2, \dots, k, H$:
	
		$$\begin{matrix}
			& \nu_\C \left( [v_0] \left[\rho^{(0)}\right] \left[h^{(1)}v_1\right] \left[\rho^{(1)}\right] \dots \left[h^{(k)} v_k\right] \left[\rho^{(k)}\right] \top^\omega \right) \\[1mm]
			= & \mu_i\left(h^{(1)} \dots h^{(k)} \rho^{(k)}\right), \\[2mm]
			
			& \nu_\C \left( [v_0] \left[\rho^{(0)}\right] \left[h^{(1)}v_1\right] \left[\rho^{(1)}\right] \dots \left[h^{(n)} v_n\right] \left[\rho^{(n)}\right] \dots \right) \\[1mm]
			= & \mu_i\left(h^{(1)} h^{(2)} \dots\right), \\[2mm]
			
			& \nu_\C\left(H \bot^\omega \right) = +\infty,
		\end{matrix}$$
		
		and by $\nu_\P = -\nu_\C$.
	\end{itemize}
\end{defi}

\begin{pptn}
	Let $G_{\|v_0}$ be an initialized Borel game, let $\lambda$ be a requirement on $G$ and let $i \in \Pi$.
	Then, the corresponding abstract negotiation game satisfies: 
	$$\inf_{\tau_\P} ~\sup_{\tau_\C} ~\nu_\C(\< \btau \>_{[v_0]}) = \inf_{\bsigma_{-i} \in \lRat(v_0)} ~ \sup_{\sigma_i} ~ \mu_i\left(\< \bsigma \>_{v_0}\right).$$
\end{pptn}

\begin{proof}
Let $\alpha \in \R$, and let us prove that the following statements are equivalent:

\begin{enumerate}
	\item \label{case1_pf_abs} there exists a strategy $\tau_\P$ such that for every strategy $\tau_\C$, $\nu_\C(\< \btau \>_{[v_0]}) < \alpha$;
	
	\item \label{case2_pf_abs} there exists a $\lambda$-rational strategy profile $\bsigma_{-i}$ in the game $G_{\|v_0}$ such that for every strategy $\sigma_i$, we have $\mu_i\left( \< \bsigma \>_{v_0} \right) < \alpha$.
\end{enumerate}

\begin{itemize}
	\item \emph{(\ref{case1_pf_abs}) implies (\ref{case2_pf_abs}).}
	
	Let $\tau_\P$ be such that for every strategy $\tau_\C$, $\nu_\C(\< \btau \>_{[v_0]}) < \alpha$.
	
	In what follows, any history $h$ compatible with an already defined strategy profile $\bsigma_{-i}$ in $G_{\|v_0}$ will be decomposed in:
	$$h = v_0 h^{(0)} v_1 h^{(1)} \dots h^{(n-1)} v_n h^{(n)},$$
	so that there exist plays $\rho^{(0)}, \dots, \rho^{(n-1)}, \eta$ and a history:
	$$[v_0] \left[\rho^{(0)}\right] \left[v_1 h^{(1)} v_2\right] \dots \left[v_{n-1}h^{(n-1)}v_n\right] \left[v_nh^{(n)}\eta\right]$$
	in the game $\Abs_{\lambda i}(G)$ compatible with $\tau_{\P}$: the existence and the unicity of that decomposition can be proved by induction. Intuitively, the history $h$ is cut in histories which are prefixes of plays that can be proposed by Prover.
	
	Then, let us define inductively the strategy profile $\bsigma_{-i}$ by, for every $h$ such that $\bsigma_{-i}$ has been defined on the prefixes of $h$, and such that the last state of $h$ is not controlled by player $i$, $\bsigma_{-i}(h) = \eta_0$ with $\eta$ defined from $h$ as higher. Let us prove that $\bsigma_{-i}$ is the desired strategy profile.
	
	\begin{itemize}
		\item \emph{The strategy profile $\bsigma_{-i}$ is $\lambda$-rational.}
		
		Let us define $\sigma_i$ so that for every history $hv$ compatible with $\bsigma_{-i}$, the play $\< \bsigma_{\|hv} \>_v$ is $\lambda$-consistent.
		
		For any history:
		$$h = v_0 h^{(0)} v_1 h^{(1)} \dots h^{(n-1)} v_n h^{(n)}$$
		compatible with $\bsigma_{-i}$ and ending in $V_i$, let $\sigma_i(h) = \eta_0$ with $\eta$ corresponding to the decomposition of $h$, so that by induction:
		$$\< \bsigma_{\|v_0 h^{(0)} v_1 h^{(1)} \dots h^{(n-1)} v_n} \>_{v_n} = v_n h^{(n)} \eta.$$
		
		Let now $hv$ be a history in $G_{\|v_0}$, and let us show that the play $\< \bsigma_{\|hv} \>_v$ is $\lambda$-consistent. If we decompose:
		$$hv = v_0 h^{(0)} v_1 h^{(1)} \dots h^{(n-1)} v_n h^{(n)}$$
		with the same definition of $\eta$ (note that the vertex $v$ is now included in the decomposition), then $\< \bsigma_{\|hv} \>_v = v\eta$, and by definition of the abstract negotiation game, $v_n h^{(n)} \eta$ is a $\lambda$-consistent play, and therefore so is $v \eta$.

		\item \emph{The strategy profile $\bsigma_{-i}$ keeps player $i$'s payoff under the value $\alpha$.}
		
		Let $\sigma_i$ be a strategy for player $i$, and let $\rho = \< \bsigma \>_{v_0}$. We want to prove that $\mu_i(\rho) < \alpha$.
		
		Let us define two finite or infinite sequences $\left( \rho^{(k)} \right)_{k \in K}$ and $\left( h^{(k)} v_k \right)_{k \in K}$, where $K = \{1, \dots, n\}$ or $K = \N \setminus \{0\}$, by for every $k \in K$:
		
		$$\left[ \rho^{(k)} \right] = \tau_{\P} \left( [v_0] \left[ \rho^{(0)} \right] \dots \left[ \rho^{(k-1)} \right] \left[ h^{(k)} v_k \right] \right)$$
		and so that for every $k$, the history $h^{(k)} v_k$ is the shortest prefix of $\rho$ that is not a prefix of $h^{(1)} \dots h^{(k-1)} \rho^{(k-1)}$ (or equivalently, the history $h^{(k)}$ is the longest common prefix of $\rho$ and $h^{(1)} \dots h^{(k-1)} \rho^{(k-1)}$).
		
		Then, the length of the longest common prefix of $h^{(1)} \dots h^{(k-1)} \rho^{(k)}$ and $\rho$ increases with $k$, and the set $K$ is finite if and only if there exists $n$ such that $h^{(1)} \dots h^{(n-1)} \rho^{(n)} = \rho$.
		
		In the infinite case, let:
		$$\chi = [v_0] \left[ \rho^{(0)} \right] \left[ h^{(1)} v_1 \right] \dots \left[ \rho^{(k)} \right] \left[ h^{(k)} v_k \right] \dots.$$
		The play $\chi$ is compatible with $\tau_{\P}$, hence $\nu_\C(\chi) < \alpha$, that is to say:
		$$\mu_i\left(h^{(1)} h^{(2)} \dots\right) < \alpha,$$
		ie. $\mu_i(\rho) < \alpha$.
		
		In the finite case, let:
		$$\chi = [v_0] \left[ \rho^{(0)} \right] \left[ h^{(1)} v_1 \right] \dots \left[ \rho^{(n)} \right] \top^\omega.$$
		For the same reason, $\nu_\C(\chi) < \alpha$, that is to say $\mu_i \left( h^{(1)} \dots h^{(n)} \rho^{(n)} \right) = \mu_i(\rho) < \alpha$.
	\end{itemize}
	
	\item \emph{(\ref{case2_pf_abs}) implies (\ref{case1_pf_abs}).}
	
	Let $\bsigma_{-i}$ be a strategy profile keeping player $i$'s payoff below $\alpha$, $\lambda$-rational assuming a strategy $\sigma_i$.
	Let us define a strategy $\tau_{\P}$ for Prover in the abstract negotiation game.
	
	Let $H = [v_0] \left[\rho^{(0)}\right] \left[h^{(1)}v_1\right] \left[\rho^{(1)}\right] \dots \left[h^{(n)} v_n\right]$ be a history in the abstract game, ending in $S_\P$. Then, we define:
	$$\tau_{\P}(H) = \left[ \< \bsigma_{\|h^{(1)} \dots h^{(n)}v_n} \>_{v_n} \right].$$
	
	If $H$ is a history ending in $\top$, then $\tau_\P(H) = \top$, and in the same way if $H$ ends in $\bot$, then $\tau_\P(H) = \bot$.
	
	Let us show that $\tau_{\P}$ is the strategy we were looking for. Let $\chi$ be a play compatible with $\tau_{\P}$, and let us note that the state $\bot$ does not appear in $\chi$. Then, the play $\chi$ can only have two forms:
	
	\begin{itemize}
		\item If $\chi = [v_0] \left[\rho^{(0)}\right] \left[h^{(1)}v_1\right]  \dots \left[\rho^{(n)}\right] \top^\omega$, then we have:
		$$\rho^{(n)} = \< \bsigma_{\|h^{(1)} \dots h^{(n)} v_n} \>_{v_n},$$
		and the history $h^{(1)} \dots h^{(n)} v_n$ in the game $G_{\|v_0}$ is compatible with $\bsigma_{-i}$. By hypothesis, we have:
		$$\mu_i\left(h^{(1)} \dots h^{(n)} \rho^{(n)}\right) < \alpha,$$
		hence $\nu_\C(\chi) < \alpha$.
		
		\item If $\chi = [v_0] \left[\rho^{(0)}\right]  \dots \left[h^{(n)} v_n\right] \left[\rho^{(n)}\right] \dots$, then the play $\rho = h^{(1)} h^{(2)} \dots$ is compatible with $\bsigma_{-i}$, and by hypothesis $\mu_i(\rho) < \alpha$, hence $\nu_\C(\chi) < \alpha$.
	\end{itemize}
\end{itemize}
\end{proof}

\begin{rek}
    We have proven the equality:
    $$\inf_{\tau_\P} ~\sup_{\tau_\C} ~\nu_\C(\< \btau \>_{\|v_0}) = \inf_{\bsigma_{-i} \in \lRat(v_0)} ~ \sup_{\sigma_i} ~ \mu_i\left(\< \bsigma \>_{v_0}\right).$$
    To be absolutely rigorous, the left member can be written $\val_\C\left( \Abs_{\lambda i}(G)_{\|[v_0]} \right)$ only if we prove that the abstract negotiation game is determined.
    That will be a consequence of its equivalence with the corresponding concrete negotiation game, which is Borel and therefore determined.
\end{rek}

		\section{Proof of Theorem \ref{thm_concrete_game}} \label{pf_concrete_game}

\noindent \textbf{Theorem~\ref{thm_concrete_game}.} \emph{Let $G_{\|v_0}$ be an initialized mean-payoff game. Let $\lambda$ be a requirement and $i$ a player.
Then, we have:}
$$\val_\C\left(\Conc_{\lambda i}(G)_{\|s_0}\right) = \inf_{\bsigma_{-i} \in \lRat(v_0)} ~ \sup_{\sigma_i} ~ \mu_i(\< \bsigma \>_{v_0}).$$

\begin{proof}
First, let us define:
$$A = \left\{\left. \sup_{\sigma_i} ~\mu_i(\< \bsigma \>_{v_0}) ~\right|~ \bsigma_{-i} \in \lambda\Rat(v_0) \right\}$$
and:
$$B = \left\{\left. \sup_{\tau_\C} ~\nu_\C(\< \btau \>_{s_0}) ~\right|~ \tau_\P \right\} \setminus \{+\infty\}.$$

We prove our point if we prove that $A = B$.

\begin{itemize}
	\item \emph{$B \subseteq A$.}
	
	Let $\tau_\P$ be a strategy such that:
	$$\sup_{\tau_\C} ~\nu_\C(\< \btau \>_{s_0}) < +\infty,$$
	and let $\bsigma$ be the strategy profile defined by:
	$$\bsigma(\dH) = w$$
	for every history $H$ compatible with $\tau_\P$ (by induction, the localized projection is injective on the histories compatible with $\tau_\P$) with $\tau_\P(H) = (vw, \cdot)$, and arbitrarily defined on any other histories.
	
	\begin{itemize}
		\item \emph{The strategy profile $\bsigma_{-i}$ is $\lambda$-rational}, assuming the strategy $\sigma_i$. Indeed, let us assume it is not.
		
		Then, there exists a history $h = h_0 \dots h_n$ in $G_{\|v_0}$ compatible with $\bsigma_{-i}$ such that the play $\drho = \< \bsigma_{\|h} \>_{h_n}$ is not $\lambda$-consistent. Then, let:
		$$Hs = \left(h_0, M_0\right) \left(h_0\bsigma(h_0), M_0\right) \dots \left(h_n, M_n\right)$$
		be the only history in $\Conc_{\lambda i}(G)_{\|s_0}$ compatible with $\tau_\P$ such that $\dH = h$.
		
		Let $\tau_\C$ be a strategy constructing the history $h$, defined by:
		$$\tau_\C\left(H_0 \dots H_{2k-1}\right) = H_{2k}$$
		for every $k$, and:
		$$\tau_\C\left(H' (vw, M)\right) = (w, M \cup \{w\})$$
		for any other history $H' (vw, M)$.
		
		Then, the play $\eta = \< \btau \>_{s_0}$ contains finitely many deviations (Challenger stops the deviations after having drawn the history $h$), and the play $\deta = h_0 \dots h_{n-1} \drho$ is not $\lambda$-consistent, i.e. there exists a dimension $j \in \Pi$ such that:
		$$\mu_j(\deta) - \max_{v \in M_n \cap V_j} \lambda(v) < 0$$
		i.e.:
		$$\hmu_j(\eta) < 0$$
		and therefore $\nu_\C(\rho) = \nu_\C(\eta) = +\infty$, which is false by hypothesis.

		\item Now, let us prove the equality:
		$$\sup_{\sigma'_i} ~\mu_i(\< \bsigma_{-i}, \sigma'_i \>_{v_0}) = \sup_{\tau_\C} ~\nu_\C(\< \btau \>_{s_0}).$$
		
		For that purpose, let us prove the equality of sets:
		$$\left\{ \mu_i(\< \bsigma_{-i}, \sigma'_i \>_{v_0}) ~|~ \sigma'_i \right\} = \left\{ \nu_\C(\< \btau \>_{s_0}) ~|~ \tau_\C \right\}.$$
		
		\begin{itemize}
			\item Let $\tau_\C$ be a strategy for Challenger, and let $\rho = \< \btau \>_{s_0}$. Since $\nu_\C(\rho) \neq + \infty$ by hypothesis, we have $\nu_\C(\rho) = \hmu_\star(\rho) = \mu_i(\drho)$, which is an element of the left-hand set.

			\item Conversely, if $\sigma'_i$ is a strategy for player $i$ and if $\eta = \< \bsigma_{-i}, \sigma'_i \>_{v_0}$, let $\tau_\C$ be a strategy such that for every $k$:
			$$\tau_\C\left((\eta_0, \cdot)(\eta_0 \cdot, \cdot) \dots (\eta_k \cdot, \cdot) = (\eta_{k+1}, \cdot)\right),$$
			i.e. a strategy forcing $\eta$.
			
			Then, since $\nu_\C(\rho) \neq +\infty$ by hypothesis on $\tau_\P$, we have $\mu_i(\eta) = \nu_\C(\rho)$, which is an element of the right-hand set.
		\end{itemize}
	\end{itemize}

	\item \emph{$A \subseteq B$.}
	
	Let $\bsigma_{-i}$ be a $\lambda$-rational strategy profile from $v_0$, assuming the strategy $\sigma_i$; let us define a strategy $\tau_\P$ by, for every history $H$ and for every $v \in V$:
	$$\tau_\P(H(v, \cdot)) = \left(v\bsigma(\dH v), \cdot\right).$$
	
	Let us prove the equality:
	$$\sup_{\sigma'_i} ~\mu_i(\< \bsigma_{-i}, \sigma'_i \>_{v_0}) = \sup_{\tau_\C} ~\nu_\C(\< \btau \>_{s_0}).$$
	
	For that purpose, let us prove the equality of sets:
	$$\left\{ \mu_i(\< \bsigma_{-i}, \sigma'_i \>_{v_0}) ~|~ \sigma'_i \right\} = \left\{ \nu_\C(\< \btau \>_{s_0}) ~|~ \tau_\C \right\}.$$
	
	\begin{itemize}
		\item Let $\tau_\C$ be a strategy for Challenger, and let $\rho = \< \btau \>_{s_0}$.
		
		If $\nu_\C(\rho) = +\infty$, then $\drho$ is compatible with $\bsigma$ and not $\lambda$-consistent after finitely many steps, which is impossible.
		
		Therefore, $\nu_\C(\< \btau \>_{s_0}) \neq + \infty$, and as a consequence we have $\nu_\C(\rho) = \hmu_\star(\rho) = \mu_i(\drho)$, which is an element of the left-hand set.

		\item Conversely, if $\sigma'_i$ is a strategy for player $i$ and if $\eta = \< \bsigma_{-i}, \sigma'_i \>_{v_0}$, let $\tau_\C$ be a strategy such that for all $k$:
		$$\tau_\C\left((\eta_0, \cdot)(\eta_0 \cdot, \cdot) \dots (\eta_k \cdot, \cdot)\right) = (\eta_{k+1}, \cdot),$$
		i.e. a strategy forcing $\eta$.
		
		Then, either $\nu_\C(\rho) = +\infty$, and therefore $\eta$ is not $\lambda$-consistent, and is compatible with $\bsigma$ after finitely many steps, which is impossible.
		
		Or, $\mu_i(\eta) = \nu_\C(\rho)$, which is an element of the right-hand set.
	\end{itemize}
\end{itemize}
\end{proof}

    \section{An example of concrete negotiation game} \label{ex_concrete_game}

Let us consider again the game from Example~\ref{ex_sans_spe}.

Figure~\ref{fig_concrete} represents the game $\Conc_{\lambda_1 \playcircle}(G)$ (with $\lambda_1(a) = 1$ and $\lambda_1(b) = 2$), where the dashed states are controlled by Challenger, and the other ones by Prover.

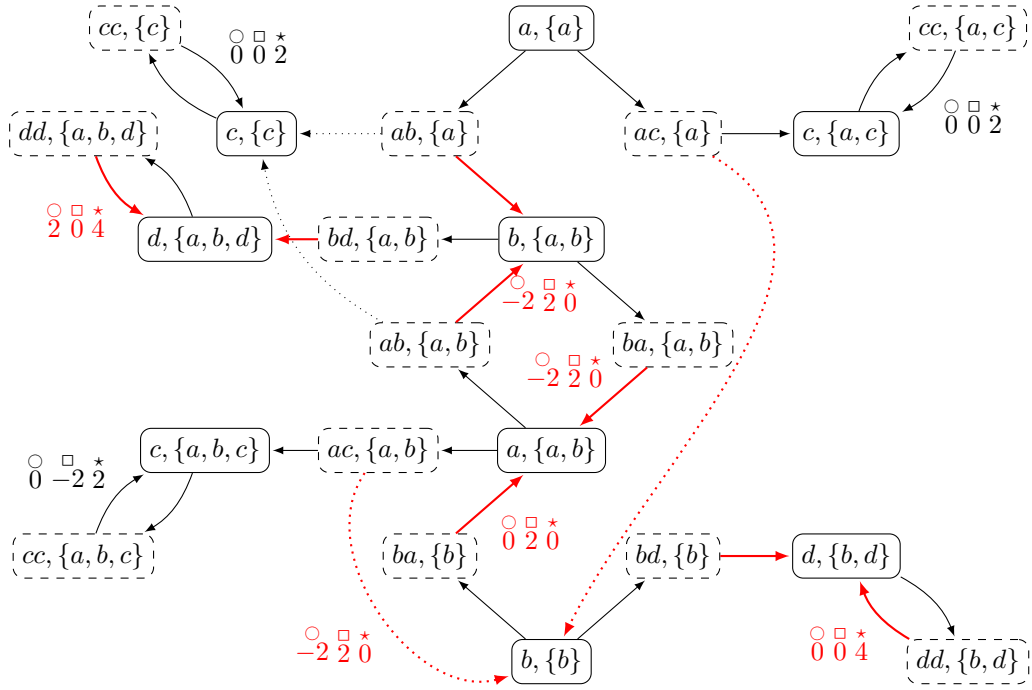
\begin{figure} 
\begin{center}
	\begin{tikzpicture}[->, >=latex,shorten >=1pt, initial text={}, scale=0.4, every node/.style={scale=1}]
	\newcommand{\deltaX}{4}
	\newcommand{\deltaY}{3.5}
	\newcommand{\DeltaX}{5.7}
	\newcommand{\DeltaY}{4.5}
	
	\node[draw, rounded corners] (b-ab) at (0, \deltaY) {$b, \{a, b\}$};
	\node[draw, rounded corners, dashed] (ba-ab) at (\deltaX, 0) {$ba, \{a, b\}$};
	\node[draw, rounded corners] (a-ab) at (0, -\deltaY) {$a, \{a, b\}$};
	\node[draw, rounded corners, dashed] (ab-ab) at (-\deltaX, 0) {$ab, \{a, b\}$};
	
	\node[draw, rounded corners] (a-a) at (0, 3*\deltaY) {$a, \{a\}$};
	\node[draw, rounded corners, dashed] (ab-a) at (-\deltaX, 2*\deltaY) {$ab, \{a\}$};
	\node[draw, rounded corners] (c-c) at (-\deltaX - \DeltaX, 2*\deltaY) {$c, \{c\}$};
	\node[draw, rounded corners, dashed] (cc-c) at (-2*\deltaX - \DeltaX, 3*\deltaY) {$cc, \{c\}$};
	\node[draw, rounded corners, dashed] (ac-a) at (\deltaX, 2*\deltaY) {$ac, \{a\}$};
	\node[draw, rounded corners] (c-ac) at (\deltaX + \DeltaX, 2*\deltaY) {$c, \{a, c\}$};
	\node[draw, rounded corners, dashed] (cc-ac) at (2*\deltaX + \DeltaX, 3*\deltaY) {$cc, \{a, c\}$};
	
	\node[draw, rounded corners] (b-b) at (0, -3*\deltaY) {$b, \{b\}$};
	\node[draw, rounded corners, dashed] (ba-b) at (-\deltaX, -2*\deltaY) {$ba, \{b\}$};
	\node[draw, rounded corners, dashed] (bd-b) at (\deltaX, -2*\deltaY) {$bd, \{b\}$};
	\node[draw, rounded corners] (d-bd) at (\deltaX + \DeltaX, -2*\deltaY) {$d, \{b, d\}$};
	\node[draw, rounded corners, dashed] (dd-bd) at (2*\deltaX + \DeltaX, -3*\deltaY) {$dd, \{b, d\}$};
	
	\node[draw, rounded corners, dashed] (ac-ab) at (-\DeltaX, -\deltaY) {$ac, \{a, b\}$};
	\node[draw, rounded corners] (c-abc) at (-2*\DeltaX, -\deltaY) {$c, \{a, b, c\}$};
	\node[draw, rounded corners, dashed] (cc-abc) at (-2*\DeltaX - \deltaX, -2*\deltaY) {$cc, \{a, b, c\}$};
	
	\node[draw, rounded corners, dashed] (bd-ab) at (-\DeltaX, \deltaY) {$bd, \{a, b\}$};
	\node[draw, rounded corners] (d-abd) at (-2*\DeltaX, \deltaY) {$d, \{a, b, d\}$};
	\node[draw, rounded corners, dashed] (dd-abd) at (-2*\DeltaX - \deltaX, 2*\deltaY) {$dd, \{a, b, d\}$};

	\path (b-ab) edge (ba-ab);
	\path[thick, red] (ba-ab) edge node[above left] {$\scriptsize{\stackrel{\playcircle}{-2} ~ \stackrel{\Box}{2} ~ \stackrel{\star}{0}}$} (a-ab);
	\path (a-ab) edge (ab-ab);
	\path (ab-ab) edge[thick, red] node[right] {$\scriptsize{\stackrel{\playcircle}{-2}~\stackrel{\Box}{2} ~ \stackrel{\star}{0}}$} (b-ab);
	
	\path (a-a) edge (ab-a);
	\path[dotted] (ab-a) edge (c-c);
	\path[bend left=20] (c-c) edge (cc-c);
	\path[bend left=20] (cc-c) edge node[above right] {$\scriptsize{\stackrel{\playcircle}{0}~\stackrel{\Box}{0} ~ \stackrel{\star}{2}}$} (c-c);
	\path (a-a) edge (ac-a);
	\path (ac-a) edge (c-ac);
	\path[bend left=20] (c-ac) edge (cc-ac);
	\path[bend left=20] (cc-ac) edge node[below right] {$\scriptsize{\stackrel{\playcircle}{0}~\stackrel{\Box}{0} ~ \stackrel{\star}{2}}$} (c-ac);
	
	\path (b-b) edge (ba-b);
	\path (b-b) edge (bd-b);
	\path[thick, red] (bd-b) edge (d-bd);
	\path[bend left=20] (d-bd) edge (dd-bd);
	\path[bend left=20, thick, red] (dd-bd) edge node[below left] {$\scriptsize{
			\stackrel{\playcircle}{0}~\stackrel{\Box}{0} ~ \stackrel{\star}{4}}$} (d-bd);
	
	\path (a-ab) edge (ac-ab);
	\path (ac-ab) edge (c-abc);
	\path[bend left=20] (c-abc) edge (cc-abc);
	\path[bend left=20] (cc-abc) edge node[above left] {$\scriptsize{
			\stackrel{\playcircle}{0}~\stackrel{\Box}{-2} ~ \stackrel{\star}{2}
			}$} (c-abc);
	
	\path (b-ab) edge (bd-ab);
	\path[thick, red] (bd-ab) edge (d-abd);
	\path[bend right=20] (d-abd) edge (dd-abd);
	\path[bend right=20, thick, red] (dd-abd) edge node[below left] {$\scriptsize{
			\stackrel{\playcircle}{2}~\stackrel{\Box}{0} ~ \stackrel{\star}{4}}$} (d-abd);
	
	\path[out=-30, in=64, dotted, thick, red] (ac-a) edge (b-b);
	\path[dotted, bend left=25] (ab-ab) edge (c-c);
	\path[thick, red] (ab-a) edge (b-ab);
	\path[thick, red] (ba-b) edge node[below right] {$\scriptsize{
			\stackrel{\playcircle}{0}~\stackrel{\Box}{2} ~ \stackrel{\star}{0}
		}$} (a-ab);
	\path[dotted, bend right=70, thick, red] (ac-ab) edge node[below left] {$\scriptsize{
			\stackrel{\playcircle}{-2}~\stackrel{\Box}{2} ~ \stackrel{\star}{0}
		}$} (b-b);
	\end{tikzpicture}
\end{center}
\caption{A concrete negotiation game} \label{fig_concrete}
\end{figure}
The dotted arrows indicate the deviations, and the transitions that are not labelled are either zero for the three coordinates, or meaningless since they cannot be used more than once.

The red arrows indicate a (memoryless) optimal strategy for Challenger. Against that strategy, the lowest outcome Prover can ensure is $2$.

Therefore, $\nego(\lambda_1)(v_0) = 2$, in line with the abstract game in Example~\ref{ex_abstract_game}.

	\section{Proof of Lemma \ref{lm_memoryless}} \label{pf_memoryless}

\noindent \textbf{Lemma~\ref{lm_memoryless}.} \emph{Let $G_{\|v_0}$ be an initialized mean-payoff game, let $i$ be a player, let $\lambda$ be a requirement and let $\Conc_{\lambda i}(G)_{\|s_0}$ be the corresponding concrete negotiation game. There exists a memoryless strategy $\tau_\C$ that is optimal for Challenger, i.e. such that:}
$$\underset{\tau_\P}{\inf} ~\nu_\C(\< \btau \>_{s_0}) = \val_\C\left( \Conc_{\lambda i}(G)_{\|s_0} \right).$$

\begin{proof}
The structure of that proof is inspired from the proof of Lemma~14 in \cite{DBLP:journals/iandc/VelnerC0HRR15}.

Let $\alpha \in \R$, and let $\Phi$ be the set of the plays $\rho$ in $\Conc_{\lambda i}(G)$ such that:

\begin{itemize}
	\item $\underset{n \to \infty}{\liminf} \frac{1}{n} \underset{k=0}{\overset{n-1}{\sum}} \left(- \hpi_\star(\rho_k\rho_{k+1})\right) \geq -\alpha$;
	
	\item and either:
	
	\begin{itemize}
		\item $\rho$ contains infinitely many deviations;
		
		\item or for each $j \in \Pi$, $\hmu_j(\rho) \geq 0$.
	\end{itemize}
\end{itemize}

Note that the set of the plays $\rho$ such that $\nu_\C(\rho) \leq \alpha$ could be defined almost the same way, but with a limit superior instead of the limit inferior.

By \cite{DBLP:conf/icalp/Kopczynski06}, if Challenger can falsify the objective $\Phi$, he can falsify it with a memoryless strategy, if $\Phi$ is \emph{prefix-independent} and \emph{convex}.

\emph{Convex} objectives are defined as follows: the objective $\Phi$ is convex if for all $\rho, \eta \in \Phi$ and for any decomposition:
$$\rho_0 \dots \rho_{k_1} \dots \rho_{k_2} \dots$$
and:
$$\eta_0 \dots \eta_{\l_1} \dots \eta_{\l_2} \dots$$
with $\eta_{\l_p} = \rho_{k_q}$ for all $p, q \in \N$, we have:
$$\chi = \rho_0 \dots \rho_{k_1} \eta_1 \dots \eta_{\l_1} \rho_{k_1+1} \dots \rho_{k_2} \eta_{\l_1+1} \dots \in \Phi.$$
Let then be such two plays and decomposition, and let us prove that $\chi \in \Phi$.

Let us write $\Phi = \Psi \cap (\mathrm{X} \cup \Xi)$, where:

\begin{itemize}
	\item $\Psi$ is the set of the plays $\rho$ such that:
	$$\underset{n \to \infty}{\liminf} \frac{1}{n} \underset{k=0}{\overset{n-1}{\sum}} \left(- \hpi_\star(\rho_k\rho_{k+1})\right) \geq -\alpha;$$
	
	\item $\mathrm{X}$ is the set of the plays containing infinitely many deviations;
	
	\item $\Xi$ is the set of the plays $\rho$ such that for each $j \in \Pi$, $\hmu_j(\rho) \geq 0$.
\end{itemize}

As shown in \cite{DBLP:journals/iandc/VelnerC0HRR15}, a mean-payoff objective defined with a limit inferior is convex: therefore, we can already say that $\chi \in \Psi$. Let us now prove that $\chi \in \mathrm{X} \cup \Xi$.

\begin{itemize}
	\item \emph{If $\rho \in \mathrm{X}$ or $\eta \in \mathrm{X}$.}
	
	Then, $\chi$ contains the deviations of $\rho$ and $\eta$, hence $\chi \in \mathrm{X}$.

	\item \emph{If $\rho, \eta \in \Xi$.}
	
	Then, since mean-payoff objectives are convex, we have $\chi \in \Xi$.
\end{itemize}

In both cases, $\chi \in \mathrm{X} \cup \Xi$, so $\chi \in \Phi$: the objective $\Phi$ is convex.

Therefore, if Challenger has some strategy to falsify the objective $\Phi$, he has a memoryless one: let us write it $\tau_\C$.

Now, we want to prove that the memoryless strategy $\tau_\C$ is also efficient when we replace the limit inferior of the definition of $\Phi$ by a limit superior, even though this new objective is no longer convex.

By definition of $\tau_\C$, for every strategy $\tau_\P$, we have $\< \btau \>_{s_0} \not\in \Phi$.
Let us prove that $\nu_\C(\< \btau \>_{s_0}) > \alpha$.

In other words, let us prove that for every infinite path $\rho$ from $s_0$ in the graph $\Conc_{\lambda i}(G)[\tau_\C]$, we have $\nu_\C(\rho) > \alpha$.
Since $\rho \not\in \Phi$, we have either $\rho \not\in \mathrm{X} \cup \Xi$ or $\rho \not\in \Psi$.
In the first case, we have $\nu_\C(\rho) = +\infty$, which ends the proof.
In the second case, we have:
$$\underset{n \to \infty}{\limsup} \frac{1}{n} \underset{k=0}{\overset{n-1}{\sum}}  \hpi_\star(\rho_k\rho_{k+1}) > \alpha.$$

We want to prove that $\nu_\C(\rho) > \alpha$, that is, since we assume $\rho \in \mathrm{X} \cup \Xi$:
$$\hmu_\star(\rho) = \underset{n \to \infty}{\liminf} \frac{1}{n} \underset{k=0}{\overset{n-1}{\sum}}  \hpi_\star(\rho_k\rho_{k+1}) > \alpha.$$

Here, the play $\rho$ is an infinite path in the graph $\Conc_{\lambda i}(G)[\tau_\C]$: by the description of the possible outcomes in a mean-payoff game given in \cite{DBLP:conf/concur/ChatterjeeDEHR10}, the mean-payoff $\hmu_\star(\rho)$ is then larger than or equal to the minimal mean-payoff $\hmu_\star$ we get by looping on a simple cycle $c$ of that graph accessible from the state $s$. Intuitively, a play can be seen as a combination of those cycles. That is to say:
$$\hmu_\star(\rho) \geq \underset{\tiny{\begin{matrix}
		c \in \SC\left(\Conc_{\lambda i}(G)[\tau_\C]\right) \\
		\mathrm{accessible~from~} s
		\end{matrix}}}{\min} \hmu_\star(c^\omega).$$

For each such cycle, since $c^\omega$ is a play compatible with $\tau_\C$, we have:
$$\underset{n \to \infty}{\limsup} \frac{1}{n} \underset{k=0}{\overset{n-1}{\sum}} \hpi_\star(c_k c_{k+1}) > \alpha$$
where the indices are taken in $\Z / |c| \Z$, i.e.:
$$\underset{n \to \infty}{\lim} \frac{1}{n} \underset{k=0}{\overset{n-1}{\sum}} \hpi_\star(c_k c_{k+1}) > \alpha,$$
and therefore:
$$\underset{n \to \infty}{\liminf} \frac{1}{n} \underset{k=0}{\overset{n-1}{\sum}} \hpi_\star(c_k c_{k+1}) > \alpha,$$
that is to say:
$$\hmu_\star(c^\omega) > \alpha,$$
hence $\hmu_\star(\rho) > \alpha$.
\end{proof}

	\section{Proof of Lemma \ref{lm_resolution_concrete_game}} \label{pf_resolution_concrete_game}

\noindent \textbf{Lemma~\ref{lm_resolution_concrete_game}.} \emph{Let $G_{\|v_0}$ be an initialized mean-payoff game, and let $\Conc_{\lambda i}(G)_{\|s_0}$ be its concrete negotiation game for some $\lambda$ and some $i$.}
	
\emph{Then, the value of the game $\Conc_{\lambda i}(G)_{\|s_0}$ is given by the formula:}
	$$\max_{\tau_\C \in \ML_\C\left(\Conc_{\lambda i}(G)\right)} ~
	\min_{\scriptsize{\begin{matrix}
			K \in \SConn\left(\Conc_{\lambda i}(G)[\tau_\C]\right) \\
			\mathrm{accessible~from~} s_0
			\end{matrix}}} \opt(K),$$
\emph{where $\opt(K)$ is the minimal value $\nu_\C(\rho)$ for $\rho$ among the infinite paths in $K$.}
	
\begin{itemize}
	\item \emph{If $K$ contains a deviation, then Prover can simply choose the simple cycle of $K$ that minimizes player $i$'s payoff:}
	$$\opt(K) = \underset{c \in \SC(K)}{\min} ~ \hmu_\star(c^\omega).$$
		
	\item \emph{If $K$ does not contain a deviation, then Prover must choose a combination of the simple cycles of $K$ that minimizes player $i$'s payoff while keeping the non-main dimensions above $0$:}
	$$\opt(K) = \minstar \underset{c \in \SC(K)}{\Conv} ~ \hmu(c^\omega).$$
	\end{itemize}

\begin{proof}
By Lemma \ref{lm_memoryless}, there exists a memoryless strategy $\tau_\C$ which is optimal for Challenger among all his possible strategies.

It follows from Theorem \ref{thm_concrete_game} that the highest value player $i$ can get against a hostile $\lambda$-rational environment is the minimal payoff of Challenger in a path in the graph $\Conc_{\lambda i}(G)[\tau_\C]$. For any such path $\rho$, there exists a strongly connected component $K$ of $\Conc_{\lambda i}(G)[\tau_\C]$ accessible from $s_0$ such that after a finite number of steps, $\rho$ is a path in $K$. The least payoff of Challenger in such a path, for a given $K$, is $\opt(K)$; let us prove that it is given by the desired formula.

There are, then, two cases to distinguish:

\begin{itemize}
	\item \emph{If there is at least one deviation in $K$.}
	
	Then, for every play $\rho$ in $K$, it is possible to transform $\rho$ into a play $\rho'$ with $\hmu(\rho') = \hmu(\rho)$, which contains infinitely many deviations: it suffices to add round trips to a deviation, endlessly, but less and less often.
	Therefore, the outcomes $\nu_\C(\rho)$ of plays in $K$ are exactly the mean-payoffs $\hmu_\star(\rho)$ of plays in $K$, and possibly $+\infty$; and in particular, the lowest outcome Prover can get in $K$ is the quantity:
	$$\underset{c \in \SC(K)}{\min} ~ \hmu_\star(c^\omega),$$
	the least value of a simple cycle in $K$.

	\item \emph{If there is no deviation in $K$.}
	
	Let us first introduce a notation: for any finite set $D$ and any set $X \subseteq \R^D$, $X^\llcorner$ denotes the set:
	$$X^\llcorner = \left\{\left. \left( \min_{\by \in Y} y_d \right)_{d \in D} ~\right|~ Y \subseteq X \mathrm{~finite} \right\}.$$
	
	For example, in $\R^2$, if $X$ is the blue area in Figure~\ref{fig_corner}, then $X^\llcorner$ is the union of the blue area and the gray area.

	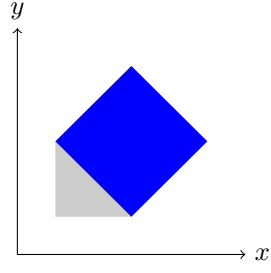
\begin{figure} 
	\begin{center}
		\begin{tikzpicture}[scale=1]
		\draw [->] (0,0) -- (3,0);
		\draw (3,0) node[right] {$x$};
		\draw [->] (0,0) -- (0,3);
		\draw (0,3) node[above] {$y$};
		\fill [gray!40] (0.5, 0.5) -- (1.5, 0.5) -- (0.5, 1.5);
		\fill [blue] (0.5, 1.5) -- (1.5, 2.5) -- (2.5, 1.5) -- (1.5, 0.5);
		\end{tikzpicture}
	\end{center}
	\caption{An example for the operator $\cdot^\llcorner$} \label{fig_corner}
	\end{figure}

	Let us already note that for all $X \in \R^{\Pi \cup \{\star\}}$,
	$$\begin{matrix}
		& \minstar X^\llcorner \\
		
		=& \min \left\{ x_\star ~\left|~ \begin{matrix}
			\bx \in X^\llcorner,\\
			\forall j \in \Pi, x_j \geq 0
		\end{matrix} \right. \right\} \\
		
		=& \min \left\{ \underset{\by \in Y}{\min} ~ y_\star ~\left|~ \begin{matrix}
			Y \subseteq X \mathrm{~finite}, \\
			\forall \by \in Y, \forall j \in \Pi, y_j \geq 0
		\end{matrix} \right. \right\} \\
		
		=& \min \left\{ y_\star ~\left|~ \begin{matrix}
			\by \in X, \\
			\forall \by \in Y, \forall j \in \Pi, y_j \geq 0
		\end{matrix} \right. \right\} \\
		
		=& \minstar X.
	\end{matrix}$$
	
	Then, it has been proved in \cite{DBLP:conf/concur/ChatterjeeDEHR10} that the set of possible values of $\hmu(\rho)$ for all plays $\rho$ in $K$ is exactly the set:
	$$X = \left( \underset{c \in \SC(K)}{\Conv} \hmu(c^\omega) \right)^\llcorner.$$
	
	Since all the plays in $K$ contain finitely many deviations (actually none), for every $\bx = \hmu(\rho) \in X$, we have $\nu_\C(\rho) = +\infty$ if and only if there exists $j \in \Pi$ such that $x_j < 0$. Then, the lowest outcome Prover can get in $K$ is:
	$$\min \left\{ x_\star ~|~ \bx \in X, \forall j \in \Pi, x_j \geq 0 \right\},$$
	that is to say:
	$$\minstar \left( \underset{c \in \SC(K)}{\Conv} \hmu(c^\omega) \right)^\llcorner,$$
	i.e. $\minstar \underset{c \in \SC(K)}{\Conv} \hmu(c^\omega)$.
\end{itemize}

Theorem \ref{thm_concrete_game} enables to conclude to the desired formula.
\end{proof}

		\section{The negotiation sequence} \label{app_nego_seq}

We assume in that appendix that $G$ is a game on which the negotiation function is \emph{Scott-continuous}, i.e. such that for every non-decreasing sequence $(\lambda_n)_n$ of requirements on $G$, we have:
$$\nego\left(\sup_n \lambda_n \right) = \sup_n \nego(\lambda_n).$$

By Kleene-Tarski fixed-point theorem, the least fixed point of the negotiation function is, then, the limit of the \emph{negotiation sequence}, defined as the sequence $(\lambda_n)_{n \in \N} = (\nego^n(\lambda_0))_n$.

In mean-payoff games, in particular, the hypothesis made above is true:

\begin{pptn} \label{pptn_cont_nego}
    In mean-payoff games, the negotiation function is Scott-continuous.
\end{pptn}

A proof of that statement is given in Appendix~\ref{pf_cont_nego}.

In many cases, the negotiation sequence is stationary, and in such a case, it is possible to compute its limit: whenever a term is equal to the previous one, we know that we reached it.
But actually, the negotiation sequence is not always stationary.
The game of Figure~\ref{fig_not_stationary} is a counter-example. Indeed, for all $n$, we have:
$$\lambda_n(a) = \lambda_n(b) = 2 - \frac{1}{2^{n-1}},$$
which converges to $2$ but never reaches it.

\begin{figure} 
\begin{center}
	\begin{tikzpicture}[<->,>=latex,shorten >=1pt, , scale=0.8, every node/.style={scale=0.8}]
	\node[state] (c) at (-1.7, -1) {$c$};
	\node[state, diamond] (d) at (-1.7, 1) {$d$};
	\node[state] (a) at (0, 0) {$a$};
	\node[state, rectangle] (b) at (2, 0) {$b$};
	\node[state, rectangle] (e) at (3.7, 1) {$e$};
	\node[state, diamond] (f) at (3.7, -1) {$f$};
	
	\path[<->] (a) edge node[above] {$\stackrel{\playcircle}{2} \stackrel{\Box}{2} \stackrel{\Diamond}{0}$} (b);
	\path[->] (a) edge (c);
	\path[<->] (c) edge node[left] {$\stackrel{\playcircle}{0} \stackrel{\Box}{1} \stackrel{\Diamond}{0}$} (d);
	\path (d) edge [loop above] node {$\stackrel{\playcircle}{2} \stackrel{\Box}{2} \stackrel{\Diamond}{0}$} (d);
	\path[->] (b) edge (e);
	\path[<->] (e) edge node[right] {$\stackrel{\playcircle}{1} \stackrel{\Box}{0} \stackrel{\Diamond}{0}$} (f);
	\path (f) edge [loop below] node {$\stackrel{\playcircle}{2} \stackrel{\Box}{2} \stackrel{\Diamond}{0}$} (f);
	\end{tikzpicture}
\end{center}
\caption{A game where the negotiation sequence is not stationary} \label{fig_not_stationary}
\end{figure}
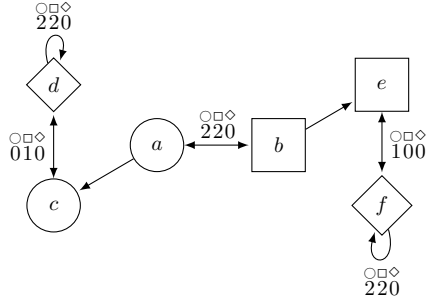

Let us give some details.
Since all the $\Diamond$ weights are equal to $0$, for all $n > 0$, we have $\lambda_n(d) = \lambda_n(f) = 0$.
It comes that for all $n > 0$, we also have $\lambda_n(c) = \lambda_n(e) = 0$.
Moreover, by symmetry of the game, we always have $\lambda_n(a) = \lambda_n(b)$. Therefore, to compute the negotiation sequence, it suffices to compute $\lambda_{n+1}(a)$ as a function of $\lambda_n(b)$, knowing that $\lambda_1(a) = \lambda_1(b) = 1$, and therefore that for all $n > 0$,  $\lambda_n(a) = \lambda_n(b) \geq 1$.

From $a$, the worst play that player $\Box$ could propose to player $\Circle$ would be a combination of the cycles $cd$ and $d$ giving her exactly $1$.
But then, player $\Circle$ will deviate to go to $b$, from which if player $\Box$ proposes plays in the strongly connected component containing $c$ and $d$, then player $\Circle$ will always deviate and generate the play $(ab)^\omega$, and then get the payoff $2$.

Then, in order to give her a payoff lower than $2$, player $\Box$ has to go to the state $e$.
Since player $\Circle$ does not control any state in that strongly connected component, the play he will propose will be accepted: he will, then, propose the worst possible combination of the cycles $ef$ and $f$ for player $\Circle$, such that he gets at least his requirement $\lambda_n(b)$.
The payoff $\lambda_{n+1}(a)$ is then the minimal solution of the system:
$$\left\{\begin{matrix}
\lambda_{n+1}(a) = x + 2(1-x) \\
2(1-x) \geq \lambda_n(b) \\
0 \leq x \leq 1
\end{matrix}\right.$$
that is to say $\lambda_{n+1}(a) = 1 + \frac{\lambda_n(b)}{2} = 1 + \frac{\lambda_n(a)}{2}$, and by induction, for all $n > 0$:
$$\lambda_n(a) = \lambda_n(b) = 2 - \frac{1}{2^{n-1}}.$$

	\section{Proof of Proposition \ref{pptn_cont_nego}} \label{pf_cont_nego}

\noindent \textbf{Proposition~\ref{pptn_cont_nego}.} \emph{In mean-payoff games, the negotiation function is Scott-continuous.}

\begin{proof}
Let $(\lambda_n)_n$ be a non-decreasing sequence of requirements on a mean-payoff game $G$, and let $\lambda = \sup_n \lambda_n$. We want to prove that $\nego(\lambda) = \sup_n \nego(\lambda_n)$.

Since the negotiation function is monotone, we already have $\nego(\lambda) \geq \sup_n \nego(\lambda_n)$. Let us prove that $\nego(\lambda) \leq \sup_n \nego(\lambda_n)$.

Let $\delta > 0$: we want to find $n$ such that $\nego(\lambda_n)(v) \geq \nego(\lambda)(v) - \delta$ for each $v \in V$.

Let:
$$\Conc_{\lambda i}(G)_{\|s_0} = \left(\{\P, \C\}, S, (S_\P, S_\C), \Delta, \nu\right)_{\|s_0}$$
be the concrete negotiation game of $G$ for $\lambda$ and player $i$ controlling $v$, and let:
$$\Conc_{\lambda_n i}(G)_{\|s_0} = \left(\{\P, \C\}, S, (S_\P, S_\C), \Delta, \nu'\right)_{\|s_0}$$
be the concrete negotiation game of $G$ for some requirement $\lambda_n$ in $v$. Let us note that both have the same underlying graph, and that the only difference are the weight functions $\hpi$ and $\hpi'$, on the non-main dimensions.

By Lemma \ref{lm_resolution_concrete_game}, we have:
$$\nego(\lambda)(v) = \max_{\tau_\C \in \ML_\C\left(\Conc_{\lambda i}(G)_{\|s_0}\right)} ~
\min_{\scriptsize{\begin{matrix}
		K \in \SConn\left(\Conc_{\lambda i}(G)[\tau_\C]\right) \\
		\mathrm{accessible~from~} s_0
		\end{matrix}}} \opt(K)$$
with:
$$\opt(K) = \left\{\begin{matrix}
    \mathrm{if~} K \mathrm{~contains~a~deviation}:\\
    \underset{c \in \SC(K)}{\min} ~ \hmu_\star(c^\omega) \\[5mm]
    \mathrm{otherwise}:\\
    \minstar \underset{c \in \SC(K)}{\Conv} ~ \hmu(c^\omega),
\end{matrix} \right.$$
and identically:
$$\nego(\lambda_n)(v) =
\max_{\tau_\C \in \ML_\C\left(\Conc_{\lambda_n i}(G)_{\|s_0}\right)} ~
\min_{\scriptsize{\begin{matrix}
		K \in \SConn\left(\Conc_{\lambda i}(G)[\tau_\C]\right) \\
		\mathrm{accessible~from~} s_0
		\end{matrix}}} \opt'(K)$$
with:
$$\opt'(K) = \left\{\begin{matrix}
\mathrm{if~} K \mathrm{~contains~a~deviation}:\\
\underset{c \in \SC(K)}{\min} ~ \hmu_\star(c^\omega) \\[5mm]
\mathrm{otherwise}:\\
\minstar \underset{c \in \SC(K)}{\Conv} ~ \hmu'(c^\omega).
\end{matrix} \right.$$

Let $\tau_\C$ be a memoryless strategy for Challenger in the game $\Conc_{\lambda i}(G)_{\|s_0}$; it can also be considered as a memoryless strategy in the game $\Conc_{\lambda_n i}(G)_{\|s_0}$.

Let us now define:
$$\gamma_n = \sup_{v \in V} (\lambda(v) - \lambda_n(v)).$$
Then, the sequence $(\gamma_n)_n$ is non-increasing and converges to $0$. Moreover, for each transition $st \in \Delta$, we have:
$$\hpi'_j(st) \in [\hpi_j(st) - \gamma_n, \hpi_j(st)].$$

Let:
$$\Gamma_n = \left\{ \left. \bx \in \R^{\Pi \cup \{\star\}} ~\right|~ x_\star = 0 \mathrm{~and~} \forall j \in \Pi, x_j \in [0, \gamma_n] \right\}.$$

Then, let $K$ be a strongly connected component of the graph $\Conc_{\lambda i}(G)[\tau_\C]$, without deviation, accessible from $s_0$; we have:
$$\underset{c \in \SC(K)}{\Conv} \hmu'(c^\omega) \subseteq \underset{c \in \SC(K)}{\Conv} \hmu(c^\omega) + \Gamma_n.$$

Let $R = \left\{\left. \bx \in \R^{\Pi \cup \{\star\}} ~\right|~ \forall j \in \Pi, x_j \geq 0 \right\}$.

\begin{itemize}
	\item If $\underset{c \in \SC(K)}{\Conv} \hmu(c^\omega) \cap R = \emptyset$, since $\underset{c \in \SC(K)}{\Conv} \hmu(c^\omega)$ and $R$ are closed sets, if $\gamma_n$ is small enough, we have $\underset{c \in \SC(K)}{\Conv} \hmu'(c^\omega) \cap R = \emptyset$. Therefore, if:
	$$\minstar \underset{c \in \SC(K)}{\Conv} \hmu(c^\omega) = +\infty,$$
	then, for $n$ great enough:
	$$\minstar \underset{c \in \SC(K)}{\Conv} \hmu'(c^\omega) = +\infty.$$

	\item Otherwise, we have:
	$$\minstar \underset{c \in \SC(K)}{\Conv} \hmu'(c^\omega) \geq
	\minstar \underset{c \in \SC(K)}{\Conv} \hmu(c^\omega) - \gamma_n \max_{\begin{tiny}\begin{matrix}
		c \in \SC(K) \\
		d \in \SC(K)
		\end{matrix}\end{tiny}}
	\sum_{\begin{tiny}\begin{matrix}
		j \in \Pi, \\
		\hmu_j(c^\omega) >\\
		\hmu_j(d^\omega)
		\end{matrix}\end{tiny}} \frac{\hmu_\star(c^\omega) - \hmu_\star(d^\omega)}{\hmu_j(c^\omega) - \hmu_j(d^\omega)}$$
	and if $\gamma_n$ is small enough, we have:
	$$\minstar \underset{c \in \SC(K)}{\Conv} \hmu'(c^\omega) \geq \minstar \underset{c \in \SC(K)}{\Conv} \hmu(c^\omega) - \delta.$$
\end{itemize}

In both cases, we find that there exists $\gamma_n$ small enough, i.e. $n$ great enough, to ensure:
$$\minstar \underset{c \in \SC(K)}{\Conv} \hmu'(c^\omega) \geq \minstar \underset{c \in \SC(K)}{\Conv} \hmu(c^\omega) - \delta.$$

We can find such $n$ for each strongly connected component $K$ without deviation, and there exists a finite number of such components. Moreover, when $K$ is a strongly connected component with a deviation, the quantity:
$$\underset{c \in \SC(K)}{\min} ~ \hmu_\star(c^\omega)$$
is the same in $\Conc_{\lambda i}(G)$ and in $\Conc_{\lambda_n i}(G)$. Therefore, there exists $n \in \N$ such that:
$$\min_{\scriptsize{\begin{matrix}
		K \in \SConn\left(\Conc_{\lambda_n i}(G)[\tau_\C]\right) \\
		\mathrm{accessible~from~} s_0
		\end{matrix}}} ~
\opt(K)
\geq
\min_{\scriptsize{\begin{matrix}
		K \in \SConn\left(\Conc_{\lambda i}(G)[\tau_\C]\right) \\
		\mathrm{accessible~from~} s_0
		\end{matrix}}} ~
\opt(K) - \delta.$$

We find such $n$ for every memoryless strategy $\tau_\C$, and there exists a finite number of such strategies. Therefore, there exists $n \in \N$ such that:
$$\nego(\lambda_n)(v) \geq \nego(\lambda)(v) - \delta.$$

Finally, since there are finitely many states $v \in V$, we can conclude to the existence of $n \in \N$ such that for each $v \in V$, we have:
$$\nego(\lambda_n)(v) \geq \nego(\lambda)(v) - \delta.$$

The negotiation function is Scott-continuous.
\end{proof}

    \section{Proof of Theorem \ref{thm_formula_nego}} \label{pf_formula_nego}
	
\noindent \textbf{Theorem~\ref{thm_formula_nego}.} \emph{Let $G$ be a mean-payoff game. Let us assimilate any requirement $\lambda$ on $G$ with finite values to the tuple $\blambda = (\lambda(v))_{v \in V}$, element of the vector space of finite dimension $\R^V$. Then, for each player $i$ and every vertex $v_0 \in V_i$, the quantity $\nego(\lambda)(v_0)$ is a piecewise linear function of $\blambda$, and an effective expression of that function can be computed in time double exponential in the size of $G$.}

\begin{proof}
By Lemma \ref{lm_resolution_concrete_game}, we have the formula:
$$\nego(\lambda)(v_0) = \max_{\tiny{\tau_\C \in \ML_\C\left(\Conc_{\lambda i}(G)\right)}} ~
\min_{\tiny{\begin{matrix}
    K \in \SConn(\Conc_{\lambda i}(G)[\tau_\C]) \\
    \mathrm{accessible~from~} (v_0, \{v_0\})
\end{matrix}}}
\opt(K).$$

Let $\tau_\C$ be a memoryless strategy of Challenger, and let $K$ be a strongly connected component of the graph $\Conc_{\lambda i}(G)[\tau_\C]$. Let us prove that the quantity:
$$\opt(K) = \left\{\begin{matrix}
\mathrm{if~} K \mathrm{~contains~a~dev.}: &
\underset{c \in \SC(K)}{\min} ~ \hmu_\star(c^\omega) \\
\mathrm{otherwise}: &
\minstar \underset{c \in \SC(K)}{\Conv} ~ \hmu(c^\omega)
\end{matrix} \right.$$
is a piecewise linear function of $\blambda$.

When $K$ contains a deviation, the quantity:
$$\underset{c \in \SC(K)}{\min} ~ \hmu_\star(c^\omega)$$
is independent of $\lambda$, and the result is then immediate. Let us now study the case where $K$ does not contain any deviation, i.e. let us prove that the quantity:
$$f(\lambda) = \minstar \underset{c \in \SC(K)}{\Conv} ~ \hmu(c^\omega)$$
is a piecewise linear function of $\lambda$.

Let $M$ be the common memory of the states of $K$ (since $K$ does not contain deviations). We know that for each $j \in \Pi$ and for every cycle $c \in \SC(K)$, we have:
$$\hmu_j(c^\omega) = \mu_j(\dc^\omega) - \max_{v \in V_j \cap M} \lambda(v).$$

Let $C = \left\{ \dc ~|~ c \in \SC(K) \right\}$. Since there is no deviation in $K$, any cycle in $C$ is a simple cycle of $G$. Then, the quantity $f(\lambda)$ is the minimal $x_i$ for $\bx$ in the set:
$$X = \underset{c \in C}{\Conv} ~ \mu(c^\omega) \cap \bigcap_{v \in M} \left\{ \bx ~|~ x_j \geq \lambda(v) \mathrm{~with~} v \in V_j \right\}.$$

The set $X$, intersection of a polyhedron and a polytope, is a polytope: therefore, there exists a vertex $\bx$ of that polytope which minimizes $x_i$ for $\bx \in X$. That vertex is the intersection between a face of the greater polytope $P = \underset{c \in C}{\Conv} ~ \mu(c^\omega)$, and some of the hyperplanes $H_v$ (possibly zero), defined as the hyperplanes of equation $x_j = \lambda(v)$ for $j \in \Pi$ controlling $v$, such that $\lambda(v) = \underset{w \in M \cap V_j}{\max} \lambda(w)$.

\begin{ex}
With three cycles and two players against player $i$, each controlling one vertex $v$ such that $\lambda(v) = 0$, the vertex $\bx$ is the red point in Figure~\ref{fig_zero_hyp} and Figure~\ref{fig_two_hyp}.

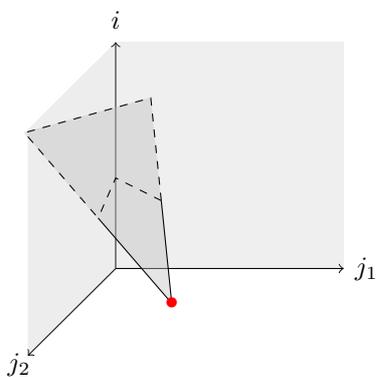
\begin{figure} 
    \centering
	\begin{tikzpicture}[scale=0.3]
	\draw[->] (0,0,0) -- (10,0,0);
	\draw[->] (0,0,0) -- (0,10,0);
	\draw[->] (0,0,0) -- (0,0,10);
	
	\draw (11,0,0) node {$j_1$};
	\draw (0,11,0) node {$i$};
	\draw (0,0,11) node {$j_2$};
	
	\fill[gray!60,opacity=0.2] (0,0,0) -- (10,0,0) -- (10,10,0) -- (0,10,0);
	\fill[gray!70,opacity=0.2] (0,0,0) -- (0,0,10) -- (0,10,10) -- (0,10,0);
	
	\filldraw[gray!50,opacity=0.4] (0,6,-4) -- (-4,6,0) -- (4,0,4);
	\draw[dashed] (2,3,0) -- (0,6,-4) -- (-4,6,0) -- (0,3,2);
	\draw (2,3,0) -- (4,0,4) -- (0,3,2);
	
	\draw[dashed] (2,3,0) -- (0,4,0) -- (0,3,2);
	
	\draw[red] (4,0,4) node {$\bullet$};
	\end{tikzpicture}
	\caption{The intersection between a $0$-dimensional face and zero hyperplane} \label{fig_zero_hyp}
\end{figure}

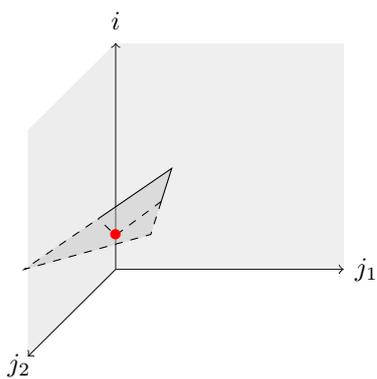
\begin{figure} 
    \centering
	\begin{tikzpicture}[scale=0.3]
	\draw[->] (0,0,0) -- (10,0,0);
	\draw[->] (0,0,0) -- (0,10,0);
	\draw[->] (0,0,0) -- (0,0,10);
	
	\draw (11,0,0) node {$j_1$};
	\draw (0,11,0) node {$i$};
	\draw (0,0,11) node {$j_2$};
	
	\fill[gray!60,opacity=0.2] (0,0,0) -- (10,0,0) -- (10,10,0) -- (0,10,0);
	\fill[gray!70,opacity=0.2] (0,0,0) -- (0,0,10) -- (0,10,10) -- (0,10,0);
	
	\filldraw[gray!50,opacity=0.4] (0,0,-4) -- (-4,0,0) -- (4,6,4);
	\draw[dashed] (2,3,0) -- (0,0,-4) -- (-4,0,0) -- (0,3,2);
	\draw (2,3,0) -- (4,6,4) -- (0,3,2);
	
	\draw[dashed] (2,3,0) -- (0,3/2,0) -- (0,3,2);
	
	\draw[red] (0,3/2,0) node {$\bullet$};
	\end{tikzpicture}
	\caption{The intersection between a $2$-dimensional face and two hyperplanes} \label{fig_two_hyp}
\end{figure}
\end{ex}

The set of vertices of the polyhedron $X$ is included in the finite set:
$$Y = \left\{ \by_{WD} \in \R^{\Pi \cup \{\star\}} ~\left|~
\begin{matrix}
	W \subseteq M, ~D \subseteq C,\\
	\underset{c \in D}{\Conv} ~ \mu(c^\omega) \cap \underset{w \in W}{\bigcap} H_w = \{\by_{WD}\} \\
	\mathrm{and~} \forall j, ~ \forall v \in M \cap V_j,\\
	y_{WDj} \geq \lambda(v)
\end{matrix} \right.\right\},$$
where the tuple $\by_{WD}$ is the intersection of the face of $P$ delimited by the values of the cycles of $D$, and the hyperplanes $H_v$ for $v \in W$, as states the condition $\underset{c \in D}{\Conv} ~ \mu(c^\omega) \cap \underset{w \in W}{\bigcap} H_w = \{\by_{WD}\}$.

The condition $\forall j, \forall v \in M \cap V_j, y_{WDj} \geq \lambda(v)$ states that the vertex $\by_{WD}$ is, moreover, the outcome of a $\lambda$-consistent play of $G$, which guarantees that this set $Y$ is itself included in $X$.

We have, therefore:
$$\opt(K) = \minstar \underset{c \in C}{\Conv} ~ \mu(c^\omega) = \min_{\by \in Y} y_i.$$

Let now $\by \in Y$, and let $W$ and $D$ be such that $\by = \by_{WD}$.

Let us choose $D$ and $W$ minimal, so that each player $j \in \Pi$ controls at most one state $w \in W$, and so that there exists only one decomposition:
$$\by = \underset{c \in D}{\sum} \alpha_c \mu(c^\omega)$$
with  $0 < \alpha_c < 1$ for each $c$, and $\sum_c \alpha_c = 1$. Furthermore, $\by$ is the only such solution of the system of equations:
$$\forall j \in \Pi, \forall w \in W \cap V_j, y_j = \lambda(w).$$

Therefore, the vector $\balpha = (\alpha_c)_{c \in D}$ is the only solution of the system:
$$\left\{\begin{matrix}
\underset{c \in D}{\sum} \alpha_c = 1 \\
\forall j \in \Pi, \forall w \in W \cap V_j, \underset{c \in D}{\sum} \alpha_c \mu_j(c^\omega) = \lambda(w) \\
\forall c \in D, \alpha_c > 0.
\end{matrix}\right.$$

Then, if $\oplus$ is a symbol and $A_{WD}$ is the matrix:
$$A_{WD} = \left(\left\{\begin{matrix}
1 & \mathrm{if~} w = \oplus \\
\mu_j(c^\omega) & \mathrm{else,~with~} w \in V_j
\end{matrix}\right. \right)_{\tiny{\begin{matrix}
		w \in W \cup \{\oplus\},\\ c \in D
		\end{matrix}}}$$
then $A_{WD}$ is invertible and:
$$\balpha = A_{WD}^{-1} \left(\left\{\begin{matrix}
1 & \mathrm{if~} w = \oplus \\
\lambda(w) & \mathrm{otherwise}
\end{matrix}\right.\right)_{w \in W \cup \{\oplus\}},$$
with $\alpha_c > 0$ for all $c \in D$.

Let us write:
$$\bbeta_{\lambda W} = \left(\left\{\begin{matrix}
1 & \mathrm{if~} j = \oplus \\
\lambda(w) & \mathrm{otherwise}
\end{matrix}\right.\right)_{w \in W \cup \{\oplus\}}.$$
We have, thus, $\balpha = A_{WD}^{-1} \bbeta_{\lambda W}$.

Let us write, for each player $j$, $\bgamma_D^j = (\mu_j(c^\omega))_{c \in D}$. Then, we can write:
$$\begin{matrix}
y_i &=& \sum_c \alpha_c \mu_i(c^\omega) \\
&=& \t \bgamma_D^i ~ \balpha \\
&=& \t \bgamma_D^i ~ A_{WD}^{-1} ~ \bbeta_{\lambda W}.
\end{matrix}$$

Finally, if we write:
$$B_W = \left(\left\{\begin{matrix}
1 & \mathrm{if~} w = v \\
0 & \mathrm{otherwise}
\end{matrix}\right.\right)_{w \in W \cup \{\oplus\}, v \in V}$$
and:
$$\bdelta_W = \left(\left\{\begin{matrix}
1 & \mathrm{if~} w = \oplus \\
0 & \mathrm{otherwise}
\end{matrix}\right.\right)_{w \in W \cup \{\oplus\}}$$
we have $\bbeta_{\lambda W} = B_W \blambda + \bdelta_W$, and therefore:
$$y_i = \t \bgamma_D^i ~ A_{WD}^{-1} ~ (B_W ~ \blambda + \bdelta_W).$$

Conversely, the tuple $\by$ defined by, for each $j \in \Pi$,
$$y_j = \t \bgamma_D^j ~ A_{WD}^{-1} ~ (B_W ~ \blambda + \bdelta_W)$$
for given $W \subseteq M$ and $D \subseteq C$, is an element of the set $Y$ if and only if:

\begin{itemize}
	\item the intersection $\underset{c \in D}{\Conv} ~ \mu(c^\omega) \cap \underset{w \in W}{\bigcap} H_w$ is a singleton, i.e. the matrix $A_{WD}$ is invertible (otherwise the matrix $A_{WD}^{-1}$ is not defined);
	
	\item $\by \in \underset{c \in D}{\Conv} ~ \mu(c^\omega)$, i.e. the tuple $\balpha = A_{WD}^{-1} ~ (B_W ~ \blambda + \bdelta_W)$ has only non-negative coordinates (actually positive if $D$ is minimal);
	
	\item for each player $j$, for each vertex $v \in M \cap V_j$, we have $y_j \geq \lambda(v)$, i.e. $\t \gamma_D^j A_{WD}^{-1} (B_W \blambda + \bdelta_W) \geq \lambda(v)$.
\end{itemize}

We finally find the formula:
$$\nego(\lambda)(v_0) = \max_{\tiny{\tau_\C \in \ML\left(\Conc_{\lambda_0 i}(G)\right)}} ~
\min_{\tiny{\begin{matrix}
		K \in \SConn\left(\Conc_{\lambda_0 i}(G)[\tau_\C]\right) \\
		\mathrm{accessible~from~} (v_0, \{v_0\})
		\end{matrix}}}
\left\{\begin{matrix}
\mathrm{if~} K \mathrm{~contains~dev.}: &
\underset{c \in \SC(K)}{\min} ~ \hmu_\star(c^\omega) \\

\mathrm{otherwise}: &
\min S_K,
\end{matrix} \right.$$
where $S_K$ is the set of real numbers of the form:
$$\t \bgamma_D^i A^{-1}_{WD} (B_W \blambda + \bdelta_W)$$
with:

\begin{itemize}
	\item $W \subseteq M$, common memory of the states of $K$;
	
	\item $D \subseteq C$, set of the cycles of the form $\dc$, where $c$ is a simple cycle of $K$;
	
	\item the matrix $A_{WD}$ is invertible;
	
	\item the vector $A_{WD}^{-1} (B_W \blambda + \bdelta_W)$ has only positive coordinates;
	
	\item and for each $j \in \Pi$, for each $v \in M_K \cap V_j$, we have $\t \bgamma_D^j A^{-1}_{WD} (B_W \blambda + \bdelta_W) \geq \lambda(v)$.
\end{itemize}

This is, indeed, the expression of a piecewise linear function.

Computing a complete and effective expression of that function can be done by constructing all the concrete games $\Conc_{\lambda_0 i}(G)_{\|(v_0, \{v_0\})}$ for $i \in \Pi$ and $v_0 \in V_i$ (there are as many of them as there are vertices in $G$, and their size is exponential in the size of $G$, hence so is the time that one needs for their construction), and then applying the formula above for each memoryless strategy of Challenger (their number is exponential in the size of the concrete game, i.e. double exponential in the size of $G$) and each strongly connected component of the induced graph (their number is bounded by the size of the concrete game). For each $K$, the computation of $\opt(K)$ with the formula given above requires only elementary operations on matrices and vectors, that can all be done in polynomial time: therefore, the effective construction of $\nego$ as a piecewise linear function requires a time double exponential in the size of $G$.
\end{proof}

    \section{Proof of Theorem \ref{thm_np_hard}} \label{app_np_hard}

\noindent \textbf{Theorem~\ref{thm_np_hard}.} \emph{The SPE constrained existence problem is NP-hard.}

\begin{proof}
    We proceed by reduction from the NP-complete problem SAT. This proof is liberally inspired from the proof of the NP-hardness of the NE constrained existence problem in co-Büchi games by Michael Ummels, in \cite{DBLP:conf/fossacs/Ummels08}.
	
	Let $\phi = \bigwedge_{i = 1}^n \bigvee_{j = 1}^m L_{ij}$ be a formula from propositional logic, written in conjunctive normal form, over the finite variable set $X$. We construct a mean-payoff game $G^\phi_{\|v_0}$ that admits an SPE where the player $\P$ gets the payoff $1$, if and only if $\phi$ is satisfiable.
	
	First, we define the set of players $\Pi = \{\P\} \cup X$: every variable of $\phi$ is a player and there is an additional special player $\P$, called \emph{Prover}, who wants to prove that $\phi$ is satisfiable.
	
	Then, let us define the state space: for each clause $C_i$, with $i \in \Z/n\Z$, of $\phi$, we define a state $C_i$ that is controlled by Prover, and for each litteral $L_{ij}$ of $C_i$ we define a state $(C_i, L_{ij})$, that is controlled by the player $x$ such that $L_{ij} = x$ or $\neg x$.
	We add a transition from $C_i$ to $(C_i, L_{ij})$, and another one from $(C_i, L_{ij})$ to $C_{i+1}$.
	Moreover, we add a sink state $\bot$, with a transition from it to itself, and transitions from all the states of the form $(C, \neg x)$ to it.
	
	We define the weight function $\pi$ on this game as follows:
	\begin{itemize}
		\item $\pi_\P(\bot\bot) = 0$, and $\pi_\P(uv) = 1$ for any other transition $vw$;
		
		\item for each player $x$, we have $\pi_x(uv) = 0$ for every transition leading to a state of the form $v = (C, x)$, and $\pi_x(uv) = 1$ for any other transition.
	\end{itemize}
	
	Note that Prover can only get the payoffs $0$ (in a play that reaches $\bot$) or $1$ (in any other play).
	Another player $x$ gets the payoff $1$ in a play that never visits (or finitely often, or infinitely often but with negligible frequence) a vertex of the form $(C, x)$.
	Otherwise, he may get any payoff between $0.5$ and $1$, depending on the frequence with which such a state is visited.

	Finally, we initialize that game in $v_0 = C_1$.
	
	\begin{ex}
    The game $G^\phi$, when $\phi$ is the tautology $(x_1 \vee \neg x_1) \wedge \dots \wedge (x_6 \vee \neg x_6)$, is represented by Figure~\ref{fig_Gphi}.
    When the weights of an edge are not written, they are equal to $1$ for all players.
    
    \begin{figure}
        \centering
        \begin{tikzpicture}[->,>=latex,shorten >=1pt, scale=0.8, initial text={}]
    		\node[state, initial right] (C1) at (0:5) {$C_1$};
    		\node[state] (C11) at (30:6) {$x_1$};
    		\node[state] (C12) at (30:4) {$\neg x_1$};
    		\node[state] (C2) at (60:5) {$C_2$};
    		\node[state] (C21) at (90:6) {$x_2$};
    		\node[state] (C22) at (90:4) {$\neg x_2$};
    		\node[state] (C3) at (120:5) {$C_3$};
    		\node[state] (C31) at (150:6) {$x_3$};
    		\node[state] (C32) at (150:4) {$\neg x_3$};
    		\node[state] (C4) at (180:5) {$C_4$};
    		\node[state] (C41) at (210:6) {$x_4$};
    		\node[state] (C42) at (210:4) {$\neg x_4$};
    		\node[state] (C5) at (240:5) {$C_5$};
    		\node[state] (C51) at (270:6) {$x_5$};
    		\node[state] (C52) at (270:4) {$\neg x_5$};
    		\node[state] (C6) at (300:5) {$C_6$};
    		\node[state] (C61) at (330:6) {$x_6$};
    		\node[state] (C62) at (330:4) {$\neg x_6$};
    		\node[state] (b) at (0,0) {$\bot$};
    		
    		\path[->] (C1) edge node[right] {$\stackrel{x_1}{0} \stackrel{x_2}{1} \stackrel{x_3}{1} \stackrel{x_4}{1} \stackrel{x_5}{1} \stackrel{x_6}{1} \stackrel{\P}{1}$} (C11);
    		\path[->] (C1) edge (C12);
            \path[->] (C11) edge (C2);
            \path[->] (C12) edge (C2);
            \path[->] (C12) edge (b);
            \path[->] (C2) edge node[above right] {$\stackrel{x_1}{1} \stackrel{x_2}{0} \stackrel{x_3}{1} \stackrel{x_4}{1} \stackrel{x_5}{1} \stackrel{x_6}{1} \stackrel{\P}{1}$} (C21);
    		\path[->] (C2) edge (C22);
            \path[->] (C21) edge (C3);
            \path[->] (C22) edge (C3);
            \path[->] (C22) edge (b);
            \path[->] (C3) edge node[above left] {$\stackrel{x_1}{1} \stackrel{x_2}{1} \stackrel{x_3}{0} \stackrel{x_4}{1} \stackrel{x_5}{1} \stackrel{x_6}{1} \stackrel{\P}{1}$} (C31);
    		\path[->] (C3) edge (C32);
            \path[->] (C31) edge (C4);
            \path[->] (C32) edge (C4);
            \path[->] (C32) edge (b);
            \path[->] (C4) edge node[left] {$\stackrel{x_1}{1} \stackrel{x_2}{1} \stackrel{x_3}{1} \stackrel{x_4}{0} \stackrel{x_5}{1} \stackrel{x_6}{1} \stackrel{\P}{1}$} (C41);
    		\path[->] (C4) edge (C42);
            \path[->] (C41) edge (C5);
            \path[->] (C42) edge (C5);
            \path[->] (C42) edge (b);
            \path[->] (C5) edge node[below left] {$\stackrel{x_1}{1} \stackrel{x_2}{1} \stackrel{x_3}{1} \stackrel{x_4}{1} \stackrel{x_5}{0} \stackrel{x_6}{1} \stackrel{\P}{1}$} (C51);
    		\path[->] (C5) edge (C52);
            \path[->] (C51) edge (C6);
            \path[->] (C52) edge (C6);
            \path[->] (C52) edge (b);
            \path[->] (C6) edge node[below right] {$\stackrel{x_1}{1} \stackrel{x_2}{1} \stackrel{x_3}{1} \stackrel{x_4}{1} \stackrel{x_5}{1} \stackrel{x_6}{0} \stackrel{\P}{1}$} (C61);
    		\path[->] (C6) edge (C62);
            \path[->] (C61) edge (C1);
            \path[->] (C62) edge (C1);
            \path[->] (C62) edge (b);
            \path (b) edge[loop left] node[left] {$\stackrel{x_1}{1} \stackrel{x_2}{1} \stackrel{x_3}{1} \stackrel{x_4}{1} \stackrel{x_5}{1} \stackrel{x_6}{1} \stackrel{\P}{0}$} (b);
        \end{tikzpicture}
        \caption{The game $G^\phi$}
        \label{fig_Gphi}
    \end{figure}
\end{ex}

	Now, let us prove that there is an SPE in $G^\phi_{\|v_0}$ in which Prover gets the payoff $1$, if and only if the formula $\phi$ is satisfiable.
	
	\begin{itemize}
		\item If such an SPE exists: let us write it $\bsigma$, and let $\rho = \< \bsigma \>_{C_1}$.
		Since $\mu_\P(\rho) = 1$, the sink state $\bot$ is never visited.
		Let us define a valuation $\nu$ on $X$ as follows: for each variable $x$, we have $\nu(x) = 1$ if and only if $\mu_x(\rho) < 1$.
		
		Now, let $C$ be a clause of $\phi$: since $C$, as a state, is necessarily visited infinitely often and with a fixed frequence in the play $\rho$ (because no player ever go to the sink state $\bot$), one of its successors, say $(C, L)$, is visited with a non-negligible frequence (more formally, the time between two occurrences of $(C, L)$ is bounded).
		If $L$ is a positive litteral, say $x$, then by definition of $\nu$, we have $\nu(x) = 1$ and the clause $C$ is satisfied.
		
		If $L$ has the form $\neg x$, then each time the state $(C, \neg x)$ is traversed, player $x$ has the possibility to deviate and to go to the sink state $\bot$, where he is sure to get the payoff $1$.
		Since $\bsigma$ is an SPE, it means that he already gets the payoff $1$ in the play $\rho$.
		By definition of $\nu$, we then have $\nu(x) = 0$, hence the litteral $\neg x$ is satisfied, hence so is the clause $C$.
		
		The valuation $\nu$ satisfies all the clauses of $\phi$, and therefore satisfies the formula $\phi$ itself.

		\item If $\phi$ is satisfied by some valuation $\nu$: let us define a strategy profile $\bsigma$ by:
		\begin{itemize}
			\item for each history $h C$ where $C$ is a clause of $\phi$, $\sigma_\P(hC) = (C, L)$ where $L$ is a litteral of $C$ that is satisfied in the valuation $\nu$;
			
			\item and for each history $h(C, \neg x)$ where $C$ is a clause of $\phi$ and $x$ is a variable, $\sigma_x(h(C, \neg x)) = \bot$ if and only if $\nu(x) = 1$.
		\end{itemize}
	
		Any other state has only one successor, hence we now have completely defined a strategy profile.
		Now, let us prove it is an SPE, in which Prover gets the payoff $1$.
		
		Let $hC$ be a history, where $C$ is a clause of $\phi$.
		We want to prove that $\bsigma_{\|hC}$ is a Nash equilibrium, in which Prover gets the payoff $1$.
		Let $\rho = \< \bsigma_{\|hC} \>_C$.
		If $\mu_\P(\rho) < 1$, i.e. if $\rho$ is of the form $h D (D, \neg x) \bot^\omega$, then by definition of $\bsigma$ we have $\nu(x) = 0$.
		But then, we cannot have $\sigma_\P(D) = (D, \neg x)$: contradiction.
		The play $\rho$ never reaches the state $\bot$, and Prover gets the payoff $1$, and as a consequence she does not have any profitable deviation.
		
		Now, if another player $x$ has a profitable deviation, it means that he does not get the payoff $1$ in $\rho$, and therefore that some state of the form $(D, x)$ is visited infinitely often.
		But then, if Prover choose to go to the state $(D, x)$, it means that the litteral $x$ is satisfied in $\nu$, i.e. that $\nu(x) = 1$.
		In that case, if some clause $D'$ contains the litteral $\neg x$, it is not a litteral satisfied by $\nu$, and therefore the strategy $\sigma_\P$, as we defined it, never chooses the transition to the state $(D', \neg x)$, where player $x$ could have the possibility to deviate from his strategy.
		Contradiction.

		Finally, after a history of the form $h(C, L)$, either:
		\begin{itemize}
			\item $L = \neg x$ with $\nu(x) = 1$, and in that case, we have $\< \bsigma_{\|h(C, L)} \>_{(C, L)} = (C, L) \bot^\omega$, player $x$ gets the payoff $1$, and no player has a profitable deviation;
			
			\item or $L$ is a positive litteral, and then there exists only one transition from the state $(C, L)$ to another clause $D$, and we go back to the previous case;
			
			\item or $L = \neg x$ with $\nu(x) = 0$, and in that case, we have $\sigma_x(C, L) = D$ where $D$ is the following clause, and by the first case the strategy profile $\bsigma_{\|h(C, \neg x)D}$ is a Nash equilibrium.
			Moreover, since the litteral $\neg x$ is not satisfied in $\nu$, the play $\< \bsigma_{\|h(C, \neg x)D} \>_D$ does never traverse again any state of the form $(D', \neg x)$, hence player $x$ wins, and therefore has no profitable deviation: the strategy profile $\bsigma_{\|h(C, \neg x)}$ is a Nash equilibrium.
		\end{itemize}
	\end{itemize}
	
	The constrained SPE existence problem is NP-hard in mean-payoff games.
\end{proof}

    \section{Construction of the automaton solving the SPE constrained existence problem} \label{app_automaton}

In the proof of Theorem \ref{thm_decidable}, we invoked a multi-mean-payoff automaton $\A_\lambda$, defined from a requirement $\lambda$ and two thresholds $\bx, \by \in \Q^\Pi$, that recognizes the language of the plays $\rho \in V^\omega$ that are $\lambda$-consistent, and that satisfy $\bx \leq  \mu(\rho) \leq \by$. We give here the details of its construction.

\begin{itemize}
    \item The state space of $\A_\lambda$ is:
    $$Q = V \times 2^V,$$
    where a state $(v, M)$ must be interpreted as follows: we are currently in the vertex $v$, and we already traversed the states of the set $M$. The initial state is $(v_0, \{v_0\})$.
    
    \item The automaton $\A_\lambda$ is $2\card\Pi$-dimensional, and its dimension set is $\Pi \times \{\lambda, 0\}$.
    
    \item The transitions of $\A_\lambda$ are all the transitions of the form:
    $$(v, M)(w, M \cup \{w\})$$
    where $vw \in E$. Each such transition is labelled by the vertex $v$ as a letter of the alphabet $V$, and is weighted by:
    \begin{itemize}
        \item $\pi_i(vw) - \underset{u \in M}{\max}~\lambda(u)$ on each dimension $(i, \lambda)$;
        
        \item $\pi_i(vw)$ on each dimension $(i, 0)$.
    \end{itemize}
    
    \item A run $\rho$ of $\A_\lambda$ is accepting if its mean-payoff is nonnegative along each dimension $(i, \lambda)$, and its mean-payoff along each dimension $(i, 0)$ belongs to the interval $[x_i, y_i]$.
\end{itemize}

	\section{Some examples of negotiation sequences} \label{app_ex}

We gather in this section some examples that could be interesting for the reader who would want to get a full overall view on the behaviour of the negotiation function on the mean-payoff games. For all of them, we computed the negotiation sequence, as defined in Appendix \ref{app_nego_seq}. For some of them, we just gave the negotiation sequence; for the most important ones, we gave a complete explanation of how we computed it, using the abstract negotiation game, as defined in Appendix \ref{app_abstract}.

\begin{ex} \label{ex_sans_spe_sequence}
Let us take again the game of Example~\ref{ex_sans_spe}: let us give (in red) the values of $\lambda_1 = \nego(\lambda_0)$, which correspond to the antagonistic values.

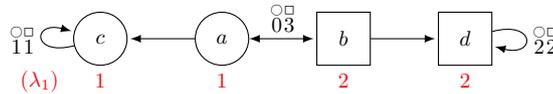
\begin{center}
	\begin{tikzpicture}[->,>=latex,shorten >=1pt, initial text={}, scale=0.8, every node/.style={scale=0.8}]
	\node[state] (a) at (0, 0) {$a$};
	\node[state] (c) at (-2, 0) {$c$};
	\node[state, rectangle] (b) at (2, 0) {$b$};
	\node[state, rectangle] (d) at (4, 0) {$d$};
	\path (a) edge (c);
	\path[<->] (a) edge node[above] {$\stackrel{\playcircle}{0} \stackrel{\Box}{3}$} (b);
	\path (b) edge (d);
	\path (d) edge [loop right] node {$\stackrel{\playcircle}{2} \stackrel{\Box}{2}$} (d);
	\path (c) edge [loop left] node {$\stackrel{\playcircle}{1} \stackrel{\Box}{1}$} (c);
	
	\node[red] (l) at (-3, -0.7) {$(\lambda_1)$};
	\node[red] (la) at (0, -0.7) {$1$};
	\node[red] (lb) at (2, -0.7) {$2$};
	\node[red] (lc) at (-2, -0.7) {$1$};
	\node[red] (ld) at (4, -0.7) {$2$};
	\end{tikzpicture}
\end{center}

At the second step, let us execute the abstract game on the state $a$, with the requirement $\lambda_1$: whatever Prover proposes at first, Challenger has the possibility to deviate and to reach the state $b$. Then, Prover has to propose a $\lambda_1$-consistent play from the state $b$, i.e. a play in which player $\Circle$ gets at least the payoff $2$: such a play necessarily ends in the state $d$, and gives player $\Box$ the payoff $2$.

The other states keep the same values.

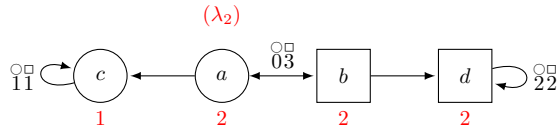
\begin{center}
	\begin{tikzpicture}[<->,>=latex,shorten >=1pt, , scale=0.8, every node/.style={scale=0.8}]
	\node[state] (a) at (0, 0) {$a$};
	\node[state] (c) at (-2, 0) {$c$};
	\node[state, rectangle] (b) at (2, 0) {$b$};
	\node[state, rectangle] (d) at (4, 0) {$d$};
	\path[->] (a) edge (c);
	\path[<->] (a) edge node[above] {$\stackrel{\playcircle}{0} \stackrel{\Box}{3}$} (b);
	\path[->] (b) edge (d);
	\path (d) edge [loop right] node {$\stackrel{\playcircle}{2} \stackrel{\Box}{2}$} (d);
	\path (c) edge [loop left] node {$\stackrel{\playcircle}{1} \stackrel{\Box}{1}$} (c);
	
	\node[red] (l) at (0, 1) {$(\lambda_2)$};
	\node[red] (la) at (0, -0.7) {$2$};
	\node[red] (lb) at (2, -0.7) {$2$};
	\node[red] (lc) at (-2, -0.7) {$1$};
	\node[red] (ld) at (4, -0.7) {$2$};
	\end{tikzpicture}
\end{center}

But then, at the third step, from the state $b$: whatever Prover proposes at first, Challenger can deviate to reach the state $a$. Then, Prover has to propose a $\lambda_2$-consistent play from $a$, i.e. a play in which player $\Circle$ gets at least the payoff $2$: such a play necessarily end in the state $d$, i.e. after possibly some prefix, Prover proposes the play $abd^\omega$. But then, Challenger can always deviate to go back to the state $a$; and the play which is thus created is $(ab)^\omega$ which gives player $\Box$ the payoff $3$.

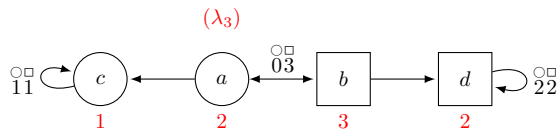
\begin{center}
	\begin{tikzpicture}[<->,>=latex,shorten >=1pt, , scale=0.8, every node/.style={scale=0.8}]
	\node[state] (a) at (0, 0) {$a$};
	\node[state] (c) at (-2, 0) {$c$};
	\node[state, rectangle] (b) at (2, 0) {$b$};
	\node[state, rectangle] (d) at (4, 0) {$d$};
	\path[->] (a) edge (c);
	\path[<->] (a) edge node[above] {$\stackrel{\playcircle}{0} \stackrel{\Box}{3}$} (b);
	\path[->] (b) edge (d);
	\path (d) edge [loop right] node {$\stackrel{\playcircle}{2} \stackrel{\Box}{2}$} (d);
	\path (c) edge [loop left] node {$\stackrel{\playcircle}{1} \stackrel{\Box}{1}$} (c);
	
	\node[red] (l) at (0, 1) {$(\lambda_3)$};
	\node[red] (la) at (0, -0.7) {$2$};
	\node[red] (lb) at (2, -0.7) {$3$};
	\node[red] (lc) at (-2, -0.7) {$1$};
	\node[red] (ld) at (4, -0.7) {$2$};
	\end{tikzpicture}
\end{center}

Finally, from the states $a$ and $b$, there exists no $\lambda_3$-consistent play, and therefore no $\lambda$-rational strategy profile.

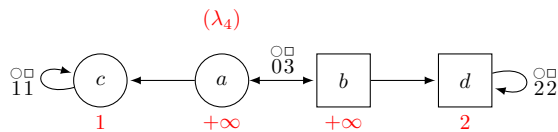
\begin{center}
	\begin{tikzpicture}[<->,>=latex,shorten >=1pt, , scale=0.8, every node/.style={scale=0.8}]
	\node[state] (a) at (0, 0) {$a$};
	\node[state] (c) at (-2, 0) {$c$};
	\node[state, rectangle] (b) at (2, 0) {$b$};
	\node[state, rectangle] (d) at (4, 0) {$d$};
	\path[->] (a) edge (c);
	\path[<->] (a) edge node[above] {$\stackrel{\playcircle}{0} \stackrel{\Box}{3}$} (b);
	\path[->] (b) edge (d);
	\path (d) edge [loop right] node {$\stackrel{\playcircle}{2} \stackrel{\Box}{2}$} (d);
	\path (c) edge [loop left] node {$\stackrel{\playcircle}{1} \stackrel{\Box}{1}$} (c);
	
	\node[red] (l) at (0, 1) {$(\lambda_4)$};
	\node[red] (la) at (0, -0.7) {$+\infty$};
	\node[red] (lb) at (2, -0.7) {$+\infty$};
	\node[red] (lc) at (-2, -0.7) {$1$};
	\node[red] (ld) at (4, -0.7) {$2$};
	\end{tikzpicture}
\end{center}
and for all $n \geq 4$, $\lambda_n = \lambda_4$.
\end{ex}

\begin{ex} \label{ex_propagation}
In this example, we show a game that can be turned into a family of games, where the negotiation function needs as many steps as there are states to reach its limit: when the requirement changes in some state, it opens new possibilities from the neighbour states, and so on.

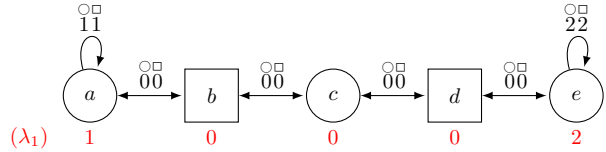
\begin{center}
	\begin{tikzpicture}[<->,>=latex,shorten >=1pt, scale=0.8, every node/.style={scale=0.8}]
	\node[state] (a) at (0, 0) {$a$};
	\node[state, rectangle] (b) at (2, 0) {$b$};
	\node[state] (c) at (4, 0) {$c$};
	\node[state, rectangle] (d) at (6, 0) {$d$};
	\node[state] (e) at (8, 0) {$e$};
	
	\path (a) edge[loop above] node[above] {$\stackrel{\playcircle}{1} \stackrel{\Box}{1}$} (a);
	\path (a) edge node[above] {$\stackrel{\playcircle}{0} \stackrel{\Box}{0}$} (b);
	\path (b) edge node[above] {$\stackrel{\playcircle}{0} \stackrel{\Box}{0}$} (c);
	\path (c) edge node[above] {$\stackrel{\playcircle}{0} \stackrel{\Box}{0}$} (d);
	\path (d) edge node[above] {$\stackrel{\playcircle}{0} \stackrel{\Box}{0}$} (e);
	\path (e) edge[loop above] node[above] {$\stackrel{\playcircle}{2} \stackrel{\Box}{2}$} (e);
	
	\node[red] (l) at (-1, -0.7) {$(\lambda_1)$};
	\node[red] (la) at (0, -0.7) {$1$};
	\node[red] (lb) at (2, -0.7) {$0$};
	\node[red] (lc) at (4, -0.7) {$0$};
	\node[red] (ld) at (6, -0.7) {$0$};
	\node[red] (le) at (8, -0.7) {$2$};
	\end{tikzpicture}
\end{center}

~

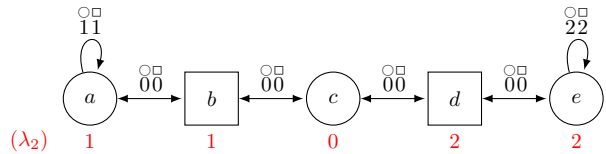
\begin{center}
	\begin{tikzpicture}[<->,>=latex,shorten >=1pt, scale=0.8, every node/.style={scale=0.8}]
	\node[state] (a) at (0, 0) {$a$};
	\node[state, rectangle] (b) at (2, 0) {$b$};
	\node[state] (c) at (4, 0) {$c$};
	\node[state, rectangle] (d) at (6, 0) {$d$};
	\node[state] (e) at (8, 0) {$e$};
	
	\path (a) edge[loop above] node[above] {$\stackrel{\playcircle}{1} \stackrel{\Box}{1}$} (a);
	\path (a) edge node[above] {$\stackrel{\playcircle}{0} \stackrel{\Box}{0}$} (b);
	\path (b) edge node[above] {$\stackrel{\playcircle}{0} \stackrel{\Box}{0}$} (c);
	\path (c) edge node[above] {$\stackrel{\playcircle}{0} \stackrel{\Box}{0}$} (d);
	\path (d) edge node[above] {$\stackrel{\playcircle}{0} \stackrel{\Box}{0}$} (e);
	\path (e) edge[loop above] node[above] {$\stackrel{\playcircle}{2} \stackrel{\Box}{2}$} (e);
	
	\node[red] (l) at (-1, -0.7) {$(\lambda_2)$};
	\node[red] (la) at (0, -0.7) {$1$};
	\node[red] (lb) at (2, -0.7) {$1$};
	\node[red] (lc) at (4, -0.7) {$0$};
	\node[red] (ld) at (6, -0.7) {$2$};
	\node[red] (le) at (8, -0.7) {$2$};
	\end{tikzpicture}
\end{center}

~

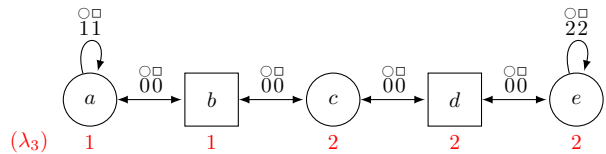
\begin{center}
	\begin{tikzpicture}[<->,>=latex,shorten >=1pt, scale=0.8, every node/.style={scale=0.8}]
	\node[state] (a) at (0, 0) {$a$};
	\node[state, rectangle] (b) at (2, 0) {$b$};
	\node[state] (c) at (4, 0) {$c$};
	\node[state, rectangle] (d) at (6, 0) {$d$};
	\node[state] (e) at (8, 0) {$e$};
	
	\path (a) edge[loop above] node[above] {$\stackrel{\playcircle}{1} \stackrel{\Box}{1}$} (a);
	\path (a) edge node[above] {$\stackrel{\playcircle}{0} \stackrel{\Box}{0}$} (b);
	\path (b) edge node[above] {$\stackrel{\playcircle}{0} \stackrel{\Box}{0}$} (c);
	\path (c) edge node[above] {$\stackrel{\playcircle}{0} \stackrel{\Box}{0}$} (d);
	\path (d) edge node[above] {$\stackrel{\playcircle}{0} \stackrel{\Box}{0}$} (e);
	\path (e) edge[loop above] node[above] {$\stackrel{\playcircle}{2} \stackrel{\Box}{2}$} (e);
	
	\node[red] (l) at (-1, -0.7) {$(\lambda_3)$};
	\node[red] (la) at (0, -0.7) {$1$};
	\node[red] (lb) at (2, -0.7) {$1$};
	\node[red] (lc) at (4, -0.7) {$2$};
	\node[red] (ld) at (6, -0.7) {$2$};
	\node[red] (le) at (8, -0.7) {$2$};
	\end{tikzpicture}
\end{center}

~

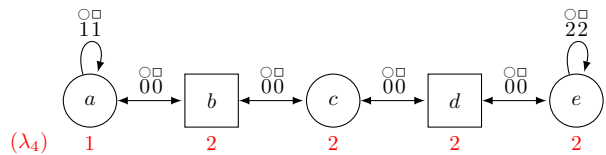
\begin{center}
	\begin{tikzpicture}[<->,>=latex,shorten >=1pt, scale=0.8, every node/.style={scale=0.8}]
	\node[state] (a) at (0, 0) {$a$};
	\node[state, rectangle] (b) at (2, 0) {$b$};
	\node[state] (c) at (4, 0) {$c$};
	\node[state, rectangle] (d) at (6, 0) {$d$};
	\node[state] (e) at (8, 0) {$e$};
	
	\path (a) edge[loop above] node[above] {$\stackrel{\playcircle}{1} \stackrel{\Box}{1}$} (a);
	\path (a) edge node[above] {$\stackrel{\playcircle}{0} \stackrel{\Box}{0}$} (b);
	\path (b) edge node[above] {$\stackrel{\playcircle}{0} \stackrel{\Box}{0}$} (c);
	\path (c) edge node[above] {$\stackrel{\playcircle}{0} \stackrel{\Box}{0}$} (d);
	\path (d) edge node[above] {$\stackrel{\playcircle}{0} \stackrel{\Box}{0}$} (e);
	\path (e) edge[loop above] node[above] {$\stackrel{\playcircle}{2} \stackrel{\Box}{2}$} (e);
	
	\node[red] (l) at (-1, -0.7) {$(\lambda_4)$};
	\node[red] (la) at (0, -0.7) {$1$};
	\node[red] (lb) at (2, -0.7) {$2$};
	\node[red] (lc) at (4, -0.7) {$2$};
	\node[red] (ld) at (6, -0.7) {$2$};
	\node[red] (le) at (8, -0.7) {$2$};
	\end{tikzpicture}
\end{center}

~

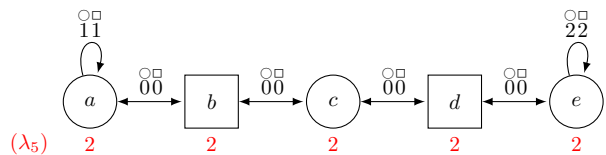
\begin{center}
	\begin{tikzpicture}[<->,>=latex,shorten >=1pt, scale=0.8, every node/.style={scale=0.8}]
	\node[state] (a) at (0, 0) {$a$};
	\node[state, rectangle] (b) at (2, 0) {$b$};
	\node[state] (c) at (4, 0) {$c$};
	\node[state, rectangle] (d) at (6, 0) {$d$};
	\node[state] (e) at (8, 0) {$e$};
	
	\path (a) edge[loop above] node[above] {$\stackrel{\playcircle}{1} \stackrel{\Box}{1}$} (a);
	\path (a) edge node[above] {$\stackrel{\playcircle}{0} \stackrel{\Box}{0}$} (b);
	\path (b) edge node[above] {$\stackrel{\playcircle}{0} \stackrel{\Box}{0}$} (c);
	\path (c) edge node[above] {$\stackrel{\playcircle}{0} \stackrel{\Box}{0}$} (d);
	\path (d) edge node[above] {$\stackrel{\playcircle}{0} \stackrel{\Box}{0}$} (e);
	\path (e) edge[loop above] node[above] {$\stackrel{\playcircle}{2} \stackrel{\Box}{2}$} (e);
	
	\node[red] (l) at (-1, -0.7) {$(\lambda_5)$};
	\node[red] (la) at (0, -0.7) {$2$};
	\node[red] (lb) at (2, -0.7) {$2$};
	\node[red] (lc) at (4, -0.7) {$2$};
	\node[red] (ld) at (6, -0.7) {$2$};
	\node[red] (le) at (8, -0.7) {$2$};
	\end{tikzpicture}
\end{center}
and the requirement $\lambda_5$ is a fixed point of the negocation function.
\end{ex}

\begin{ex} \label{ex_sans_spe_sc}
In all the previous examples, all the games whose underlying graphs were strongly connected contained SPEs. Here is an example of game with a strongly connected underlying graph that does not contain SPEs.

\begin{center}
	\begin{tikzpicture}[<->,>=latex,shorten >=1pt, scale=0.8, every node/.style={scale=0.8}]
	\node[state, rectangle] (b) at (0, 0) {$b$};
	\node[state] (c) at (1, -1.7) {$c$};
	\node[state, rectangle] (a) at (2, 0) {$a$};
	\node[state] (d) at (4, 0) {$d$};
	\node[state] (e) at (6, 0) {$e$};
	\node[state, rectangle] (f) at (5, 1.7) {$f$};
	
	\path[<->] (a) edge node[above] {$\stackrel{\playcircle}{3} \stackrel{\Box}{0}$} (d);
	\path[->] (a) edge node[above] {$\stackrel{\playcircle}{0} \stackrel{\Box}{0}$} (b);
	\path[->] (b) edge node[below left] {$\stackrel{\playcircle}{0} \stackrel{\Box}{0}$} (c);
	\path[->] (c) edge node[below right] {$\stackrel{\playcircle}{0} \stackrel{\Box}{0}$} (a);
	\path (b) edge [loop left] node {$\stackrel{\playcircle}{1} \stackrel{\Box}{1}$} (b);
	\path (c) edge [loop below] node {$\stackrel{\playcircle}{3} \stackrel{\Box}{0}$} (c);
	\path[->] (d) edge node[below] {$\stackrel{\playcircle}{0} \stackrel{\Box}{0}$} (e);
	\path[->] (e) edge node[above right] {$\stackrel{\playcircle}{0} \stackrel{\Box}{0}$} (f);
	\path[->] (f) edge node[above left] {$\stackrel{\playcircle}{0} \stackrel{\Box}{0}$} (d);
	\path (e) edge [loop right] node {$\stackrel{\playcircle}{2} \stackrel{\Box}{2}$} (e);
	\path (f) edge [loop above] node {$\stackrel{\playcircle}{0} \stackrel{\Box}{4}$} (f);
	
	\node[red] (l) at (1, 1.7) {$(\lambda_1)$};
	\node[red] (la) at (2, 0.7) {$1$};
	\node[red] (lb) at (0, 0.7) {$1$};
	\node[red] (lc) at (1.7, -1.7) {$3$};
	\node[red] (ld) at (4, -0.7) {$2$};
	\node[red] (le) at (6, -0.7) {$2$};
	\node[red] (lf) at (5.7, 1.7) {$4$};
	\end{tikzpicture}
\end{center}

\begin{center}
	\begin{tikzpicture}[<->,>=latex,shorten >=1pt, scale=0.8, every node/.style={scale=0.8}]
	\node[state, rectangle] (b) at (0, 0) {$b$};
	\node[state] (c) at (1, -1.7) {$c$};
	\node[state, rectangle] (a) at (2, 0) {$a$};
	\node[state] (d) at (4, 0) {$d$};
	\node[state] (e) at (6, 0) {$e$};
	\node[state, rectangle] (f) at (5, 1.7) {$f$};
	
	\path[<->] (a) edge node[above] {$\stackrel{\playcircle}{3} \stackrel{\Box}{0}$} (d);
	\path[->] (a) edge node[above] {$\stackrel{\playcircle}{0} \stackrel{\Box}{0}$} (b);
	\path[->] (b) edge node[below left] {$\stackrel{\playcircle}{0} \stackrel{\Box}{0}$} (c);
	\path[->] (c) edge node[below right] {$\stackrel{\playcircle}{0} \stackrel{\Box}{0}$} (a);
	\path (b) edge [loop left] node {$\stackrel{\playcircle}{1} \stackrel{\Box}{1}$} (b);
	\path (c) edge [loop below] node {$\stackrel{\playcircle}{3} \stackrel{\Box}{0}$} (c);
	\path[->] (d) edge node[below] {$\stackrel{\playcircle}{0} \stackrel{\Box}{0}$} (e);
	\path[->] (e) edge node[above right] {$\stackrel{\playcircle}{0} \stackrel{\Box}{0}$} (f);
	\path[->] (f) edge node[above left] {$\stackrel{\playcircle}{0} \stackrel{\Box}{0}$} (d);
	\path (e) edge [loop right] node {$\stackrel{\playcircle}{2} \stackrel{\Box}{2}$} (e);
	\path (f) edge [loop above] node {$\stackrel{\playcircle}{0} \stackrel{\Box}{4}$} (f);
	
	\node[red] (l) at (1, 1.7) {$(\lambda_2)$};
	\node[red] (la) at (2, 0.7) {$2$};
	\node[red] (lb) at (0, 0.7) {$1$};
	\node[red] (lc) at (1.7, -1.7) {$3$};
	\node[red] (ld) at (4, -0.7) {$2$};
	\node[red] (le) at (6, -0.7) {$2$};
	\node[red] (lf) at (5.7, 1.7) {$4$};
	\end{tikzpicture}
\end{center}

\begin{center}
	\begin{tikzpicture}[<->,>=latex,shorten >=1pt, scale=0.8, every node/.style={scale=0.8}]
	\node[state, rectangle] (b) at (0, 0) {$b$};
	\node[state] (c) at (1, -1.7) {$c$};
	\node[state, rectangle] (a) at (2, 0) {$a$};
	\node[state] (d) at (4, 0) {$d$};
	\node[state] (e) at (6, 0) {$e$};
	\node[state, rectangle] (f) at (5, 1.7) {$f$};
	
	\path[<->] (a) edge node[above] {$\stackrel{\playcircle}{3} \stackrel{\Box}{0}$} (d);
	\path[->] (a) edge node[above] {$\stackrel{\playcircle}{0} \stackrel{\Box}{0}$} (b);
	\path[->] (b) edge node[below left] {$\stackrel{\playcircle}{0} \stackrel{\Box}{0}$} (c);
	\path[->] (c) edge node[below right] {$\stackrel{\playcircle}{0} \stackrel{\Box}{0}$} (a);
	\path (b) edge [loop left] node {$\stackrel{\playcircle}{1} \stackrel{\Box}{1}$} (b);
	\path (c) edge [loop below] node {$\stackrel{\playcircle}{3} \stackrel{\Box}{0}$} (c);
	\path[->] (d) edge node[below] {$\stackrel{\playcircle}{0} \stackrel{\Box}{0}$} (e);
	\path[->] (e) edge node[above right] {$\stackrel{\playcircle}{0} \stackrel{\Box}{0}$} (f);
	\path[->] (f) edge node[above left] {$\stackrel{\playcircle}{0} \stackrel{\Box}{0}$} (d);
	\path (e) edge [loop right] node {$\stackrel{\playcircle}{2} \stackrel{\Box}{2}$} (e);
	\path (f) edge [loop above] node {$\stackrel{\playcircle}{0} \stackrel{\Box}{4}$} (f);
	
	\node[red] (l) at (1, 1.7) {$(\lambda_3)$};
	\node[red] (la) at (2, 0.7) {$2$};
	\node[red] (lb) at (0, 0.7) {$1$};
	\node[red] (lc) at (1.7, -1.7) {$3$};
	\node[red] (ld) at (4, -0.7) {$3$};
	\node[red] (le) at (6, -0.7) {$2$};
	\node[red] (lf) at (5.7, 1.7) {$4$};
	\end{tikzpicture}
\end{center}

\begin{center}
	\begin{tikzpicture}[<->,>=latex,shorten >=1pt, scale=0.8, every node/.style={scale=0.8}]
	\node[state, rectangle] (b) at (0, 0) {$b$};
	\node[state] (c) at (1, -1.7) {$c$};
	\node[state, rectangle] (a) at (2, 0) {$a$};
	\node[state] (d) at (4, 0) {$d$};
	\node[state] (e) at (6, 0) {$e$};
	\node[state, rectangle] (f) at (5, 1.7) {$f$};
	
	\path[<->] (a) edge node[above] {$\stackrel{\playcircle}{3} \stackrel{\Box}{0}$} (d);
	\path[->] (a) edge node[above] {$\stackrel{\playcircle}{0} \stackrel{\Box}{0}$} (b);
	\path[->] (b) edge node[below left] {$\stackrel{\playcircle}{0} \stackrel{\Box}{0}$} (c);
	\path[->] (c) edge node[below right] {$\stackrel{\playcircle}{0} \stackrel{\Box}{0}$} (a);
	\path (b) edge [loop left] node {$\stackrel{\playcircle}{1} \stackrel{\Box}{1}$} (b);
	\path (c) edge [loop below] node {$\stackrel{\playcircle}{3} \stackrel{\Box}{0}$} (c);
	\path[->] (d) edge node[below] {$\stackrel{\playcircle}{0} \stackrel{\Box}{0}$} (e);
	\path[->] (e) edge node[above right] {$\stackrel{\playcircle}{0} \stackrel{\Box}{0}$} (f);
	\path[->] (f) edge node[above left] {$\stackrel{\playcircle}{0} \stackrel{\Box}{0}$} (d);
	\path (e) edge [loop right] node {$\stackrel{\playcircle}{2} \stackrel{\Box}{2}$} (e);
	\path (f) edge [loop above] node {$\stackrel{\playcircle}{0} \stackrel{\Box}{4}$} (f);
	
	\node[red] (l) at (1, 1.7) {$(\lambda_4)$};
	\node[red] (la) at (2, 0.7) {$+\infty$};
	\node[red] (lb) at (0, 0.7) {$+\infty$};
	\node[red] (lc) at (1.7, -1.7) {$+\infty$};
	\node[red] (ld) at (4, -0.7) {$+\infty$};
	\node[red] (le) at (6, -0.7) {$+\infty$};
	\node[red] (lf) at (5.7, 1.7) {$+\infty$};
	\end{tikzpicture}
\end{center}
\end{ex}

\begin{ex} \label{ex_big}
This example shows how a new requirement can emerge from the combination of several cycles.

Let $G$ be the following game:

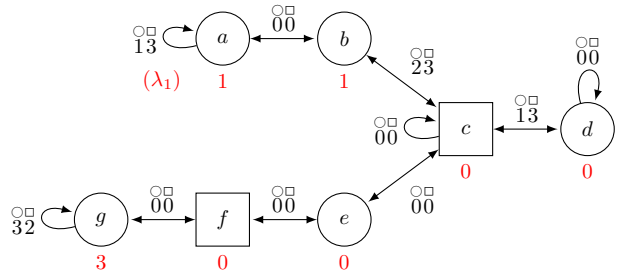
\begin{center}
	\begin{tikzpicture}[->,>=latex,shorten >=1pt, scale=0.8, every node/.style={scale=0.8}]
	\node[state] (a) at (2, 1.5) {$a$};
	\node[state] (b) at (4, 1.5) {$b$};
	\node[state, rectangle] (c) at (6, 0) {$c$};
	\node[state] (d) at (8, 0) {$d$};
	\node[state] (e) at (4, -1.5) {$e$};
	\node[state, rectangle] (f) at (2, -1.5) {$f$};
	\node[state] (g) at (0, -1.5) {$g$};
	
	\node[red] (l) at (1, 0.8) {$(\lambda_1)$};
	\node[red] (a') at (2, 0.8) {$1$};
	\node[red] (b') at (4, 0.8) {$1$};
	\node[red] (c') at (6, -0.7) {$0$};
	\node[red] (d') at (8, -0.7) {$0$};
	\node[red] (e') at (4, -2.2) {$0$};
	\node[red] (f') at (2, -2.2) {$0$};
	\node[red] (g') at (0, -2.2) {$3$};
	
	\path (a) edge[loop left] node[left] {$\stackrel{\playcircle}{1} \stackrel{\Box}{3}$} (a);
	\path (c) edge[loop left] node[left] {$\stackrel{\playcircle}{0} \stackrel{\Box}{0}$} (c);
	\path (d) edge[loop above] node[above] {$\stackrel{\playcircle}{0} \stackrel{\Box}{0}$} (d);
	\path (g) edge[loop left] node[left] {$\stackrel{\playcircle}{3} \stackrel{\Box}{2}$} (g);
	\path[<->] (a) edge node[above] {$\stackrel{\playcircle}{0} \stackrel{\Box}{0}$} (b);
	\path[<->] (b) edge node[above right] {$\stackrel{\playcircle}{2} \stackrel{\Box}{3}$} (c);
	\path[<->] (c) edge node[above] {$\stackrel{\playcircle}{1} \stackrel{\Box}{3}$} (d);
	\path[<->] (c) edge node[below right] {$\stackrel{\playcircle}{0} \stackrel{\Box}{0}$} (e);
	\path[<->] (e) edge node[above] {$\stackrel{\playcircle}{0} \stackrel{\Box}{0}$} (f);
	\path[<->] (f) edge node[above] {$\stackrel{\playcircle}{0} \stackrel{\Box}{0}$} (g);
	\end{tikzpicture}
\end{center}

At the first step, the requirement $\lambda_1$ captures the antagonistic values.

Then, from the state $c$, if player $\Box$ forces the access to the state $b$, then player $\Circle$ must get at least $1$: the worst play that can be proposed to player $\Box$ is then $(babc)^\omega$, which gives player $\Box$ the payoff $\frac{3}{2}$.

From the state $f$, if player $\Box$ forces the access to the state $g$, then the worst play that can be proposed to them is $g^\omega$.

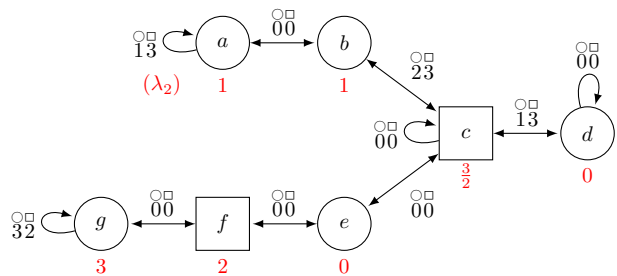
\begin{center}
	\begin{tikzpicture}[->,>=latex,shorten >=1pt, scale=0.8, every node/.style={scale=0.8}]
	\node[state] (a) at (2, 1.5) {$a$};
	\node[state] (b) at (4, 1.5) {$b$};
	\node[state, rectangle] (c) at (6, 0) {$c$};
	\node[state] (d) at (8, 0) {$d$};
	\node[state] (e) at (4, -1.5) {$e$};
	\node[state, rectangle] (f) at (2, -1.5) {$f$};
	\node[state] (g) at (0, -1.5) {$g$};
	
	\node[red] (l) at (1, 0.8) {$(\lambda_2)$};
	\node[red] (a') at (2, 0.8) {$1$};
	\node[red] (b') at (4, 0.8) {$1$};
	\node[red] (c') at (6, -0.7) {$\frac{3}{2}$};
	\node[red] (d') at (8, -0.7) {$0$};
	\node[red] (e') at (4, -2.2) {$0$};
	\node[red] (f') at (2, -2.2) {$2$};
	\node[red] (g') at (0, -2.2) {$3$};
	
	\path (a) edge[loop left] node[left] {$\stackrel{\playcircle}{1} \stackrel{\Box}{3}$} (a);
	\path (c) edge[loop left] node[left] {$\stackrel{\playcircle}{0} \stackrel{\Box}{0}$} (c);
	\path (d) edge[loop above] node[above] {$\stackrel{\playcircle}{0} \stackrel{\Box}{0}$} (d);
	\path (g) edge[loop left] node[left] {$\stackrel{\playcircle}{3} \stackrel{\Box}{2}$} (g);
	\path[<->] (a) edge node[above] {$\stackrel{\playcircle}{0} \stackrel{\Box}{0}$} (b);
	\path[<->] (b) edge node[above right] {$\stackrel{\playcircle}{2} \stackrel{\Box}{3}$} (c);
	\path[<->] (c) edge node[above] {$\stackrel{\playcircle}{1} \stackrel{\Box}{3}$} (d);
	\path[<->] (c) edge node[below right] {$\stackrel{\playcircle}{0} \stackrel{\Box}{0}$} (e);
	\path[<->] (e) edge node[above] {$\stackrel{\playcircle}{0} \stackrel{\Box}{0}$} (f);
	\path[<->] (f) edge node[above] {$\stackrel{\playcircle}{0} \stackrel{\Box}{0}$} (g);
	\end{tikzpicture}
\end{center}

Then, from the state $d$, if player $\Circle$ forces the access to the state $c$, then player $\Box$ must get at least $\frac{3}{2}$: the worst play that can be proposed to player $\Circle$ is then $(cccd)^\omega$, which gives player $\Circle$ the payoff $\frac{1}{2}$.

At the same time, from the state $e$, player $\Circle$ can now force the acces to the state $f$: then, the worst play that can be proposed to them is $fg^\omega$.

\begin{center}
	\begin{tikzpicture}[->,>=latex,shorten >=1pt, scale=0.8, every node/.style={scale=0.8}]
	\node[state] (a) at (2, 1.5) {$a$};
	\node[state] (b) at (4, 1.5) {$b$};
	\node[state, rectangle] (c) at (6, 0) {$c$};
	\node[state] (d) at (8, 0) {$d$};
	\node[state] (e) at (4, -1.5) {$e$};
	\node[state, rectangle] (f) at (2, -1.5) {$f$};
	\node[state] (g) at (0, -1.5) {$g$};
	
	\node[red] (l) at (1, 0.8) {$(\lambda_3)$};
	\node[red] (a') at (2, 0.8) {$1$};
	\node[red] (b') at (4, 0.8) {$1$};
	\node[red] (c') at (6, -0.7) {$\frac{3}{2}$};
	\node[red] (d') at (8, -0.7) {$\frac{1}{2}$};
	\node[red] (e') at (4, -2.2) {$3$};
	\node[red] (f') at (2, -2.2) {$2$};
	\node[red] (g') at (0, -2.2) {$3$};
	
	\path (a) edge[loop left] node[left] {$\stackrel{\playcircle}{1} \stackrel{\Box}{3}$} (a);
	\path (c) edge[loop left] node[left] {$\stackrel{\playcircle}{0} \stackrel{\Box}{0}$} (c);
	\path (d) edge[loop above] node[above] {$\stackrel{\playcircle}{0} \stackrel{\Box}{0}$} (d);
	\path (g) edge[loop left] node[left] {$\stackrel{\playcircle}{3} \stackrel{\Box}{2}$} (g);
	\path[<->] (a) edge node[above] {$\stackrel{\playcircle}{0} \stackrel{\Box}{0}$} (b);
	\path[<->] (b) edge node[above right] {$\stackrel{\playcircle}{2} \stackrel{\Box}{3}$} (c);
	\path[<->] (c) edge node[above] {$\stackrel{\playcircle}{1} \stackrel{\Box}{3}$} (d);
	\path[<->] (c) edge node[below right] {$\stackrel{\playcircle}{0} \stackrel{\Box}{0}$} (e);
	\path[<->] (e) edge node[above] {$\stackrel{\playcircle}{0} \stackrel{\Box}{0}$} (f);
	\path[<->] (f) edge node[above] {$\stackrel{\playcircle}{0} \stackrel{\Box}{0}$} (g);
	\end{tikzpicture}
\end{center}

But then, from the state $c$, player $\Box$ can now force the access to the state $e$: then, the worst play that can be proposed to them is $efg^\omega$.

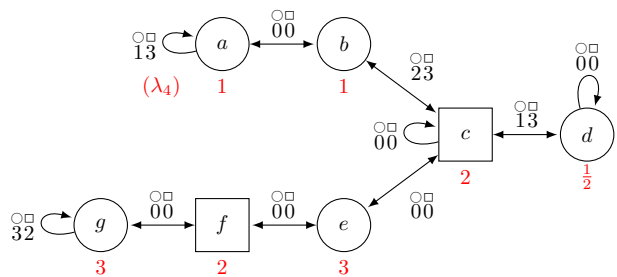
\begin{center}
	\begin{tikzpicture}[->,>=latex,shorten >=1pt, scale=0.8, every node/.style={scale=0.8}]
	\node[state] (a) at (2, 1.5) {$a$};
	\node[state] (b) at (4, 1.5) {$b$};
	\node[state, rectangle] (c) at (6, 0) {$c$};
	\node[state] (d) at (8, 0) {$d$};
	\node[state] (e) at (4, -1.5) {$e$};
	\node[state, rectangle] (f) at (2, -1.5) {$f$};
	\node[state] (g) at (0, -1.5) {$g$};
	
	\node[red] (l) at (1, 0.8) {$(\lambda_4)$};
	\node[red] (a') at (2, 0.8) {$1$};
	\node[red] (b') at (4, 0.8) {$1$};
	\node[red] (c') at (6, -0.7) {$2$};
	\node[red] (d') at (8, -0.7) {$\frac{1}{2}$};
	\node[red] (e') at (4, -2.2) {$3$};
	\node[red] (f') at (2, -2.2) {$2$};
	\node[red] (g') at (0, -2.2) {$3$};
	
	\path (a) edge[loop left] node[left] {$\stackrel{\playcircle}{1} \stackrel{\Box}{3}$} (a);
	\path (c) edge[loop left] node[left] {$\stackrel{\playcircle}{0} \stackrel{\Box}{0}$} (c);
	\path (d) edge[loop above] node[above] {$\stackrel{\playcircle}{0} \stackrel{\Box}{0}$} (d);
	\path (g) edge[loop left] node[left] {$\stackrel{\playcircle}{3} \stackrel{\Box}{2}$} (g);
	\path[<->] (a) edge node[above] {$\stackrel{\playcircle}{0} \stackrel{\Box}{0}$} (b);
	\path[<->] (b) edge node[above right] {$\stackrel{\playcircle}{2} \stackrel{\Box}{3}$} (c);
	\path[<->] (c) edge node[above] {$\stackrel{\playcircle}{1} \stackrel{\Box}{3}$} (d);
	\path[<->] (c) edge node[below right] {$\stackrel{\playcircle}{0} \stackrel{\Box}{0}$} (e);
	\path[<->] (e) edge node[above] {$\stackrel{\playcircle}{0} \stackrel{\Box}{0}$} (f);
	\path[<->] (f) edge node[above] {$\stackrel{\playcircle}{0} \stackrel{\Box}{0}$} (g);
	\end{tikzpicture}
\end{center}

And finally, from that point, if from the state $d$ player $\Circle$ forces the access to the state $c$, then player $\Box$ must have at least the payof $2$; and therefore, the worst play that can be proposed to player $\Circle$ is now $(ccd)^\omega$, which gives her the payoff $\frac{2}{3}$.

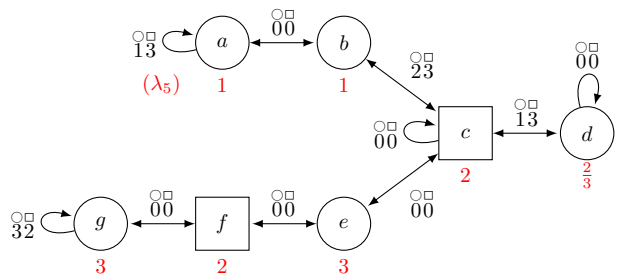
\begin{center}
	\begin{tikzpicture}[->,>=latex,shorten >=1pt, scale=0.8, every node/.style={scale=0.8}]
	\node[state] (a) at (2, 1.5) {$a$};
	\node[state] (b) at (4, 1.5) {$b$};
	\node[state, rectangle] (c) at (6, 0) {$c$};
	\node[state] (d) at (8, 0) {$d$};
	\node[state] (e) at (4, -1.5) {$e$};
	\node[state, rectangle] (f) at (2, -1.5) {$f$};
	\node[state] (g) at (0, -1.5) {$g$};
	
	\node[red] (l) at (1, 0.8) {$(\lambda_5)$};
	\node[red] (a') at (2, 0.8) {$1$};
	\node[red] (b') at (4, 0.8) {$1$};
	\node[red] (c') at (6, -0.7) {$2$};
	\node[red] (d') at (8, -0.7) {$\frac{2}{3}$};
	\node[red] (e') at (4, -2.2) {$3$};
	\node[red] (f') at (2, -2.2) {$2$};
	\node[red] (g') at (0, -2.2) {$3$};
	
	\path (a) edge[loop left] node[left] {$\stackrel{\playcircle}{1} \stackrel{\Box}{3}$} (a);
	\path (c) edge[loop left] node[left] {$\stackrel{\playcircle}{0} \stackrel{\Box}{0}$} (c);
	\path (d) edge[loop above] node[above] {$\stackrel{\playcircle}{0} \stackrel{\Box}{0}$} (d);
	\path (g) edge[loop left] node[left] {$\stackrel{\playcircle}{3} \stackrel{\Box}{2}$} (g);
	\path[<->] (a) edge node[above] {$\stackrel{\playcircle}{0} \stackrel{\Box}{0}$} (b);
	\path[<->] (b) edge node[above right] {$\stackrel{\playcircle}{2} \stackrel{\Box}{3}$} (c);
	\path[<->] (c) edge node[above] {$\stackrel{\playcircle}{1} \stackrel{\Box}{3}$} (d);
	\path[<->] (c) edge node[below right] {$\stackrel{\playcircle}{0} \stackrel{\Box}{0}$} (e);
	\path[<->] (e) edge node[above] {$\stackrel{\playcircle}{0} \stackrel{\Box}{0}$} (f);
	\path[<->] (f) edge node[above] {$\stackrel{\playcircle}{0} \stackrel{\Box}{0}$} (g);
	\end{tikzpicture}
\end{center}

The requirement $\lambda_5$ is a fixed point of the negotiation function.
\end{ex}
\end{document}